\documentclass[a4paper, 11pt]{article}
\usepackage{amssymb,amsmath,amsthm,lmodern,cancel}
\usepackage{amsfonts}

\usepackage[activate={true,nocompatibility},final,tracking=true,kerning=true,spacing=true,factor=1100,stretch=15,shrink=15]{microtype}
\microtypecontext{spacing=nonfrench}

\usepackage{graphicx}
\usepackage[english]{babel}
\usepackage[T1]{fontenc}
\usepackage{textcomp}
\usepackage[dvipsnames]{xcolor}
\usepackage[shortlabels]{enumitem}
\usepackage{thmtools, mathtools,verbatim} \usepackage{upgreek}

\usepackage[colorlinks=false,hyperindex=true]{hyperref}
\usepackage{setspace,wrapfig}

\usepackage[font=normal]{caption}
\usepackage{floatrow}

\usepackage{thmtools}
\usepackage{thm-restate}

\title{Large deviations of multichordal $\SLE_{0+}$, 
real rational functions, 
and zeta-regularized determinants of Laplacians}  
\date{\today}

\author{Eveliina Peltola\thanks{Institute for Applied Mathematics, University of Bonn, Germany. \protect\url{eveliina.peltola@hcm.uni-bonn.de}} \thanks{Department of Mathematics and Systems Analysis, Aalto University, Finland.} \,  and Yilin Wang\thanks{ Institut des Hautes \'Etudes Scientifiques, Bures-sur-Yvette, France. \protect\url{yilin@ihes.fr}}}

\setlist[enumerate]{topsep = 1ex, leftmargin=1cm, itemsep= -2pt}

\let\OLDthebibliography\thebibliography
\renewcommand\thebibliography[1]{
  \OLDthebibliography{#1}
  \setlength{\parskip}{1pt}
  \setlength{\itemsep}{2pt}
}

\allowdisplaybreaks

\newtheorem{thm}{Theorem}[section]
\newtheorem{cor}[thm]{Corollary}

\newtheorem{lem}[thm]{Lemma}
\newtheorem{prop}[thm]{Proposition}

\newtheorem{theorem}{Theorem}

\newtheorem{lemA}[theorem]{Lemma}

\theoremstyle{definition} 
\newtheorem{df}[thm]{Definition}

\newtheorem{remark}[thm]{Remark}

\numberwithin{equation}{section}

\global\long\def\bR{\mathbb{R}}

\global\long\def\bN{\mathbb{N}}

\global\long\def\bR{\mathbb{R}}

\global\long\def\tr{\mathrm{Tr}\,}
\global\long\def\ud{\mathrm{d}}

\global\long\def\domain{D}

\global\long\def\open{O}
\global\long\def\closed{F}

\global\long\def\bpt{x} 
\global\long\def\ept{y} 

\global\long\def\np{n} 

\global\long\def\link#1#2{\{#1,#2\}}

\global\long\def\Hmin{\mc M}

\global\long\def\PartF{\mc Z}

\global\long\def\ii{\mathfrak{i}}

\global\long\def\rfmultichord{\rho}
\global\long\def\Catalan{C}

\global\long\def\cst{\lambda}

\global\long\def\confmap{\vartheta}

\global\long\def\bigdisc{\ad D}

\renewcommand{\liminf}{\varliminf}
\renewcommand{\limsup}{\varlimsup}

\global\long\def\ARBC{\beta}

\newcommand{\abs}[1]{\left\lvert #1 \right \rvert}

\newcommand{\norm}[1]{\lVert #1 \rVert}

\newcommand{\mc}[1]{\mathcal{#1}}
\newcommand{\m}[1]{\mathbb{#1}}

\renewcommand\Re{\operatorname{Re}}
\renewcommand\Im{\operatorname{Im}}

\def\PSL{\operatorname{PSL}}
\def\SLE{\operatorname{SLE}}

\def\tr{\operatorname{Tr}}

\def\Rat{\operatorname{Rat}}
\def\Poly{\operatorname{Poly}}

\def\a{\alpha}
\def\b{\beta}
\def\g{\gamma}
\def\G{\Gamma}
\def\d{\delta}
\def\D{\Delta}

\def\t{\theta}

\def\l{\lambda}
\def\k{\kappa}

\def\s{\sigma}

\def\o{\omega}

\def\vare{\varepsilon}

\def\Chat{\hat{\m{C}}}

\def\dd{\,\mathrm{d}}
\def\vol{\mathrm{vol}}

\newcommand{\ad}[1]{\overline{#1}}

\def\detz{\mathrm{det}_{\zeta}}

\def\1{\mathbf{1}}

\begin{document}
\maketitle

\begin{abstract}
We prove a strong large deviation principle (LDP) for multiple chordal $\SLE_{0+}$ 
curves with respect to the Hausdorff metric. 
In the single-chord
case, this result strengthens an earlier 
partial result by the second author. 
We also introduce a Loewner potential, which in the smooth case has a simple expression in terms of zeta-regularized determinants of Laplacians. 
This potential differs from the LDP rate function by an additive constant depending only on 
the boundary data, that satisfies PDEs
arising as a semiclassical limit of the Belavin-Polyakov-Zamolodchikov equations of level two in conformal field theory with central charge $c \to -\infty$.

Furthermore, we show that every multichord minimizing the potential
in the upper half-plane 
for given boundary data is the real locus of a rational function and is unique, thus coinciding
with the $\k \to 0+$ limit of the multiple $\SLE_\k$.
As a by-product, we provide an analytic proof of the Shapiro conjecture in real enumerative geometry, first proved by Eremenko and Gabrielov: 
if all critical points of a rational function are real, then the function is real up to post-composition by a M\"obius transformation.
\bigskip

\noindent \textbf{Keywords:} Schramm-Loewner evolution (SLE), large deviations, enumeration of real rational functions, determinants of Laplacians, BPZ partial differential equations, 
semiclassical limit of conformal field theory

\bigskip

\noindent \emph{Mathematics Subject Classification} (2010): Primary: 30C55; Secondary: 60J67 60F10 
 
\end{abstract}

\newpage

\tableofcontents

\newpage

\section{Introduction}

The Schramm-Loewner evolution (SLE) is a model for random
conformally invariant fractal curves in the plane,  introduced by Schramm 
by combining Loewner's classical theory for evolution of planar slit domains 
with stochastic analysis~\cite{Schramm2000}. 
Schramm's SLE is a one-parameter family of probability measures on non-self-crossing curves, indexed by $\k \ge 0$, and thus denoted by $\SLE_\k$.
 SLEs play a central role in 2D random conformal geometry. For instance, they describe interfaces in conformally invariant systems arising from statistical physics, which was also Schramm's original motivation, see, e.g.,~\cite{LSW04LERWUST,Smi:ICM,Schramm:ICM,SS09GFF}.
Through their relationship with critical statistical physics models, 
SLEs are also closely related to conformal field theory (CFT), 
see, e.g., 
\cite{BB:CFTSLE, Car03,Friedrich_Werner_03,FK_CFT,DRC_SET,Kontsevich_SLE, Dub_SLEVir1, Dub_SLEVir2, Peltola}.
The parameter $\k$ reflects the roughness of these fractal curves,
and it also determines
the central charge $c(\k) = (3\k-8)(6-\k) / 2\k$ of the associated CFT. Renormalization group arguments~\cite{Cardy_RG} 
suggest that $\k$ encodes the universality classes of the models.

In this article, we consider the chordal case, where $\SLE$s are families of random curves (multichords) connecting pairwise distinct boundary points of some planar domain.
Throughout, we let $\domain$ be a simply connected Jordan domain of the Riemann sphere $\Chat = \m C \cup \{\infty\}$. 
We include the marked distinct boundary points to the domain data $(\domain; x_1, \ldots, x_{2\np})$, assuming that 
they appear in counterclockwise order along the boundary $\partial \domain$. 
We also assume that $\partial \domain$ is smooth in a neighborhood of the marked points, unless stated otherwise\footnote{Some of these regularity assumptions could be relaxed, but we do not gain anything by doing so.}. 
Due to the planarity, there exist $\Catalan_\np$ different possible pairwise non-crossing connections for the curves, where
\begin{align} \label{eq:catalan}
\Catalan_\np = \frac{1}{\np+1} \binom{2\np}{\np}
\end{align}
is the $\np$:th Catalan number.
We enumerate them in terms of \emph{$\np$-link patterns} 
\begin{align} \label{eq: alpha}
\a = \{ \link{a_1}{b_1}, \link{a_2}{b_2}, \ldots, \link{a_\np}{b_\np} \} ,
\end{align}
that is, planar pair partitions of $\{1,2,\ldots,2\np\}$ giving a pairing of the marked points.

We investigate the asymptotic behavior of multichordal $\SLE_\k$ as $\k \to 0+$.
The case of a single $\SLE_{0+}$ was studied in \cite{W1} (see also
\cite[Sec.\,9.3]{Dub_comm}), 
where the second author introduced the Loewner energy (for a single chord) and studied 
large deviations of left-right passing events. 
We strengthen and generalize this result, establishing
a strong large deviation principle for chordal and 
multichordal $\SLE_{0+}$ with respect to the Hausdorff metric. 
The rate function is given by the multichordal Loewner energy, which we introduce shortly.

We also introduce a Loewner potential, which differs from the Loewner energy by 
a certain function (the minimal potential) only depending on the boundary data $(x_1, \ldots, x_{2\np};\a)$.
We show that the Loewner potential has a simple expression in terms of zeta-regularized determinants of Laplacians 
--- similar relations between SLEs and determinants of Laplacians have been observed in, e.g., \cite{FK_CFT,Kontsevich_SLE,Dub_couplings,Dub_SLEVir1,W2}. 
In particular, the minimal potential is directly linked to SLE partition functions studied in~\cite{BB03, BBK, Dub_Euler, Dub_comm, KL07, Law09, KP_PPF, PW19} 
and correlation functions and conformal blocks in boundary CFT with $c \leq 1$
(see, e.g.,~\cite{Dub_SLEVir1,KKP_conformal,Peltola} for details and 
discussion on further literature). 
Indeed, we show that the minimal potential can be seen as 
the semiclassical $c \to -\infty$ limit of certain CFT correlation functions.

Rather surprisingly, we classify
all minimizers of the potential 
by showing that, for the upper half-plane
$\m H$ with marked points $x_1 < \cdots < x_{2\np}$, any minimizer must be the real locus of a rational function.
Applying an algebraic geometry result by Goldberg~\cite{Gol91}, we show that for each $2\np$-tuple $x_1 < \cdots < x_{2\np}$, there are exactly the $\np$\textnormal{:}th 
Catalan number $\Catalan_\np$ of different real loci of rational functions, each corresponding to a different $\np$-link pattern $\a$ for the multichord.
Our large deviation result then implies that the multichordal $\SLE_\k$ measure, given $\a$, converges to the minimizer of the multichordal Loewner potential. 
We thereby establish a previously unknown link between SLEs
and real enumerative geometry, and as a by-product of this, we obtain a new proof for the so-called Shapiro conjecture~\cite{Sot00,EG02}.

\bigskip

Next, we summarize our main findings and discuss the organization of the article.
For the readers' convenience, we first briefly recall chordal SLE and multichordal SLE. 
Then, we discuss the main results of the present article 
divided to four sub-topics.

\subsection*{Reminder: Schramm-Loewner evolutions}

By a \emph{chord}, we refer to a simple curve connecting two distinct boundary points $\bpt,\ept \in \partial \domain$ 
in $\domain$ touching 
$\partial \domain$ only at its endpoints.
Since we are interested in the asymptotic behavior of $\SLE_\k$ as $\k \to 0+$, we shall assume throughout 
that $\k \in [0,8/3)$, 
so that $c(\k) < 0$ and chordal $\SLE_\k$ in $(\domain; \bpt, \ept)$ is
a random chord in $\domain$ from $\bpt$ to $\ept$~\cite{Rohde_Schramm}.
Furthermore, by its reversibility property~\cite{Zhan}, 
the laws of $\SLE_\k$ in $(\domain;\bpt, \ept)$ and in $(\domain;\ept, \bpt)$ are the same as unparametrized  
curves, so the order of the endpoints is
insignificant.

\begin{wrapfigure}{R}{0.35\textwidth}
\centering
\includegraphics[width=0.9\textwidth]{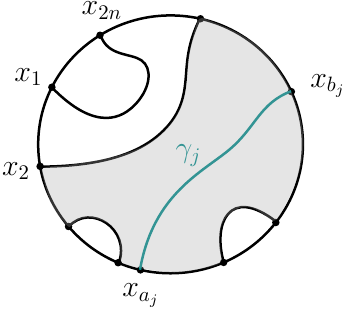}
\captionsetup{width=0.8\textwidth}
\caption{Illustration of a multichord and the domain $\hat \domain_j$ containing $\g_j$.
\label{fig:NSLE} }
\end{wrapfigure}

In fact, the law of chordal $\SLE_\k$ can be uniquely characterized 
by its \emph{conformal invariance} and \emph{domain Markov property}.
The former asserts that, for any conformal map (i.e., biholomorphic function) $\varphi \colon \domain \to \domain'$,
the law of $\SLE_\k$ in $(\domain'; \varphi(\bpt), \varphi(\ept))$ is just the pushforward  law of $\SLE_\k$ in $(\domain; \bpt, \ept)$ by $\varphi$.
The latter Markovian property for the growth of the SLE curve is very natural from the point of view of statistical mechanics systems.
For details, we refer to
the textbooks~\cite{WW_St_Flour,Law05,Kemppainen_book,Beliaevs_book}, 
which provide detailed introductions to SLEs and applications, all from slightly different perspectives.

In this article, we denote by $\mc X(\m H; 0, \infty)$ 
the family of unparametrized chords in $(\m H; 0, \infty)$ with infinite total half-plane capacity,  
and by $\mc X(\domain;\bpt, \ept)$ the image of $\mc X(\m H; 0, \infty)$ under a conformal map sending $(\m H; 0, \infty)$ to $(\domain;\bpt, \ept)$.
For each $\np \geq 1$, domain data $(\domain; x_1, \ldots, x_{2\np})$, and $\np$-link pattern $\a$ as in~\eqref{eq: alpha}, 
we let $\mc X_{\a}(\domain;  x_1, \ldots,  x_{2\np})$ 
denote the set of \emph{multichords} $\ad \g = (\g_1, \ldots, \g_\np)$
consisting of pairwise disjoint chords where $\g_j \in \mc X(\domain; x_{a_j}, x_{b_j})$  
for each $j\in\{1,\ldots,\np\}$. 
\emph{Multichordal $\SLE_\k$} is a probability measure 
on $\mc X_\a(\domain;x_1, \ldots,  x_{2\np})$. 
It can be defined by the property that, for each $j$, given the other curves $\{\g_i \; | \; i \neq j \}$, 
the conditional law of the random curve $\g_j$ is the single chordal $\SLE_\k$ in 
the connected component $\hat \domain_j$ of $\domain \smallsetminus \bigcup_{i \neq  j} \g_i$
containing $\g_j$, highlighted in grey in Figure~\ref{fig:NSLE}.

Constructions for multiple chordal SLEs have been obtained 
in many works 
\cite{Car03, WW_toulouse_Girsanov, BBK, Dub_comm, KL07, Law09, IG1, IG2, PW19, BPW}. 
The above definition of multichordal $\SLE_\k$
appeared explicitly in,  e.g.,~\cite{BPW}, where the authors
proved that multichordal $\SLE_\k$ is
the unique stationary measure of a Markov chain 
on $\mc X_{\a}(\domain;x_1,\ldots,x_{2\np})$ 
defined by re-sampling the curves from their conditional laws. However, this idea was already introduced and used earlier in~\cite{IG1, IG2}, 
where Miller \&~Sheffield studied interacting SLE curves 
coupled with the Gaussian free field (GFF) in the framework of ``imaginary geometry''.
See also the recent~\cite[App.\,~A]{MSW}.

Concretely, the multichordal $\SLE_\k$ measures 
can be constructed from independent 
chordal $\SLE_\k$ measures by weighting by a suitable Radon-Nikodym derivative. 
However, when $\k = 0$, the measures must be singular to each other, so this method does not apply.

\subsection{Real rational functions from geodesic multichords}
\label{subsec:Real_rational_functions_from_geodesic_multichords}

Multichordal $\SLE_0$ in $\mc X_\a(\domain; x_1, \ldots, x_{2\np})$ 
is a deterministic multichord $ \ad \eta = (\eta_1, \ldots, \eta_\np)$ with the property that 
each $\eta_j$ is $\SLE_0$ in its own component $(\hat \domain_j; x_{a_j}, x_{b_j})$. 
In other words, each $\eta_j$ is the \emph{hyperbolic geodesic} in $(\hat \domain_j; x_{a_j}, x_{b_j})$, i.e., the image of $\ii \m R_+$ by a uniformizing conformal map $\varphi \colon \m H \to \hat \domain_j$ such that $\varphi(0) = x_{a_j}$ and $\varphi(\infty) = x_{b_j}$.

We call a multichord with this property a \emph{geodesic multichord}. 
 As far as we are aware, the existence  of such objects  for $\np \geq 3$ 
 has not been explicitly stated, and their uniqueness has been unknown.
Combining  an algebraic geometry result of Goldberg~\cite{Gol91} with analytic techniques,
we obtain:

\begin{restatable}{thm}{mainuniqueminimizing}
\label{thm:main_unique_minimizing}
 There exists a unique geodesic multichord in $\mc X_\a (\domain; x_1, \ldots, x_{2\np})$ for each~$\a$. 
\end{restatable}

The existence of geodesic multichords for each $\a$ follows by characterizing them as minimizers of a lower semicontinuous Loewner energy (or potential), 
to be discussed shortly (see Corollary~\ref{cor:main_limit}). 
The uniqueness is a consequence of the following algebraic result. 
(Note that, by the conformal invariance of the geodesic property, we may assume that $\domain = \m H$.)

A \emph{rational function} is an analytic branched covering $h$ of $\Chat$ over $\Chat$ (or equivalently, a ratio of two polynomials). The  \emph{real locus} of $h$ is the preimage
of $\hat {\m R} : =  \m R \cup\{\infty\}$ by $h$.
A real rational function is a ratio of two real polynomials.
We discuss rational functions in Section~\ref{subsec:rational}, where we give a constructive proof for the next theorem
(see Proposition~\ref{prop:gamma_h}).

\begin{restatable}{thm}{realrational}
\label{thm:real_rational}
Let $\bar{\eta}$ be a geodesic multichord  in $\mc X_\a (\m H; x_1, \ldots, x_{2\np})$.
The union of $\bar{\eta}$, its complex conjugate, and the real line is the real locus of a real rational function of degree $\np+1$ with critical points $\{x_1, \ldots, x_{2\np}\}$.
\end{restatable}

By the Riemann-Hurwitz formula, a rational function of degree $\np+1$ has $2\np$ critical points 
if and only if all of its critical points are of index two 
(namely, the function is locally a two-to-one branched cover near the critical points).  
To complete the proof of Theorem~\ref{thm:main_unique_minimizing}, it thus suffices to classify such rational functions with a given set of $2\np$ critical points. 
Goldberg showed in~\cite{Gol91} that there are \emph{at most} 
$\Catalan_\np$ of them, up to post-composition by a M\"obius transformation of $\Chat$. Since $\Catalan_\np$ is also the number of $\np$-link patterns, 
this yields the uniqueness in Theorem~\ref{thm:main_unique_minimizing}.

In 1995, B.~Shapiro and M.~Shapiro made a remarkable conjecture  concerning real solutions to enumerative geometric problems on the Grassmannian, see~\cite{Sot00}. 
In~\cite{EG02},  Eremenko and Gabrielov proved this conjecture 
for the Grassmannian of $2$-planes in $n$-spaces, 
in which case
the conjecture is equivalent to the following statement (Corollary~\ref{cor:Catalan}).
We obtain this result easily from Theorem~\ref{thm:main_unique_minimizing} when the critical points are all of index two.
(The case of degenerating critical points follows similarly as in~\cite{EG02}, and we do not attempt to provide a new proof.) 

\begin{restatable}{cor}{corCatalan}
\label{cor:Catalan}
If all critical points of a rational function are real, then it is a real rational function up to post-composition by a M\"obius transformation of $\Chat$.
\end{restatable}
Conversely, the proof in \cite{EG02} also implies 
Theorem~\ref{thm:main_unique_minimizing} via Theorem~\ref{thm:real_rational}. 

\subsection{Large deviations of SLEs}
\label{subsec:Large_deviations_of_SLEs}

Chordal $\SLE_\k$ in $(\m H; 0 ,\infty)$ can be constructed as a random Loewner evolution with 
driving function $W := \sqrt \k B$, where $B$ is the standard one-dimensional Brownian motion (see Section~\ref{subsec:Loewner_chain}). 
In this description, $\SLE_\k$ is a dynamically growing random curve. 

In general, Loewner's theorem asserts that any simple curve $\g$ in $(\m H; 0 ,\infty)$ can be encoded in a Loewner evolution with some real-valued continuous driving function $t \mapsto W_t$ with $W_0 = 0$. 
It can then be pulled back to any domain $(\domain;\bpt,\ept)$ via a uniformizing conformal map $\varphi \colon \domain \to \m H$ sending $\bpt$ and $\ept$ respectively to $0$ and $\infty$.
The conformally invariant \emph{Loewner energy} of a chord $\g$ in $(\domain;\bpt,\ept)$ was introduced in~\cite{W1} as 
\begin{align} \label{eq_LE_intro} 
I_{\domain;\bpt,\ept} (\g) := I_{\m H; 0 ,\infty}(\varphi (\gamma)) : = \frac{1}{2}\int_{0}^{\infty} \abs{\frac{\ud W_t}{\ud t}}^2 \dd t,
\end{align}
where $W$ is the Loewner driving function of $\varphi(\g)$ 
and the right-hand side is the Dirichlet energy of $W$.
The Loewner energy of a chord is not always finite, 
but a driving function $W$ of finite Dirichlet energy such that $W_0=0$ 
always generates a chord in $(\m H; 0,\infty)$.
Notice also that $\g$ has zero energy in $(\domain;\bpt,\ept)$ if and only if $W \equiv 0$ 
(which generates $\SLE_0$), 
or equivalently, $\g$ is the hyperbolic geodesic in $(\domain;\bpt,\ept)$.

The above definition of the Loewner energy was motivated by 
large deviations of chordal $\SLE_\k$ as $\k \to 0+$. 
In~\cite{W1}, a weak large deviation result following from Schilder's theorem for Brownian paths was proved: loosely speaking, we have
\begin{align*}
``\; \m P [\text{SLE}_{\k} \text{ in $(\domain;\bpt,\ept)$ stays close to } \gamma] \quad
\overset{\text{$\k \to 0+$}}{\text{\scalebox{1.5}[1]{$\approx$}}}
\quad
\exp \Big(- \frac{I_{\domain;\bpt,\ept}(\gamma)}{\k} \, \Big) \;".
\end{align*}
This result was stated with a rather complicated notion of ``staying close''  in terms of events of passing to the right or left of a collection of points. 
It was sufficient to deduce the reversibility of the Loewner energy based on that of $\SLE_\k$, namely, the property that  switching the endpoints of the chord does not matter, 
$I_{\domain;\bpt,\ept} = I_{\domain;\ept,\bpt}$. 
We thus consider the chords as being unoriented and sometimes omit the endpoints $\bpt, \ept$ from the notation.

In the present article, we strengthen the above result to a (strong) large deviation principle (LDP) 
for $\SLE_{0+}$ with respect to the Hausdorff metric. 
This is a special case of the general LDP for multichordal $\SLE_{0+}$.
To state it, for each link pattern $\a$ as in~\eqref{eq: alpha},
we endow the curve space $\mc X_{\a}(\domain;  x_1, \ldots,  x_{2\np}) \subset \prod_j \mc X(\domain; x_{a_j}, x_{b_j})$ with the product topology induced by the Hausdorff metric in the closed unit disc $\ad {\m D}$ via 
a uniformizing\footnote{The uniformizing map is not unique, and different maps do not induce the same Hausdorff metric. The induced topologies are still equivalent, 
as conformal automorphisms of $\m D$ are uniformly continuous on $\ad {\m D}$.}  conformal map $\m D \to \domain$ (see Section~\ref{subsec:topologies}).

To describe the rate function, we first introduce the \emph{Loewner potential} 
\begin{align} \label{eq:initialdef_Hn}
\mc H_\domain(\ad \g) := \; & \frac{1}{12} \sum_{j=1}^\np I_{\domain}(\g_j) 
+ m_\domain(\ad \g) - \frac{1}{4} \sum_{j=1}^\np \log P_{\domain; x_{a_j}, x_{b_j}} ,
\end{align}
where the factor $m_\domain(\ad \g)$ is expressed in terms of
the so-called Brownian loop measure introduced by Lawler, Schramm \&~Werner \cite{LSW_CR_chordal,LW2004loupsoup} 
(see Equation~\eqref{eq: m_alpha})
and $P_{\domain;\bpt, \ept}$ is the Poisson excursion kernel
(see Sections~\ref{subsec:BLM} and~\ref{subsec:a_loop_measure}).
Intuitively, the factor $m_\domain(\ad \g)$ represents interaction of the chords, and the other terms in~\eqref{eq:initialdef_Hn} constitute their individual potentials.
Alternatively, 
in Section~\ref{subsec:Zetaregularized_determinants_of_Laplacians} below we will write the potential 
(for smooth chords)
in terms of zeta-regularized determinants of Laplacians (Equation~\eqref{eq:main_det_H}).

Taking infimum over all multichords $\ad \g \in \mc X_\a(\domain;x_1, \ldots,x_{2\np})$, we denote the minimal Loewner potential by
\begin{align} \label{eq:initialdef_Hmin}
\Hmin^\a_\domain (x_1, \ldots,x_{2\np})
:= \inf_{\ad \g} \mc H_\domain (\ad \g).
\end{align}
Then, we define the \emph{multichordal Loewner energy}
\begin{align}\label{eq:def_I_multi}
I_\domain^\a (\ad \g) : = I_{\domain; x_1, \ldots,x_{2\np}}^\a (\ad \g) 
:= 12 ( \mc H_\domain (\ad \g) - \Hmin^\a_\domain (x_1, \ldots,x_{2\np}) ).
\end{align}

\begin{remark}
The multichordal Loewner energy is a non-negative functional on multichords with \emph{fixed} boundary
data $(x_1, \ldots, x_{2\np};\a)$. However, the Loewner potential, from which we can deduce the Loewner energy, is more general since it takes into account different boundary data and can take negative values. The minimal potential~\eqref{eq:initialdef_Hmin} is also a meaningful function of the boundary data, as we will see in Section~\ref{subsec:Semiclassical_limit_of_CFT_nullstate_equations}.
\end{remark}

\begin{restatable}{thm}{mainLDP}
\label{thm:main_LDP}
The family of laws $(\m P_\a^\k)_{\k > 0}$ of the multichordal $\SLE_\k$ curves $\ad \g^\k$
satisfies the following LDP in $\mc X_\a (\domain;x_1,\ldots,x_{2\np})$ with good rate function $I^\a_\domain$\textnormal{:}

For any closed subset $\closed$ and open subset $\open$ of $\mc X_\a(\domain;x_1,\ldots,x_{2\np})$, we have
\begin{align*}
\limsup_{\k \to 0+} \k \log  \m P_\a^\k [ \ad \g^\k \in \closed ] 
 \; \leq \; & - \inf_{\ad \g \in \closed}  I_{\domain}^\a(\ad \g) ,   \\
\liminf_{\k \to 0+} \k \log  \m P_\a^\k [ \ad \g^\k \in \open ] 
 \; \geq \; & - \inf_{\ad \g \in \open}  I_{\domain}^\a(\ad \g) , 
\end{align*}
$I_\domain^\a$ is lower semicontinuous, and 
its level set $(I_\domain^\a)^{-1} [0,c]$ is compact for any $c \geq 0$.
\end{restatable}

Throughout the present work, ``\emph{an LDP}'' shall refer to a large deviation principle of the above type. The reader should mind carefully the topology involved, which can be a subtle point. 
Also, by a ``\emph{good rate function}'' we refer to 
a rate function whose level sets are compact
(which also implies its lower semicontinuity).

To show the compactness of the level sets in Theorem~\ref{thm:main_LDP}, we generalize the result \cite[Prop.\,2.1]{W1} for the single-chord Loewner energy,
by proving that any finite-energy multichord $\ad \g$ is the image of a smooth reference multichord in $\mc X_\a (\domain;x_1,\ldots,x_{2\np})$ under a $K$-quasiconformal 
self-map of $\domain$, where the constant $K$ only depends on the energy (see Proposition~\ref{prop: H_semicontinuous}). 
From this, we obtain the required precompactness in order to find convergent subsequences, which also shows that the infimum in~\eqref{eq:initialdef_Hmin} is achieved.

To establish the LDP, we first prove it for the
single-chord 
$\SLE_{0+}$ (Section~\ref{subsec:LDP_single_whole}) via Schilder's theorem as in \cite{W1}\footnote{However, we consider general closed and open sets as opposed to 
just left-right passing events.}. Note that the contraction principle does not apply naively, 
since the Loewner transform is not continuous from the space of driving functions to the curve space endowed with the Hausdorff metric, 
and Schilder's theorem applies for Brownian paths over a finite time interval. 
We deal with these technical subtleties by identifying the discontinuity points  of the Loewner transform (Lemma~\ref{lem:closure}) 
and by introducing an appropriate truncation
(Sections~\ref{subsec:LDP_single_finite_T}--\ref{subsec:LDP_single_whole}). 
The general LDP for multichordal $\SLE_{0+}$ is then obtained via its Radon-Nikodym derivative\footnote{When $c (\k) > 0$, 
one has to multiply the exponent by the indicator function $\1\{\ad \g \text{ is disjoint}\}$. 
As we are anyway concerned with the limit $\k \to 0+$, 
we shall restrict our discussion to $\k < 8/3$ throughout.}
\begin{align} \label{eq:multi_SLE_RN}
\frac{1}{Z_\a } \exp\left(\frac{1}{2}c(\k) \, m_\domain(\ad \g) \right) 
\end{align} 
with respect to the product measure
of $\np$ independent $\SLE_\k$ chords with the same boundary data (Section~\ref{subsec: LDP multiple}), 
where $Z_\a = Z_\a(\domain; x_1,\ldots, x_{2\np})$ is the normalizing factor making this a probability measure.
The passage from the single-chord case to the multichord case is 
an application of Varadhan's lemma  (Lemma~\ref{lem:Varadhan}),
and the constant $12$ in the energy~\eqref{eq:def_I_multi} emerges from the asymptotic behavior of the central charge: $c(\k) / 2 \sim - 12 / \k$ as $\k \to 0+$.

We also show that any energy-minimizing multichord is a geodesic multichord (Corollary~\ref{cor:min_H_geodesic}).
The uniqueness of this geodesic multichord (Theorem~\ref{thm:main_unique_minimizing}) thus implies:

\begin{restatable}{cor}{mainlimitcor} 
\label{cor:main_limit}
As $\k \to 0+$, multichordal $\SLE_\k$ 
in $(\domain;x_1, \ldots, x_{2\np})$ associated to the $\np$-link pattern $\a$ converges in probability to 
the unique 
minimizer of the potential
in $\mc X_\a(\domain;x_1, \ldots, x_{2\np})$. 
This minimizer is the unique geodesic multichord associated to $\a$.
\end{restatable}

\subsection{Semiclassical limit of CFT null-state equations}
\label{subsec:Semiclassical_limit_of_CFT_nullstate_equations}

To define the energy $I_\domain^{\a}$, one could have added
to the potential $\mc H_\domain$ an arbitrary constant that depends only on the boundary data $(x_1, \ldots, x_{2\np};\a)$. 
However, with the precise form~\eqref{eq:initialdef_Hn} of our definition,
not only the Loewner potential $\mc H_\domain$ captures the rate function of multichordal $\SLE_{0+}$, but it also relates multichords of different boundary data. 
This becomes interesting when $n \ge 2$, as the moduli space of the boundary data is non-trivial.

For instance, $\mc H_\domain$ encodes a system of Loewner differential equations which generates 
the geodesic multichord by varying the marked points (Proposition~\ref{prop:main_Loewner}).
Moreover, it relates to the more general framework of defining a global notion of Loewner energy
(cf. Theorem~\ref{thm:main_H_det} and Equation~\eqref{eq:main_H_general}), to be discussed shortly.
Lastly, the minimal Loewner potential $\Hmin_\domain^\a$ satisfies PDEs
arising as a semiclassical limit of the Belavin-Polyakov-Zamolodchikov (BPZ) equations of level two in CFT with central charge $c \to -\infty$ (Proposition~\ref{prop:deterministic_PDEs}).

We state below the results for $\domain = \m H$ and 
$x_1 < \cdots < x_{2\np}$.
We also fix the $\np$-link pattern $\a$ and denote $\mc U :=12 \Hmin_{\m H}^\a$.
We pick one chord $\eta_j$ going from $x_{a_j}$ to $x_{b_j}$
(by reversibility, the role of $x_{a_j}$ to $x_{b_j}$ could be interchanged),
and we consider the Loewner flow generated by $\eta_j$
(parametrized by capacity) while the starting points of the other chords are evolving according to the Loewner equation.

\begin{restatable}{prop}{mainLoewner} 
\label{prop:main_Loewner} 
Let $\ad \eta$ be the minimizer of $\mc H_{\m H}$ in $\mc X_\a (\m H; x_1, \ldots, x_{2\np})$.
Then, for each $j \in \{1,\ldots,\np\}$, the Loewner driving function $W$
of the chord $\eta_j$ from $x_{a_j}$ to $x_{b_j}$ 
and the time evolutions $V_t^i = g_t(x_i)$ of the other marked points satisfy the differential equations
\begin{align} \label{eqn:DE_for_drivers}
\begin{split}
\begin{dcases}
\frac{\ud W_t}{\ud t}  =  - \partial_{a_j} \mc U (V_t^1, \ldots, V_t^{a_j-1}, W_t, V_t^{a_j+1}, \ldots, V_t^{2\np} )  , 
\qquad W_0 = x_{a_j} , \\
\frac{\ud  V_t^i}{\ud t} = \dfrac{2}{V_t^i-W_t}, \qquad V_0^i = x_i, 
\quad \textnormal{for } i\neq a_j, 
\end{dcases}
\end{split}
\end{align}
for $0 \leq t < T$, where $T$ is the lifetime of the solution
and $(g_t)_{t \in [0,T]}$ is the Loewner flow generated by $\eta_j$.
\end{restatable}

\begin{restatable}{prop}{deterministicPDEs} 
\label{prop:deterministic_PDEs}
For each $j \in \{1,\ldots, 2\np\}$, we have
\begin{align} \label{eq:deterministic_PDEs}
\frac{1}{2} ( \partial_j \mc U (x_1, \ldots,  x_{2\np}) )^2 - \sum_{i \neq j} 
\frac{2}{x_i - x_j} \partial_i \mc U (x_1, \ldots,  x_{2\np})
 = \sum_{ i \neq j}  \frac{6}{(x_i - x_j)^2}.
\end{align}
\end{restatable}

We prove Proposition~\ref{prop:main_Loewner} in Section~\ref{subsec:char_minimizer}
and Proposition~\ref{prop:deterministic_PDEs} in Section~\ref{subsec:null-field_PDE}. 
Both proofs are based on direct analysis of the minimal potential $\Hmin_{\m H}^\a$ 
and the associated self-similar Loewner flow, using the fact that the geodesic property of a multichord is preserved under the Loewner flow generated by any of its chords.

\bigskip

We now relate these results to SLE partition functions (see, e.g.,~\cite[Def.\,3.1]{Dub_couplings}). 
For each $\np$-link pattern $\a$, one associates to the multichordal $\SLE_\k$ 
a (\emph{pure}) \emph{partition function} 
$\PartF_\a$ defined in terms of the total mass of the multichordal $\SLE_\k$ measure in~\eqref{eq:multi_SLE_RN}, i.e., 
\begin{align*}
\PartF_\a(\m H;x_1, \ldots,x_{2\np}) := 
\Big( \prod_{j=1}^\np P_{\m H;x_{a_j},x_{b_j}} \Big)^{(6-\k)/2\k} \times
Z_\a(\m H;x_1, \ldots,x_{2\np}) .
\end{align*}
As a consequence of the  proof of Theorem~\ref{thm:main_LDP}, we obtain 
that the minimal potential can be viewed as a semiclassical limit of $\PartF_\a$ (Corollary~\ref{cor:semicl_lim_of_pf}), in the sense that
\begin{align}\label{eq:intro_limit_Z}
\k \log \PartF_\a(\m H;x_1, \ldots,x_{2\np}) 
\quad \overset{\k \to 0}{\longrightarrow}  \quad
- 12 \, \Hmin_{\m H}^\a (x_1, \ldots,x_{2\np}) = - \mc U (x_1, \ldots,x_{2\np}).
\end{align}
Equation~\eqref{eqn:DE_for_drivers} is analogous to
the marginal law of one chord in the multichordal $\SLE_\k$, given by 
a stochastic Loewner equation
derived from $\mc Z_\a$
(explicitly, see~\cite[Eq.\,(4.10)]{PW19}).

It is well known that
the partition functions $\PartF_\a$ can also be seen as correlation functions associated to conformal fields  with degeneracy at level two in CFTs of central charge $c(\k) \in (-\infty, 1]$ 
(see, e.g.,~\cite{Dub_SLEVir1, Peltola} and references therein for more discussion on this). 
In particular, the correlators $\PartF_\a$ labeled by all $\np$-link patterns $\a$ 
form a basis for a solution space of the following system of null-state BPZ PDEs (predicted in~\cite{BPZ}):
\begin{align} \label{eq: multichordal SLE PDEs}
\bigg( \frac{\k}{2}\partial_j^2 + \sum_{i \neq j}
 \bigg(\frac{2}{x_{i}-x_{j}}\partial_i - 
\frac{(6-\k)/\k}{(x_{i}-x_{j})^{2}} \bigg) \bigg) 
\PartF(x_1,\ldots,x_{2\np}) =  0 , 
\end{align}
for all $j \in \{1,\ldots,2\np\}$. 
Here, the number $(6-\k)/\k$ is two times the conformal weight $h_{1,2}(\k)$ of a conformal primary field 
associated to a particular representation of the Virasoro algebra (labeled by the so-called Kac weights $h_{r,s}(\k)$); 
correlation functions $\PartF(x_1,\ldots,x_{2\np})$ involving fields of weight $h_{1,2}(\k)$
satisfy null-state PDEs of the form~\eqref{eq: multichordal SLE PDEs}. 

Heuristically,
plugging the expression 
$\PartF = \exp(-\mc U / \k)$ into the PDE system~\eqref{eq: multichordal SLE PDEs} 
gives exactly 
the asymptotic form~\eqref{eq:deterministic_PDEs} as $\k \to 0$. 
In the physics literature, this appears as a 
system of Hamilton-Jacobi type equations 
(also related to  Painlev\'e VI)
associated to 
the semiclassical conformal blocks, see~\cite{4_authors_2013} and references therein, as also
briefly pointed out in~\cite[Sec.\,4.5]{BBK}. 
A semiclassical limit of the dual version of these PDEs (associated to the dual fields of conformal weight $h_{2,1}$ in the Liouville CFT of central charge $c \geq 25$) 
were used to  
give a probabilistic proof for
the Takhtajan-Zograf theorem relating Poincar\'e's accessory parameters to the classical Liouville action, see~\cite[Eq.\,(1.11)~and~Sec.\,4.4]{LRV}.
Let us also remark that in the recent preprint~\cite{alberts2020pole}, the evolution of the poles and critical points of the rational function in our Theorem~\ref{thm:real_rational} was investigated, and related to a particular Calogero-Moser integrable system. A closed form expression of the minimal potential in terms of the rational function is also given in \cite[Thm.\,2.3]{alberts2020pole}.

\subsection{Zeta-regularized determinants of Laplacians}
\label{subsec:Zetaregularized_determinants_of_Laplacians}

Our definition of the Loewner potential is not only natural from the point of view of SLE partition functions and CFT, 
but also guided by the more general quest of defining a general, 
\emph{global} notion for the Loewner energy and potential for curves on Riemannian surfaces.
Indeed, the local growth definition~\eqref{eq_LE_intro} of the energy via the Loewner evolution
is only well-adapted to the case of a single chord, or a rooted Jordan curve (i.e., \emph{Loewner loop energy} introduced in \cite{RW}),
and does not explain the presence of many symmetries such as the reversibility and root-invariance.
In contrast, these symmetries are apparent from the global definitions, which can also be naturally extended to other scenarios: e.g.,
\begin{itemize}[itemsep=-3pt]
    \item multichords as in the present work,
    \item radial, or multi-radial curves growing from the boundary towards interior points,
    \item collections of loops embedded in a (closed) surface of higher genus,
    \item chords joining two boundary components of an annulus, etc.
\end{itemize}

As a first step towards such a general notion, we establish
a relation between the Loewner potential $\mc H$ and zeta-regularized determinants of Laplacians (denoted by $\detz \D$). 
We use the sign convention so that the Laplacian
$\D_{\domain;g}$ on $\domain$ 
with Dirichlet boundary conditions 
is a positive operator:
\begin{align}\label{eq:Laplacian}
\D_{\domain;g} : = - \frac{1}{\sqrt{\det(g)}} \sum_{i,j= 1}^2 \partial_i \sqrt{\det (g)}g^{ij} \partial_j ,
\end{align}
where $g$ is a Riemannian metric on $\domain$.
If the metric is not mentioned, $g$ is taken by default to be the Euclidean metric $\ud z^2$, in which case we write $\D_{\domain;\dd z^2} = \D_\domain =  -\partial_1^2 - \partial_2^2$. 
Throughout, we only consider the Laplacian with Dirichlet boundary conditions.
 Also, we only\footnote{In fact, one could also allow $\partial \domain$ to have finitely many corners, cf.~Definition~\ref{def:curvilinear}.} consider a bounded smooth  $\domain$, for then the Laplacians appearing in the following theorem have a discrete spectrum, and their determinants are well-defined (see Section~\ref{subsec:Lap_det}). 

\begin{restatable}{thm}{mainHdet} 
\label{thm:main_H_det}
For any smooth multichord $\ad \g$ in a bounded smooth domain $\domain$, we have 
\begin{align}\label{eq:main_det_H}
\mc{H}_{\domain} (\ad \g) = \log\detz \D_{\domain} - 
\sum_C \log\detz \D_{C} - n\cst ,
\end{align}
where the sum is taken over all connected components $C$  of $\domain \smallsetminus \bigcup_{i} \g_i$ and $\cst = \frac{1}{2}\log \pi \approx 0.5724$.
\end{restatable}

In Theorem~\ref{thm:main_H_det}, we consider the Euclidean metric to simplify notation ---  the result can, however, be easily generalized to arbitrary smooth Riemannian metrics and the associated Laplace-Beltrami operators 
(in which case $\mc H_\domain$ will depend on the metric, and $\domain$ can be unbounded if the metric tensor decays fast enough).

Theorem~\ref{thm:main_H_det} provides another reason for why it is natural to include Poisson kernels in the definition~\eqref{eq:initialdef_Hn} of $\mc H$. Since $I_\domain$ and $m_\domain$ are conformally invariant quantities, 
only the Poisson kernel terms contribute to the 
change of $\mc H_\domain$ under conformal maps $\varphi \colon \domain \to \domain'$\textnormal{:}
\begin{align} \label{eq:conformal_covariance}
\mc H_\domain (\g_1, \ldots, \g_\np) = \mc H_{\domain'} (\varphi (\g_1), \ldots, \varphi(\g_\np))- \frac{1}{4} \sum_{j=1}^{2\np} \log |\varphi'(x_j)|. 
\end{align}
The expression~\eqref{eq:main_det_H} captures this conformal covariance more intrinsically.
In fact, finite-energy multichords always meet the boundary $\partial \domain$ perpendicularly (see Lemma~\ref{lem:perp}), and $\log \detz \D$ vary 
under a conformal change of metric according to the Polyakov-Alvarez anomaly formula~\cite{R_inprep} for domains with piecewise smooth boundary, allowing corners (see
Theorem~\ref{thm:PA}).
Importantly, the corners have a non-trivial global contribution to the variation of $\log \detz \D$, which coincides with the contribution from the Poisson kernels. 
This fact is also instrumental in our proof of Theorem~\ref{thm:main_H_det} presented in Section~\ref{subsec:Det_identity_proof}.

Generally speaking, consider a pair $(\ad \g, \mc S)$ where $\ad \g$ is a collection of simple curves and $\mc S$ is 
an ambient Riemannian surface endowed with a metric $g$.
In light of our results, the following definition (up to additive constants) 
for the Loewner potential seems to be the most natural one to consider:
\begin{align}\label{eq:main_H_general}
\mc H(\ad \g, \mc S) := \log\detz \D_{\mc S; g} - \sum_{C } \log\detz \D_{C; g},
\end{align}
where the sum is taken over all connected components $C$ of $\mc S \smallsetminus \ad \g$.

For the Loewner loop energy, the identity~\eqref{eq:main_H_general} 
was proved in \cite{W2},
where the loop energy was further connected  to the K\"ahler potential on the Weil-Petersson Teichm\"uller space, also of independent interest. 
Our proposed definition~\eqref{eq:main_H_general} agrees with our Theorems~\ref{thm:main_LDP} and~\ref{thm:main_H_det}.
Furthermore, when $\g$ is a simple loop embedded in a Riemann surface $\mc S$ 
(in this case, $\mc H(\g, \mc S)$ is conformally invariant, so we can view $\mc S$ as a Riemann surface),
Equation~\eqref{eq:main_H_general} is consistent with the axiomatization of 
the Malliavin-Kontsevich-Suhov loop measure~\cite{Kontsevich_SLE} and Dub\'edat's SLE/GFF partition function couplings~\cite{Dub_couplings}. 
In turn, the $\lambda$-SAW introduced by Kozdron \&~Lawler in~\cite{KL07} is closely related to these ideas on the discrete lattice model side.

Finally, let us point out that $-\log \detz \D_\domain$ can be formally identified with the total mass of Brownian loop measure in $\domain$, as was first observed in \cite{LeJan2006det,Dub_couplings}. 
Theorem~\ref{thm:main_H_det} then suggests that $- \mc H_\domain (\ad \g)$ equals the total mass of Brownian loops touching the multichord $\ad \g$.  
However, computing naively, Brownian loop measure has an infinite total mass coming from the small loops, 
so a renormalization procedure is needed 
(as for determinants of Laplacians).
The recent
\cite{ang2020brownian} shows how one can make sense of the identity by cutting out the small loops
(i.e., introducing a UV-cutoff).  Using the same idea, we obtain:

\begin{restatable}{thm}{mainUV}
Let $\ad \g$ be a smooth finite-energy  multichord in $\domain$.
The total mass of loops in $\domain$ touching $\ad \g$ 
with quadratic variation greater than $4\d$ under Brownian loop measure  has the expansion
\begin{align*}
\frac{ \mathrm l(\ad \g)}{2\sqrt{\pi \d}} -  \mc H_\domain (\ad \g) - \np \cst  + \frac{\np}{4} (\log \d + \boldsymbol{\upgamma}) + O(\d^{1/2} \log \d), \qquad  \text{as } \; \d \to 0+,
\end{align*}
where $\mathrm l (\ad \g)$ is the total arclength of $\ad \g$, $\cst$ is the universal constant from Theorem~\ref{thm:main_H_det}, and  $\boldsymbol{\upgamma} \approx 0.5772$ is the Euler-Mascheroni constant.
\end{restatable} 

The asserted expansion follows rather directly from the short-time expansion of the heat trace (Theorem~\ref{thm: heat trace multpilied}) and the Mellin transform. 
We give the proof in Section~\ref{subsec:UV}.

\bigskip

\noindent {\bf This article is structured as follows:}
In Section~\ref{sec:prelim}, we recall definitions and collect necessary preliminary results. 
Then, in Section~\ref{sec_energy_def_and_properties} we define the multichordal Loewner potential and energy, and derive their first properties. 
Sections~\ref{sec:min},~\ref{sec:LDPS}, and~\ref{sec:det} correspond respectively to our results presented above in Sections~\ref{subsec:Real_rational_functions_from_geodesic_multichords}~\&~\ref{subsec:Semiclassical_limit_of_CFT_nullstate_equations}, 
\ref{subsec:Large_deviations_of_SLEs}, 
and~\ref{subsec:Zetaregularized_determinants_of_Laplacians}. 
These three sections can be read independently. 
We also include 
two appendices: Appendix~\ref{app:estimate} concerns a technical $\SLE$ estimate needed for Theorem~\ref{thm:main_LDP}, and in Appendix~\ref{app:PA} we discuss the Polyakov-Alvarez formula for curvilinear domains, used for Theorem~\ref{thm:main_H_det}.

\noindent {\bf Conventions:}
All probability spaces  are assumed to be completed, and we only consider completed Borel $\s$-algebras.
We label all known results by letters, and results derived in the present work by numbers.

\subsection*{Acknowledgments}
We would like to thank 
the Mathematisches Forschungsinstitut Oberwolfach 
(MFO) 
for supporting us through the Research in Pairs program, 
where this project was initiated.
We also thank 
Dapeng Zhan and an anonymous referee for pointing out a gap in the application 
of~\cite[Thm.\,1]{FieldLawler} in an earlier draft, and Greg Lawler for helpful correspondence on filling this gap 
in Appendix~\ref{app:estimate}.
We thank Osama Abuzaid, 
Alexandre Eremenko, Ellen Krusell,
Julie Rowlett, 
Vincent Vargas, Fredrik Viklund, and Wendelin Werner
for
useful comments and discussions, 
and the referee for very careful reading of our manuscript and many constructive comments. 
Y.W. thanks Wendelin Werner for the suggestion to look into rational functions, and is funded by the NSF grant DMS-1953945.
E.P. is funded by the Deutsche Forschungsgemeinschaft (DFG, German Research Foundation) under Germany's Excellence Strategy - GZ 2047/1, Projekt-ID 390685813, and by the Academy of Finland grant number 340461 ``Conformal invariance in planar random geometry''.

\section{Preliminaries}
\label{sec:prelim}

This section is devoted to fixing notation and stating preliminary results.
Throughout this article, $\domain$ will denote a simply connected Jordan domain of 
the Riemann sphere $\Chat = \m C \cup \{\infty\}$
and $\bpt, \ept \in \partial \domain$ or $x_1, x_2, \ldots,x_{2\np} \in \partial \domain$ distinct boundary points on smooth boundary segments ordered counterclockwise.
Common reference domains are the \emph{unit disc}
$\m D := \{ z \in \m C \; | \; |z| < 1 \}$
and the \emph{upper half-plane} 
$\m H := \{ z \in \m C \; | \; \Im z > 0 \}$.

\subsection{Hausdorff and Carath\'eodory topologies}
\label{subsec:topologies}

To begin, we discuss two different topologies considered in this work. 

\begin{df} \label{defn: Hausdorff}
The \emph{Hausdorff distance} $d_h$ of two compact sets $\closed_1, \closed_2 \subset \ad{\m D}$ is defined as
\begin{align*}
d_h(\closed_1, \closed_2) := \inf \Big\{ \vare \geq 0 \; \Big| \; \closed_1 \subset \bigcup_{x \in \closed_2} \ad{B}_\vare(x) \; \text{ and } \; \closed_2 \subset \bigcup_{x \in \closed_1} \ad{B}_\vare(x) \Big\} ,
\end{align*}
where $B_\vare(x)$ denotes the Euclidean ball around $x \in \ad{\m D}$ of radius $\vare$. 
Using this, we define the Hausdorff distance for 
closed subsets of
$\domain$ via the pullback by a conformal map $\domain \to \m D$. 
Although the Hausdorff distance
depends on the choice of this conformal map, the topology induced by $d_h$ is canonical, as (continuous extensions of) conformal automorphisms of $\m D$ are uniformly continuous on $\ad {\m D}$. 
We endow the curve spaces $\mc X(\domain;\bpt;\ept)$ (resp.~$\mc X_\a ( \domain; x_1, \ldots, x_{2\np})$) with the relative topology (resp.~product topology) induced by $d_h$. 
\end{df}

When $\domain = \m H$, we fix throughout this article a conformal map $\confmap \colon  \m H \to \m D$ 
and we let $\mc C$ denote the space of non-empty closed subsets of $\ad{\m H}$ endowed with the Hausdorff distance defined by pullback by $\confmap$. 
In particular, $\mc C$ is a compact metric space.
We say that $K \in \mc C$ is \emph{bounded} if it is bounded for the Euclidean metric in $\ad{\m H}$, i.e.,
there exists $R >0$ such that $|z| \le R$ for all $z \in K$.

We also consider the set of \emph{$\ad{\m H}$-hulls} (or simply, \emph{hulls}), denoted by 
\begin{align*}
\mc K := \{ K \in \mc C \; | \; K \text{ is bounded, } \m H \smallsetminus K \text{ is simply connected, and } \ad{K \cap \m H} = K \}  ,
\end{align*}
where $\ad{K \cap \m H}$ is the closure of $K \cap \m H$ in $\m C$.
For each hull $K \in \mc K$, there exists a unique uniformizing conformal map 
$g_K \colon \m H \smallsetminus K \to \m H$, 
referred to as the \emph{mapping-out function of} $K$, such that 
$g_K(z) - z \to 0$ as $|z| \to \infty$. 
Via Schwarz reflection, $g_K$ always extends to a conformal map on the open set $\Chat \smallsetminus (K\cup K^*)$, where $K^*$ is the complex conjugate of $K$. 
In particular, $g_K$ is also well-defined on $\m R \smallsetminus K$.

It has the expansion
\begin{align}\label{eq:g_expansion}
g_K(z) = z + \frac{\mathrm{hcap}(K)}{z} + O(|z|^{-2}) , \qquad |z| \to \infty ,
\end{align}
where the non-negative constant $\mathrm{hcap}(K) \geq 0$ is called the \emph{(half-plane) capacity} of the hull~$K$.
The capacity is an increasing function: for all $K, \tilde K \in \mc K$,
\begin{align}\label{eq:monotonicity_capacity}
    K \subset \tilde K \quad \implies \quad \mathrm{hcap} (K) \leq \mathrm{hcap} (\tilde K). 
\end{align}
For fixed $T \in (0,\infty)$, 
we set 
\begin{align*}
\mc{K}_T := \{ K \in \mc K \; | \; \mathrm{hcap}(K_T) = 2T \}. 
\end{align*}

The following topology on $\mc K$ plays well with Loewner theory (see, e.g.,~\cite[Sec.\,3.3]{Beliaevs_book} for a detailed discussion on the Carath\'eodory convergence, and~\cite[Sec.\,3.6]{Law05} and~\cite[Sec.\,6]{Kemppainen_book} for its applications to Loewner theory).

\begin{df}\label{defn:Cara}
We say that a sequence $(K^k)_{k\in \bN}$  of hulls converges in the \emph{Carath\'eodory topology} to $K \in \mc K$ if the functions
$(g_{K^k}^{-1})_{k \ge 0}$ converge uniformly away from\footnote{That is, for any $\vare > 0$, we have $g_{K^k}^{-1} \to g_K^{-1}$ uniformly on $\{ z \in \m C \; | \; \Im z \geq \vare \}$, cf.~\cite[Sec.\,3.6]{Law05}.} $\m R$ 
to $g_K^{-1}$.
Geometrically, this is equivalent to the \emph{Carath\'eodory kernel convergence}
of the complementary domains $\domain^k: = \m H \smallsetminus K^k$ to $\domain : = \m H \smallsetminus K$ with respect to $\infty$\textnormal{:} 
for any subsequence $(\domain^{k_j})_{j \in \m N}$,  
we have $\domain = \bigcup_{j \ge 1} V_j$, where $V_j$ is the unbounded connected component of the interior of $\bigcap_{i \ge j} \domain^{k_i}$. 
\end{df}

Note that the Hausdorff convergence and Carath\'eodory convergence of hulls are different, and none implies the other.
For example, the arcs $K^k = \{e^{i \theta} \; | \; \theta \in [ 1/k, \pi] \}$ converge to 
the half-disc $K = \ad{\m D} \cap \ad{\m H}$ in the Carath\'eodory topology, but to the semi-circle $\partial \m D \cap \ad{\m H}$
for the Hausdorff metric.

\begin{lem} \label{lem:useful_for_bdry_lemma}
If a sequence  $(K^k)_{k \in \m N}$ of hulls converges to $K \in \mc{K}$ in the Carath\'eodory topology and 
it also converges
to a bounded set $\tilde K \in \mc C$
for the Hausdorff metric,
then 
$\m H \smallsetminus K$ coincides with the unbounded connected component of $\m H \smallsetminus \tilde{K}$.
In particular, we have 
$\tilde{K} \cap \m H = K \cap \m H$
if and only if $\m H \smallsetminus \tilde{K}$ is connected. 

\end{lem}

\begin{proof} 
Set $D^k: = \m H \smallsetminus K^k$ and for $j \geq 1$, let $V_j$ (resp.~$\tilde D$) be the unbounded connected component of the interior of $\bigcap_{k \ge j} D^k$ (resp.~$\m H \smallsetminus \tilde{K}$). 
It suffices to show that $\tilde D  = \bigcup_{j \ge 1} V_j$.

By definition,
any $z \in V_j$ has a neighborhood contained in all $D^k$ for $k \ge j$. 
On the one hand, 
as $K^k$ converge in the Hausdorff metric to $\tilde K$, we see that $z \notin \tilde K$, 
so $V_j \subset \m H \smallsetminus \tilde K$. Hence, as $V_j$ are connected and unbounded, we deduce that $\bigcup_{j \ge 1} V_j \subset \tilde D$.
On the other hand, for $z \in \tilde D$, let $\G_z$ be a path connecting $z$ to $\infty$ at a positive distance from $\tilde K$. The Hausdorff convergence of $K^k$ implies that $\G_z$ has a neighborhood which belongs to only finitely many $K^k$. 
This implies that $z \in \bigcup_{j \ge 1} V_j$, so $\tilde D \subset \bigcup_{j \ge 1} V_j$, finishing the proof. 
\end{proof}

\subsection{Loewner chains, driving functions, and SLE}
\label{subsec:Loewner_chain}

We denote by $C^0[0,\infty)$ the space of continuous, real-valued functions $t \mapsto W_t$ such that $W_0 = 0$.
A \emph{chordal Loewner chain} driven by
$W \in C^0[0,\infty)$ is a family of conformal maps $(g_t)_{t \geq 0}$
associated with a collection $(K_t)_{t \geq 0}$ of growing hulls 
obtained by solving the Loewner equation
\begin{align} \label{eqn:LE}
\partial_{t} g_t(z) = \frac{2}{g_t(z)-W_t}  
\qquad \text{with initial condition} \qquad  g_0(z)=z 
\end{align}
for each $z \in \ad{\m H}$.
The solution $t \mapsto g_t(z)$ to~\eqref{eqn:LE} is defined up to the swallowing time of $z$,
\begin{align*}
\tau(z) := \sup\{ t \geq 0 \; | \; \inf_{s\in[0,t]}|g_{s}(z)-W_{s}|>0\} , \qquad \tau(0) : = 0,
\end{align*}
and the increasing family of sets 
$K_t := \{z \in \ad{\m H} \; | \; \tau(z) \leq t\}$ are $\ad {\m H}$-hulls. In particular, they satisfy $K_t = \ad{K_t \cap \m H}$.
Moreover, $g_t = g_{K_t}$ are the mapping-out functions of $K_t$, and the hulls $(K_t)_{t \geq 0}$ are parametrized by capacity, i.e.,
$K_t \in \mc K_t$ for all $t \geq 0$.

Throughout, we endow the space $C^0 [0,T]$ of continuous functions $t \mapsto W_t$ on $[0,T]$ 
such that $W_0 = 0$ with the topology induced by the uniform norm $\norm W_\infty := \sup_{t \in [0,T]} |W_t|$, 
and $C^{0}[0,\infty)$ with the topology of uniform convergence on all compact subsets of $[0,\infty)$.

Note that, for fixed $T>0$, the hull $K_T$ generated by $W$ depends only on $W_{[0,T]}$, that is, $W$ restricted on $[0,T]$. 
Moreover, the \emph{Loewner transform} 
\begin{align*}
\mc{L}_T \colon C^0[0,T] \to \mc K_T, 
\qquad W_{[0,T]} \mapsto K_T ,
\end{align*}
is continuous when $\mc{K}_T$ is equipped with the Carath\'eodory topology,
see~\cite[Prop.\,6.1]{Kemppainen_book}.
However, this is not true for the Hausdorff metric. When dealing with the LDP of SLE curves with respect to the Hausdorff metric, we will use the following observation.

We say that $\closed$ (resp.~$\open$) is a \emph{Hausdorff-closed} (resp.~\emph{Hausdorff-open}) set in $\mc K_T$ if it is closed (resp.~open) for the relative topology on $\mc K_T$ induced by the Hausdorff metric.

\begin{lem} \label{lem:closure}
Let $\closed$ be a Hausdorff-closed set 
in $\mc{K}_T$.
If $W \in \ad{\mc{L}_T^{-1} (\closed)} \smallsetminus \mc{L}_T^{-1} (\closed)$,
then the hull $\mc L_T(W)$ has a non-empty interior.
\end{lem}

\begin{proof}
Take a sequence $(W^k)_{k \in \bN}$ of driving functions 
such that $W^k \in \mc{L}_T^{-1} (\closed)$
and $W^k$ converges to $W \notin \mc L_T^{-1} (F)$ in $C^0[0,T]$.
Then, 
the corresponding hulls 
$K^k := \mc L_T(W^k)$ converge to 
$K := \mc L_T(W) \in \mc{K}_T$ in the Carath\'eodory topology. 
Also, by the compactness of $\mc C$, there exists a subsequence $K^{k_j}$ that converges to some $\tilde K \in \ad F$ for the Hausdorff metric, where $\ad F$ is the closure of $F$ in $\mc C$.
We will first show that $\tilde K$ is bounded and then apply Lemma~\ref{lem:useful_for_bdry_lemma} to compare $\tilde K$ with the Carath\'eodory limit $K$ of $K^k$.
Indeed, by considering the real part and the imaginary part of the Loewner equation~\eqref{eqn:LE},
we see that for all $z \in K^k$, 
we have $|\Re (z)| \le \norm{W^k}_{\infty}$ and $\Im (z) \le 2\sqrt{T}$. 
Since the sequence $(W^k)_{k\in \m N}$ is uniformly bounded in $C^0[0,T]$, 
the sequence $(K^k)_{k\in \m N}$ is uniformly bounded in $\ad{\m H}$ for the Euclidean distance, which implies that $\tilde K$ is also bounded.

If  $K \cap \m H \neq \tilde K \cap \m H$, Lemma~\ref{lem:useful_for_bdry_lemma} shows that $\m H \smallsetminus \tilde K$ has several connected components, so $K$ has a non-empty interior. 
If $K \cap \m H = \tilde K \cap \m H$, we claim that $\tilde K = K$. Indeed, we have 
\begin{align} \label{eq: inclusion}
K = \ad{K\cap \m H} = \ad{\tilde K\cap \m H} \subset \tilde K.
\end{align}
Now let $x \in \m R \smallsetminus K$ and $\vare >0$ such that $\operatorname{dist}(x,K) > 2 \vare$. Since $W^k$ converges uniformly to $W$, the Loewner equation \eqref{eqn:LE} 
shows that for large enough $k$, $\operatorname {dist}(x, K^k) >\vare$. Hence, $x \notin \tilde K$ and $K = \tilde K$. From~\eqref{eq: inclusion}, we obtain $K = \tilde K  \in \ad F \cap \mc K_T = F$, which contradicts with the assumption that $W \notin \mc L_T^{-1} (F)$.
This finishes the proof.
\end{proof}

We say that a function $W \in C^0[0,\infty)$ \emph{generates a chord} $\g$ in $\mc X(\m H; 0, \infty)$
if the map $t\mapsto \g_t$  is a continuous injection from $[0,\infty)$ to $\ad {\m H}$ 
such that $\g_0 = 0$, 
$\g_{(0,\infty)} \subset \m H$, 
$|\g_t| \to \infty$ as $t \to \infty$, 
and the image  $\g_{[0,t]}$ equals $\mc L_t(W_{[0,t]})$ for all $t \geq 0$. 
We also say that $\g$ is 
\emph{capacity-parametrized} 
if $\g_{[0,t]} \in \mc K_t$ for all $t \geq 0$.  
Conversely, any chord $\g \in \mc X(\m H; 0, \infty)$ can be endowed with capacity parametrization and is generated by the driving function $W \in C^0 [0,\infty)$
given by the image of its tip $\g_t$ by the mapping-out function of $\g_{[0,t]}$, i.e.,
\begin{align}\label{eq:W_from_curve}
W_t := g_{\g_{[0,t]}} (\g_t) \qquad \text{for all } t \ge 0.
\end{align}

We define the \emph{Dirichlet energies} for 
$W \in C^0[0,\infty)$ as
\begin{align} \label{eq_DE}
I (W) := \frac{1}{2} \int_{0}^{\infty} \abs{\frac{\ud  W_t}{\ud t}}^2 \ud t \qquad  \text{ and } \qquad I_T(W) := \frac{1}{2} \int_{0}^{T} \abs{\frac{\ud  W_t}{\ud t}}^2 \ud t \quad \text{ for } T > 0,
\end{align}
 if $W$ is absolutely continuous, and set them to equal $\infty$ otherwise.
When $I(W) <\infty$,
the function $W$ always generates a chord in $(\m H; 0, \infty)$ by~\cite[Prop.\,2.1]{W1}.

\begin{df}  \label{def:SLE} 
$\SLE_{\k}$ in $(\m H;0,\infty)$ is the Loewner chain
driven by $W = \sqrt{\k}B$, where $B=(B_t)_{t \geq 0}$ is the standard Brownian motion and $\k \geq 0$ is the diffusivity parameter.
\end{df} 

In the present article, we are only interested in the behavior of $\SLE_\k$ when $\k \to 0+$. Thus, for convenience, we assume throughout that $\k < 8/3$, so the associated central charge  
$c(\k) = (3\k-8)(6-\k) / 2\k $ is negative, 
and $\SLE_\k$ is almost surely 
generated by
a simple chord\footnote{$\SLE_\k$ is generated by a simple chord whenever $\k \leq 4$, but the central charge is positive if $\k \in (8/3,4]$.}. 
The latter fact  was proved by Rohde \&~Schramm \cite{Rohde_Schramm}. 
Then, $\SLE_{\k}$ in $(\domain;\bpt,\ept)$ is defined 
as the pullback of $\SLE_{\k}$ in $(\m H;0,\infty)$ by any conformal map $\varphi$ from $\domain$ to $\m H$  sending $\bpt$ to $0$ and  $\ept$ to $\infty$.
We remark that another choice $\tilde \varphi$ of the conformal map differs  only 
by a scaling factor, i.e., $ \tilde \varphi = c \varphi$ for some $c >0$. 
However, the driving function transforms from $W$ to $t \mapsto c W_{c^{-2}t}$ under the scaling $z \mapsto cz$, so in particular, the law of Brownian motion is preserved.
Thus, $\SLE_\k$ in $(\domain;\bpt,\ept)$ is well-defined and conformally invariant.

\subsection{Chordal Loewner energy}\label{sec:chordal_single_I}

The \emph{Loewner energy} of a chord $\g \in \mc X(\domain; \bpt, \ept)$ is defined as the Dirichlet energy~\eqref{eq_DE} of its driving function,
\begin{align} \label{eq_LE}
I_\domain(\g) = I_{\domain;\bpt, \ept} (\g)
:= I_{\m H; 0 ,\infty}(\varphi (\g)) 
: = I (W),
\end{align}
where $\varphi$ is any conformal map from $D$ to $\m H$ 
such that $\varphi(\bpt) = 0$ and $\varphi(\ept) = \infty$,
and $W$ is the driving function of $\varphi (\g)$. 
Note that the definition of $I_\domain (\g)$ does not depend on the choice of $\varphi$, since 
$I(t\mapsto W_t) = I(t\mapsto c W_{c^{-2} t})$ as we discussed after Definition~\ref{def:SLE}.

The Loewner energy $I_{\domain;\bpt, \ept} (\g)$ is non-negative and it is minimized by the 
\emph{hyperbolic geodesic} $\eta = \eta_{\domain; \bpt, \ept}$, that is, the preimage of $\ii \m R_+$ 
under $\varphi$. 
Indeed, the driving function of $\varphi(\eta)$ is just the constant function $W \equiv 0$, so $I_\domain(\eta) = 0$.

The Loewner energy was introduced in~\cite{W1}, where some of its basic properties were also studied. 
 A key property of finite-energy chords is that they are images of analytic curves by quasiconformal maps.
 We briefly recall  the results needed in the present work, and refer the readers to~\cite{W1} for more details and to, e.g.,~\cite[Ch.\,1]{lehto2012univalent} for basics on quasiconformal maps.

\begin{lem} \label{lem:finite_energy_quasiconformal} 
If $I_{\domain; \bpt, \ept}(\g) <\infty$, then there exists a constant $K \in [1,\infty)$ depending only on $I_{\domain; \bpt, \ept}(\g)$ 
such that $\g$ is the image of $\eta$ by a $K$-quasiconformal map $\varphi$ with $\varphi(\domain) = \domain$, and $\varphi$
extends continuously to $\ad {\domain}$ and equals the identity function on $\partial \domain$.
\end{lem}

\begin{proof}
We prove the assertion first for $(\m H; 0, \infty)$. 
According to~\cite[Prop.\,2.1]{W1}, there exists a $K$-quasiconformal map $\tilde\varphi \colon \m H \to \m H$ 
such that $\tilde\varphi (0) = 0$, $\tilde\varphi (\infty) = \infty$, and $\tilde\varphi (\ii \m R_+) = \g$. 
Here, $K \ge 1$ denotes a constant depending only on $I_{\m H; 0, \infty} (\g)$ and may change from line to line.
Because $\tilde\varphi$ fixes $\infty$, we know from~\cite[Thm.\,5.1]{lehto2012univalent} that $\tilde\varphi$ extends to a homeomorphism on $\ad{\m H}$, whose restriction 
$\tilde\varphi|_{\m R}$ is moreover a $\lambda(K)$-quasisymmetric homeomorphism of $\m R$ for some distortion function $\lambda(K)$.

Now, note that $\tilde\varphi|_{\m R}$ is not yet the identity function. 
To fix this issue, from the Jerison-Kenig extension theorem (Theorem~\ref{lem_KJ} stated below), we know that $\tilde\varphi|_{\m R}^{-1}$ can be extended to 
a $K$-quasiconformal map $\psi$ on $\m H$ with $\ii \m{R}_+$ fixed. Hence, the map $\varphi:= \tilde\varphi \circ \psi$ is 
the desired quasiconformal map. 

Finally, for a general Jordan domain $(\domain; \bpt, \ept)$,
the result follows 
by conjugating $\varphi$ by a uniformizing conformal map 
from $D$ to $\m H$. 
Indeed, recall that pre-composing or post-composing a $K$-quasiconformal map by a conformal map is again $K$-quasiconformal. 
Moreover, Carath\'eodory's theorem~\cite[Thm.\,2.6]{Pommerenke_boundary} 
shows that the uniformizing map extends to a homeomorphism from $\ad \domain$ to $\ad {\m H}$. 
Therefore, the conjugated quasiconformal map from $\domain$ to $\domain$ also has the asserted boundary values, being the identity function on $\partial \domain$. 
\end{proof}

\begin{theorem}[{Jerison-Kenig extension}] 
\label{lem_KJ}
There exists a function $K(\lambda)$ such that 
every $\lambda$-quasisymmetric function $h$ on $\m{R}$ with  
$h(0) = 0$ can be extended to a $K(\lambda)$-quasiconformal map on $\m H$ with $\ii \m{R}_+$ fixed.  
\end{theorem}

\begin{proof}
The extension is constructed in~\cite[Thm.\,5.8.1]{astala2008elliptic}, following the original proof from~\cite{JK82}. The fact that the extension fixes $\ii \m{R}_+$ follows from the construction.
\end{proof}

The next lemma shows that the Loewner energy is good as a rate function.

\begin{lem} \label{lem: I_semicontinuous}
The map $I_{\domain; \bpt, \ept}$ 
 from $\mc X(\domain; \bpt, \ept)$ to $[0,\infty]$ is lower semicontinuous, and furthermore,  its level set 
$\{\g \in \mc X(\domain; \bpt,\ept) \;|\; I_{\domain; \bpt,\ept}(\g) \le c\}$ is compact for any $c \geq 0$.
\end{lem}

\begin{proof}
Note that compactness of all level sets implies that they are closed, which is equivalent to the lower semicontinuity. 
Also, by conformal invariance, we may take $\domain = \m D$ without loss of generality. 
Therefore, it suffices to show that for any $c \geq 0$, the set $\{\g \in \mc X(\m D; \bpt,\ept) \;|\; I_{\m D; \bpt,\ept}(\g) \le c\}$ is compact.
For this, let $(\g^k)_{k \in \bN}$ be a sequence of chords in $\mc X(\m D; \bpt, \ept)$ with $I_{\m D}(\g^k) \leq c$.
As $\mc C$ is compact, there exists a convergent subsequence in $\mc C$ (still denoted by $\g^k$). 
We prove that the limit is a curve $\g \in \mc X(\m D; \bpt, \ept)$ with $I_{\m D} (\g)\le c$.

We first show that the convergence of $\g^k$ (along a further subsequence) also holds uniformly on compact subsets of $[0,\infty)$ as capacity-parametrized curves. 
For this purpose, for each $k$,  we let $\varphi^k \colon \m D \to \m D$ be a $K$-quasiconformal map as in Lemma~\ref{lem:finite_energy_quasiconformal}, 
so that $\varphi^k (\eta_{\m D; \bpt, \ept}) = \g^k$. Note that $K$ depends on the Loewner energy of $\g^k$ only, and since the energies are uniformly bounded by $c$, we can take $K$ independent of $k$. 
Since $\varphi^k|_{\partial \m D}$ is the identity function,  we can extend $\varphi^k$ by reflection to a $K$-quasiconformal map on 
$\Chat$, and they form a normal family (see~\cite[Thm.\,2.1]{lehto2012univalent}). 
In particular, along a subsequence, $\varphi^{k_j}$ converges uniformly on $\ad{\m D}$ to a quasiconformal map $\varphi$.
This shows that  the limit of $\g^k$ in $\mc C$ is the curve  $\g : = \varphi (\eta_{\m D, x, y})$.
Furthermore, by~\cite[Proof~of~Lem.\,4.1]{lind2010collisions},  
the capacity parametrization of a $K$-quasiconformal curve has modulus of continuity depending only on $K$, so the Arzela-Ascoli theorem implies that, as capacity-parametrized curves, 
the sequence $\g^{k_j}$ converges to $\g $ uniformly on compact subsets of $[0,\infty)$.

It remains to prove the bound $I_{\m D}(\g) \leq c$ for the limit curve. Since the corresponding level set
$\{W \in C^0[0,t] \;|\; I_t(W) \leq c 
\}$ of the Dirichlet energy is compact for all $t > 0$,
we see that along a subsequence, the driving function $W^k$ of $\g^k$ converges uniformly on compact subsets of $[0,\infty)$
to some function $W \in C^0[0,\infty)$. 
Furthermore, \cite[Lem.\,4.2]{lind2010collisions} and the uniform convergence of $\g^k$ imply that $W$ is the driving function of $\g$.
The lower semicontinuity of the Dirichlet energy then gives 
\begin{align*}
I_{\m D}(\g)
= I(W) \leq  \liminf_{k\to \infty} I(W^k)
= \liminf_{k\to \infty} I_{\m D}(\g^k) \leq c ,
\end{align*}
which concludes the proof.
\end{proof}

The above proof also shows that $I_\domain$ is a good rate function on the space of capacity-parametrized curves with the topology of uniform convergence on compact subsets.

Lastly, we recall that any finite-energy chord meets the domain's boundary perpendicularly. 
By the conformal invariance of $I_\domain$ and the assumption that $\partial \domain$ is smooth in neighborhoods of the marked points, we may assume that $(\domain; \bpt,\ept) = (\m H; 0, \infty)$.
 For any angle $\t \in (0,\pi/2)$, we define $\text{Cone}(\t) : = \{z \in \m H \; | \: \t < \arg z <  \pi-\t\}$.

\begin{lemA} \label{lem:perp}
We have 
$\inf I_{\m H; 0, \infty} (\g) = - 8 \log \sin (\t)$,
where the infimum is taken over all chords $\g \in \mc X(\m H; 0, \infty)$ which ever exit $\text{Cone}(\t)$.
In particular, if $\g \in \mc X(\m H; 0, \infty)$ has finite energy, then for all $\t \in (0, \pi/2)$, there exists $\d > 0$ such that $\g_{[0,\d]} \subset \ad{\text{Cone}(\t)}$. 
\end{lemA}

\begin{proof}
The asserted identity was shown in~\cite[Prop.\,3.1]{W1}.
Let $W$ be a driving function of a chord $\g$ with $I(W) < \infty$ and $\t \in (0, \pi/2)$. 
Then, there exists $\d > 0 $ such that $I_\d (W) < -8 \log \sin (\t)$, so the curve generated by $t \mapsto W_{\min(t,\d)}$ 
is contained in $\ad {\text{Cone} (\t)}$ and coincides with $\g$ up to capacity $2\d$.
This proves the lemma.
\end{proof}

Let us remark here on a subtle point in the definition of Loewner energy. 
We require from the definition of $\mc X (\m H;0,\infty)$ that the chords in this set have infinite total half-plane capacity. A priori, one can define Loewner energy for any chord
using $I_T (W)$ in~\eqref{eq_DE}, without assuming that 
its total half-plane capacity $2T$ 
is infinite. However, Lemma~\ref{lem:perp} then implies that if
the chord has finite energy, then it is contained in $\text{Cone} (\t)$ for some $\theta \in (0, \pi/2)$, which implies by~\cite[Thm.\,1]{LLN_capacity} that 
the chord must have
infinite total capacity and 
thus belongs to $\mc X (\m H;0,\infty)$. 
In other words, chords with finite total capacity always have infinite energy.

\subsection{Brownian loop measure}\label{subsec:BLM}

To define the multichordal Loewner energy and potential in Section~\ref{sec_energy_def_and_properties}, 
we use the \emph{Brownian loop measure}~\cite{LSW_CR_chordal, LW2004loupsoup,Law05}. 
In short, it is a $\s$-finite measure on planar unrooted Brownian loops.
We refer to~\cite[Sec.\,3--4]{LW2004loupsoup} for the properties of this measure, and only briefly recall below its features 
important to the present work.

Fix $z \in \domain$ and $t > 0$. Consider the (sub-probability) measure $\m W_z^t$ on Brownian paths started from $z$ on the time interval $[0,t]$ and killed upon hitting the boundary of $\domain$
(run at speed $2$, namely, with diffusion generator $-\D_\domain$ with Dirichlet boundary conditions). 
The disintegration of this measure with respect to the endpoint $w$ gives 
the (sub-probability) measures $\m W_{z \to w}^t$  on Brownian paths from $z$ to $w$ such that 
$$\m W_z^t = \int_\domain \m W^t_{z \to w} \dd w^2. $$
The total mass of $\m W^t_{z \to w}$ is the heat kernel $p_t (z,w) = e^{-t \D_\domain} (z,w)$. 
The \emph{Brownian loop measure} on $\domain$ is defined as
\begin{align*}
\mu_\domain^{\mathrm{loop}}: = \int_0^{\infty} \frac{\ud t}{t} \int_\domain \m W_{z \to z}^t \dd z^2 .
\end{align*}
Upon forgetting the root $z$, this yields a measure on the set of unrooted and  unparametrized loops in $\domain$ 
(so we distinguish loops only by their trace).

It satisfies the following properties:
\begin{itemize}[itemsep=-1pt]
\item \emph{Restriction property}: 
If $U \subset D$, then $\ud \mu_{U}^{\mathrm{loop}} (\ell)= \1\{\ell \subset U\} \, \ud \mu_\domain^{\mathrm{loop}} (\ell)$.

\item \emph{Conformal invariance}: 
If $\varphi \colon \domain \to \domain'$ is a conformal map, then  $\mu^{\mathrm{loop}}_{\domain'} = \varphi_* (\mu^{\mathrm{loop}}_{\domain})$.
\end{itemize}

The total mass of $\mu^{\mathrm{loop}}_{\domain}$ is infinite (e.g., because of small loops) but when considering only loops that intersect two macroscopic disjoint sets, the measure is finite: if $K_1, K_2$ are two disjoint compact subsets of $\ad \domain$, the total mass of Brownian loops that stay in $\domain$ and intersect both $K_1$ and $K_2$ is finite. More generally, if $K_1, \ldots, K_m$ are disjoint compact subsets of $\ad \domain$, 
we denote by 
\begin{align*}
\mc B_\domain(K_1, \ldots, K_m)
:= \mu_\domain^{\mathrm{loop}} \big( \{ \ell \; | \; \ell \cap K_j \neq \emptyset \text{ for all } j = 1,\ldots,m \} \big)
\end{align*}
the mass of the Brownian loops in $\domain$ which intersect all of $K_1, \ldots, K_m$. 
This quantity is positive, finite, and conformally invariant.

We also often use the \emph{Poisson excursion kernel} $P_{\domain;\bpt,\ept}$, defined via
\begin{align}\label{eq:Poisson_def}
P_{\domain;\bpt,\ept} := |\varphi'(\bpt)| |\varphi'(\ept)| P_{\m H;\varphi(\bpt),\varphi(\ept)} , \qquad
\text{where} \qquad P_{\m H;\bpt,\ept} := |\ept - \bpt|^{-2} ,
\end{align}
and where $\varphi \colon \domain \to \m H$ is a conformal map.
(Recall that $\partial \domain$ is smooth in neighborhoods of $\bpt$ and $\ept$ so that $\varphi'(\bpt)$ and $\varphi'(\ept)$ are defined.)

The following conformal restriction formula for the Loewner energy, expressed in terms of Brownian loop measure and Poisson excursion kernels, is crucial.

\begin{lemA}[{\cite[Prop.\,3.1]{W3}}] \label{lem: I_conformal_restriction}
Let $U \subset \domain$ be a simply connected subdomain 
which agrees with $\domain$ in neighborhoods of $\bpt$ and $\ept$.
For each $\g \in \mc X(U;\bpt,\ept)$, we also have
$\g \in \mc X(\domain;\bpt,\ept)$ and
its Loewner energies in $U$ and in $\domain$ differ by 
\begin{align} \label{eq: I_conformal_restriction}
I_{U} (\g) - I_{\domain} (\g)
= 12 \, \mc B_{\domain} (\domain \smallsetminus U, \g) + 3 \log\abs{\varphi'(\bpt) \varphi'(\ept)},
\end{align}
where $\varphi$ is a conformal map from $U$ to $\domain$ fixing $\bpt$ and $\ept$. 
\end{lemA}

Applying~\eqref{eq: I_conformal_restriction} to the hyperbolic geodesic $\eta$ in $(U; \bpt,\ept)$, we get the following relation.
 
\begin{cor}\label{cor:compare_Poisson}
If $\eta$ is the hyperbolic geodesic in $(U; \bpt,\ept)$, then
\begin{align*} 
\log P_{U; \bpt, \ept} - \log P_{\domain; \bpt, \ept}   &=\log \abs{\varphi'(\bpt) \varphi'(\ept)} =  - \frac{1}{3}I_{\domain} (\eta) - 4 \, \mc B_{\domain} (\domain \smallsetminus U, \eta) \le 0.
\end{align*}
\end{cor}

Vir\'ag~\cite{Virag} showed that
$\abs{\varphi'(\bpt) \varphi'(\ept)}$ equals the probability for Brownian excursion in $\domain$ from $\bpt$ to $\ept$ to avoid $D \smallsetminus U$, which also implies that $\log|\varphi'(x) \varphi'(y)| \le 0$.

\section{Multichordal Loewner potential}
\label{sec_energy_def_and_properties}

In this section, we introduce the Loewner potential $\mc H$ 
and Loewner energy $I$ for multichords.
They are defined using 
Brownian loop measure 
(cf.~Section~\ref{subsec:BLM})
and the single-chord Loewner energy~\eqref{eq_LE}
(cf.~Section~\ref{sec:chordal_single_I}). 
The quantities $12 \mc H$ and $I$ only differ by an additive constant that is a function of the boundary data $(x_1, \ldots,x_{2\np};\a)$, in particular, independent of the multichord.
Salient properties of $\mc H$ and $I$ include conformal covariance (Lemma~\ref{lem:H_transform_conformal}), 
the fact that all finite-energy
multichords are obtained as quasiconformal images of analytic multichords  (Proposition~\ref{prop:qc_multi}),
and the fact that the multichordal Loewner energy is a good rate function (Proposition~\ref{prop: H_semicontinuous}),
important to the LDP in Section~\ref{sec:LDPS}.

\subsection{A loop measure}
\label{subsec:a_loop_measure}

The following
loop measure $m_\domain$
 will be used to define the multichordal Loewner potential and energy.
The setup and notation is illustrated in Figure~\ref{fig:NSLE}. 
Throughout, we fix an $\np$-link pattern
\begin{align*}
\a = \{ \link{a_1}{b_1}, \link{a_2}{b_2}, \ldots, \link{a_\np}{b_\np} \} .
\end{align*}
As in~\cite{Law09}, for a multichord $\ad \g = (\g_1, \ldots , \g_\np) \in \mc X_{\a}(\domain;  x_1, \ldots,  x_{2\np})$, we set 
\begin{align} \label{eq: m_alpha}
\begin{split} 
m_\domain(\ad \g) & := \sum_{p=2}^\np 
\mu_\domain^{\mathrm{loop}} \Big( \big\{ \ell \; \big| \; \ell \cap \g_i \neq \emptyset \text{ for at least $p$ of the } i \in \{1,\ldots,n\} \big\} \Big) \\
 & = \int \max\Big( \# \{\text{chords hit by } \ell\} - 1, \, 0\Big) \,\ud \mu_\domain^{\mathrm{loop}} (\ell)  ,
\end{split} 
\end{align}
 and $m_\domain(\g) := 0$ if $\np = 1$.
See also~\cite{Dub_Euler,Dub_comm, KL07, PW19} for alternative forms.

\begin{lem} \label{lem: malpha_cascade}
For each $j \in \{1,\ldots,\np\}$, we have
\begin{align*}
m_\domain(\ad \g) 
= \mc B_{\domain} (\domain \smallsetminus \hat{\domain}_j, \g_j) 
+ m_\domain(\g_1, \ldots, \g_{j-1}, \g_{j+1} \ldots, \g_\np) 
\end{align*}
\textnormal(recall that 
$\hat{\domain}_j$ is
the connected component  of $\domain \smallsetminus \bigcup_{i \neq j} \g_i$
containing the chord $\g_j$\textnormal).
\end{lem}

\begin{proof}
This follows immediately from the definition of $m_\domain$. 
\end{proof}

\begin{lem} \label{lem: malpha_continuous}
The map $m_\domain$ 
is continuous from $\mc X_{\a}(\domain;  x_1, \ldots,  x_{2\np})$ to $[0,\infty)$.
\end{lem}

\begin{proof}
By conformal invariance, we may take $\domain = \m D$ without loss of generality. 
The non-negativity of $m_{\m D}(\ad \g)$ follows from its definition~\eqref{eq: m_alpha}. 
For the continuity, let $\ad \g^k$ be a sequence converging to $\ad \g$ as $k \to \infty$ in 
$\mc X_{\a}(\m D;  x_1, \ldots,  x_{2\np})$. Then, each $\g_j^k$ converges to $\g_j$ in $\mc X(\m D; x_{a_j}, x_{b_j})$.
Now, it suffices to show that 
\begin{align*}
\mu_{\m D}^{\mathrm{loop}} (A_p^k) 
\quad \overset{k \to \infty}{\longrightarrow} \quad 
\mu_{\m D}^{\mathrm{loop}} (A_p) \qquad \text{for each $p = 2,\ldots,\np$} ,
\end{align*}
where 
\begin{align*}
A_p^k := \; & \big\{ \ell \; \big| \; \ell \cap \g_i^k \neq \emptyset  \text{ for at least $p$ of the } i \in \{1,\ldots,n\} \big\} , \\
A_p := \; & \big\{ \ell \; \big| \; \ell \cap \g_i \neq \emptyset  \text{ for at least $p$ of the } i \in \{1,\ldots,n\} \big\} .
\end{align*}
For this, it suffices to show that for each $p = 2,\ldots,\np$, we have
\begin{align} \label{eqn: claim}
|m_{\m D}(\ad \g) - m_{\m D}(\ad \g^k)|  
\leq \; & 
\mu_{\m D}^{\mathrm{loop}} \big( A_p \, \Delta \, A_p^k \big) 
\quad \overset{k \to \infty}{\longrightarrow} \quad 0 ,
\end{align}
where $A_p \, \Delta \, A_p^k := (A_p \smallsetminus A_p^k) \cup (A_p^k \smallsetminus  A_p)$ is the symmetric difference.
To prove~\eqref{eqn: claim}, we fix $p \in \{2,\ldots,\np\}$ and consider a Brownian loop $\ell \in A_p \, \Delta \, A_p^k$. Then, either $\ell$ intersects less than $p$ of the chords $\ad \g^k$ and at least $p$ of the chords in $\ad \g$, or vice versa.
Without loss of generality, let us consider the former case. Then, since $p \geq 2$, $\ell$ intersects at least two the chords in $\ad \g$, so in particular, $\ell$ is a macroscopic loop.
Furthermore, there exists an index $j \in \{1,\ldots,\np\}$ such that $\ell$ intersects $\g_j$ but not $\g_j^k$.
On the other hand, as $\g_j^k$ converges to $\g_j$, we see that when $k$ is large enough,  
the Hausdorff distance of $\g_j^k$ and $\g_j$ is small, so
both  $\g_j^k$ and $\g_j$ belong to a narrow tube.
But then,  
the total mass of $\ell$ intersecting $\g_j$ but avoiding $\g_j^k$ 
tends to zero when $k \to \infty$.  
This proves~\eqref{eqn: claim} and concludes the proof.
\end{proof}

\subsection{Definition of potential and energy}
\label{subsec:multichord_energy_defn}

The loop measure term $m_\domain(\ad \g)$ represents interaction of the chords. We now use it to  define the multichordal Loewner potential $\mc H$ and energy $I$.

\begin{df} \label{defn:def_H_multi}
For $\np \geq 1$, we define the \emph{Loewner potential}
\begin{align} \label{eqn:def_H_multi}
\begin{split}
\mc H_\domain(\ad \g) 
:= \; & \frac{1}{12} \sum_{j=1}^\np I_{\domain; x_{a_j}, x_{b_j}}(\g_j) 
+ m_\domain(\ad \g)
- \frac{1}{4} \sum_{j=1}^\np \log P_{\domain; x_{a_j}, x_{b_j}},
\end{split}
\end{align}
where $x_{a_j}, x_{b_j}$ are the endpoints of the chord $\g_j$ for each $j \in \{1,\ldots,\np\}$.
Note that if $\np = 1$, then the loop measure term is zero, so 
\begin{align} \label{eq:def_H_single}
\mc H_\domain (\g) 
:= \frac{1}{12} I_{\domain; \bpt, \ept}(\g) - \frac{1}{4} \log P_{\domain;\bpt,\ept} .
\end{align}

Since the first two terms in~\eqref{eqn:def_H_multi} 
are non-negative, we can also define
\begin{align*}
\Hmin^\a_\domain (x_1, \ldots,x_{2\np})
:= \inf_{\ad \g}
\mc H_\domain (\ad \g) > -\infty,
\end{align*}
where the infimum is taken over all $\ad \g \in \mc X_\a(\domain;x_1, \ldots,x_{2\np})$.
Note that the infimum depends on the marked points $x_1, \ldots,  x_{2\np} \in \partial \domain$ as well as on the link pattern $\a$.
\end{df}

\begin{df} \label{df:multichord_energy} 
The \emph{multichordal Loewner energy} of $\ad \g \in \mc X_\a (\domain; x_1, \ldots, x_{2\np})$ is
\begin{align*}
I^{\a}_\domain (\ad \g):= 12\,(\mc H_\domain (\ad \g) - \Hmin^\a_\domain (x_1, \ldots,x_{2\np})) \ge 0.
\end{align*}
We say that the multichord $\ad \g$ in $D$ has finite energy if $I^{\a}_\domain (\ad \g) < \infty$. 
\end{df}

\begin{remark} 
The third term in the Loewner potential~\eqref{eqn:def_H_multi} only depends on the boundary data, so it does not affect the Loewner energy. However, as we will see, this term is relevant when we compare multichords of different boundary data (Section~\ref{sec:min}). This term also encodes a conformal covariance property (Lemma~\ref{lem:H_transform_conformal}) that is natural in light of the conformal restriction properties satisfied by SLEs, as studied in~\cite{LSW_CR_chordal}.
Moreover, in the $\k \to 0+$ limit, the $\SLE_\k$ partition function gives the minimal Loewner potential, which can be seen as the semiclassical limit of certain CFT correlation functions (cf.~Proposition~\ref{prop:deterministic_PDEs}). 
\end{remark}

Note that $\ad \g$ having finite energy is equivalent to it having finite potential.

\begin{lem}  \label{lem:H_finite}
The multichord $\ad \g$ has finite energy  in $\domain$ if and only if $I_\domain(\g_j) < \infty$ for all $j$.
\end{lem}

\begin{proof}
The third term in Definition~\ref{defn:def_H_multi} of $\mc H_\domain$ is finite, as the boundary points are distinct. The second term is finite, as $\g_1, \ldots, \g_{\np}$ are pairwise disjoint by the definition of multichords. 
The asserted equivalence thus
follows from the non-negativity of Loewner energy.
\end{proof}

\begin{lem}[Conformal covariance]
\label{lem:H_transform_conformal}
If $\varphi \colon \domain \to \domain'$ is a conformal map, then 
\begin{align*}
\mc H_\domain (\g_1, \ldots, \g_\np) & = \mc H_{\domain'} (\varphi (\g_1), \ldots, \varphi(\g_\np))- \frac{1}{4} \sum_{j=1}^{2\np} \log |\varphi'(x_j)|, \\
I_\domain^\a (\g_1, \ldots, \g_\np) & = I_{\domain'}^\a  (\varphi (\g_1), \ldots, \varphi(\g_\np)). 
\end{align*}
\end{lem}

\begin{proof}
The Poisson kernel transforms under the conformal map $\varphi$ as \begin{align*}
P_{\domain;\bpt,\ept} = |\varphi'(\bpt)||\varphi'(\ept)| P_{\domain',\varphi(\bpt),\varphi(\ept)}
\end{align*}
by definition~\eqref{eq:Poisson_def},
and the single-chord Loewner energy and Brownian loop measure are conformally invariant.
Both assertions thus follow 
from Definitions~\ref{defn:def_H_multi}~and~\ref{df:multichord_energy}.
\end{proof}

\subsection{Further properties}
\label{subsec:multichord_energy_properties}

Next, we show that 
minimizers of the Loewner potential
are geodesic multichords (Corollary~\ref{cor:min_H_geodesic}); 
finite-energy multichords are quasiconformal images of analytic multichords  (Proposition~\ref{prop:qc_multi}); 
the potential is lower semicontinuous with compact level sets; 
and potential minimizers exist (both in Proposition~\ref{prop: H_semicontinuous}).

\begin{lem}[Cascade relation of $\mc H$]
\label{lem:H_cascade}
For each $j \in \{1,\ldots,\np\}$, we have
\begin{align} \label{eq:induction_H}
\mc H_\domain(\ad \g) 
= \; & \mc H_{\hat{\domain}_j} (\g_j) 
+ \mc H_\domain (\g_1, \ldots, \g_{j-1}, \g_{j+1} \ldots, \g_\np) . 
\end{align}
\end{lem}

\begin{proof}
By Definition~\ref{defn:def_H_multi} and Lemma~\ref{lem: malpha_cascade}, the left-hand side of 
\eqref{eq:induction_H} reads
\begin{align*}
\mc H_\domain(\ad \g) 
= \; & \frac{1}{12} \sum_{j=1}^\np I_{\domain; x_{a_j}, x_{b_j}}(\g_j) 
- \frac{1}{4} \sum_{j=1}^\np \log P_{\domain; x_{a_j}, x_{b_j}} \\
\; & + \mc B_{\domain} (\domain \smallsetminus \hat{\domain}_j, \g_j) 
+ m_\domain(\g_1, \ldots, \g_{j-1}, \g_{j+1} \ldots, \g_\np) .
\end{align*}
On the other hand, by Definition~\ref{defn:def_H_multi}, 
the right-hand side of~\eqref{eq:induction_H} reads
\begin{align*}
\; & \mc H_{\hat{\domain}_j} (\g_j) 
+ \mc H_\domain (\g_1, \ldots, \g_{j-1}, \g_{j+1} \ldots, \g_\np) \\
= \; & \frac{1}{12} I_{\hat{\domain}_j; x_{a_j}, x_{b_j}} (\g_j) - \frac{1}{4} \log P_{\hat{\domain}_j; x_{a_j}, x_{b_j}}
+ \frac{1}{12} \sum_{i \neq j} I_{\domain; x_{a_i}, x_{b_i}}(\g_i) \\
\; & + m_\domain(\g_1, \ldots, \g_{j-1}, \g_{j+1} \ldots, \g_\np) 
- \frac{1}{4} \sum_{i \neq j} \log P_{\domain; x_{a_i}, x_{b_i}} .
\end{align*}
After using the conformal restriction formula~\eqref{eq: I_conformal_restriction} 
with $(U;\bpt, \ept)  = (\hat{\domain}_j;x_{a_j}, x_{b_j})$,
we see that the left and right-hand sides of the asserted formula~\eqref{eq:induction_H} agree.
\end{proof}

\begin{cor}
\label{cor:min_H_geodesic}
Any minimizer of $\mc H_\domain$ in $\mc X_\a (\domain;x_1, \ldots,  x_{2\np})$ is a geodesic multichord. 
\end{cor}

We will show in Proposition~\ref{prop: H_semicontinuous}
that there actually exists a minimizer 
and in Corollary~\ref{cor:geod_unique} that the minimizer is unique. 
Thus, in particular, for each link pattern $\a$ there exists a unique geodesic multichord.

\begin{proof}[Proof of Corollary~\ref{cor:min_H_geodesic}]
Let $\ad \eta = (\eta_1, \ldots, \eta_\np)$ be a minimizer of $\mc H_\domain$ in $\mc X_\a(\domain;x_1, \ldots,  x_{2\np})$. 
Then, by Lemma~\ref{lem:H_cascade}, for each $j$, 
the chord $\eta_j$ minimizes $\mc H_{\hat{\domain}_j} $ and thus $I_{\hat{\domain}_j}$ 
among all chords in $\hat{\domain}_j$ with the same boundary points.
Hence, each $\eta_j$ is the hyperbolic geodesic in its component $\hat{\domain}_j$,
i.e., $\ad \eta$ is a geodesic multichord. 
\end{proof}

\begin{remark}
Note that the minimizers of $\mc H_\domain$ and $I_\domain^\a$ in $\mc X_\a (\domain;x_1, \ldots,  x_{2\np})$ are the same. 
In particular, 
$\ad \eta$ is a minimizer of $\mc H_\domain$ in $\mc X_\a(\domain;x_1, \ldots,  x_{2\np})$ if and only if
$I_\domain^\a(\ad \eta) = 0$.
\end{remark}

The following estimates will be useful later.
\begin{lem}  \label{lem:compare}
For each $j \in \{1, \ldots, \np\}$, the following hold:
\begin{enumerate}[(i),itemsep=3pt]
  \item \label{it:potential_estimate} We have
\begin{align*}
\mc H_\domain(\g_1, \ldots, \g_{j-1}, \g_{j+1}, \ldots,  \g_{\np-1}) 
\leq \; & \mc H_\domain(\ad \g) 
+ \frac{1}{4} \log P_{D,x_{a_j}, x_{b_j} }. 
\end{align*}

  \item \label{it:energy_estimate} For each $\np$-link pattern $\a$, we have
\begin{align*}
I_{\hat \domain_j} (\g_j) \le I^{\a}_\domain (\ad \g) .
\end{align*}
\end{enumerate}
\end{lem} 

\begin{proof} 
For \ref{it:potential_estimate}, 
Lemma~\ref{lem:H_cascade} and formula~\eqref{eq:def_H_single} imply that 
\begin{align*}
\mc H_\domain(\g_1, \ldots, \g_{j-1}, \g_{j+1}, \ldots,  \g_{\np-1}) 
= \; & \mc H_\domain(\ad \g) - \mc H_{\hat D_j} (\g_j) \\
= \; & \mc H_\domain(\ad \g) - \frac{1}{12} I_{\hat D_j} (\g_j) + \frac{1}{4} \log P_{\hat D_j, x_{a_j}, x_{b_j}}\\
\leq \; & \mc H_\domain(\ad \g) + \frac{1}{4} \log P_{D,x_{a_j}, x_{b_j} } ,
\end{align*}
where the last inequality follows from the non-negativity of the Loewner energy and the domain monotonicity of the Poisson kernel
(Corollary~\ref{cor:compare_Poisson} with $U = \hat D_j$).

To prove \ref{it:energy_estimate}, 
without loss of generality, we may assume that $j=1$.
Let $\eta$ be the hyperbolic geodesic in $\hat \domain_1$ with the same endpoints as $\g_1$.
Because $I_{\hat \domain_1} (\eta) = 0$, Definition~\ref{defn:def_H_multi} 
and Lemma~\ref{lem:H_cascade} show that
\begin{align*}
I_{\hat \domain_1} (\g_1)& = I_{\hat \domain_1} (\g_1) - I_{\hat \domain_1} (\eta) = 12 (\mc H_{\hat \domain_1} (\g_1) - \mc H_{\hat \domain_1} (\eta)) \\
& = 12 (\mc H_{\domain} (\ad \g) - \mc H_{ \domain} (\eta, \g_2, \ldots, \g_\np))\\
& \le 12( \mc H_{\domain} (\ad \g) - \Hmin^{\a}_{\domain} (x_1, \ldots, x_{2\np}))  =  I^{\a}_\domain (\ad \g)
\end{align*}
as claimed.
\end{proof}

Next, in Proposition~\ref{prop:qc_multi} we show that any finite-energy multichord is a quasiconformal image of a certain smooth reference multichord, 
where the quasiconformal constant only depends on the domain data, link pattern, and the Loewner energy of the multichord.
Similarly to Lemma~\ref{lem: I_semicontinuous} in Section~\ref{sec:chordal_single_I}, this property provides the required precompactness in order to find a subsequence converging to a potential (and energy) minimizer. 
This allows us to conclude with Proposition~\ref{prop: H_semicontinuous},  comprising the most important properties of the Loewner potential.

We fix an (arbitrary) ordering $(\link{a_1}{b_1}, \ldots ,\link{a_\np}{b_\np})$ of the pairs in  each  $\np$-link pattern $\a$.
Then, we define the \emph{reference multichord}
$(\rfmultichord_1, \ldots, \rfmultichord_\np)$ associated to $\a$ and $(\domain;x_1,\ldots,x_{2\np})$ as follows.
We first take $\rfmultichord_1$ to be the hyperbolic geodesic in $(\domain;x_{a_1},x_{b_1})$, 
and next, for each $j = 2,\ldots,\np$, we let $\rfmultichord_j$ 
be the hyperbolic geodesic from $x_{a_j}$ to $x_{b_j}$ in 
the appropriate connected component of 
$\domain \smallsetminus \bigcup_{i < j} \rfmultichord_i$. 
(We cautiously note that this object is not a geodesic multichord unless $\np=1$.)

\begin{prop}\label{prop:qc_multi}
Fix 
domain data 
$(\domain; x_1, \ldots, x_{2\np})$ and a link pattern $\a$.
If $I^{\a}_\domain (\ad \g) <\infty$, then there exists $K \in [1,\infty)$, depending only on $I^{\a}_\domain (\ad \g)$, and 
a $K$-quasiconformal map $\varphi$ 
such that $\g_j = \varphi (\rfmultichord_j)$ for all $j \in \{1,\ldots, \np\}$, $\varphi(\domain) = \domain$, and $\varphi$
extends continuously to $\ad {\domain}$ and equals the identity function on $\partial \domain$.
\end{prop}

\begin{proof}
We construct the desired $K$-quasiconformal map by induction on $\np \geq 1$. 
The constant $K$ may change from line to line, 
as long as it depends only on 
the domain data $(\domain; x_1,\ldots,x_{2\np})$, link pattern $\a$, and energy $I^{\a}_\domain (\ad \g)$.
The initial case $\np = 1$ is covered by Lemma~\ref{lem:finite_energy_quasiconformal}.
For the induction step, we assume that the assertion holds for the $(\np-1)$-link pattern $\hat \a = (\link{a_1}{b_1}, \ldots ,\link{a_{\np-1}}{b_{\np-1}})$ on the same marked domain after omitting the markings at $x_{a_\np}$ and $x_{b_\np}$.
It follows immediately from 
the construction that $(\rfmultichord_1, \ldots, \rfmultichord_{\np-1})$ is the reference multichord associated to $\hat \a$.

Item~\ref{it:energy_estimate} of
Lemma~\ref{lem:compare} and Definition~\ref{df:multichord_energy} 
imply that $I^{\hat \a}_\domain(\g_1, \ldots, \g_{\np-1})$ is uniformly bounded by a constant that only depends on $(\domain; x_1,\ldots,x_{2\np})$,  $\a$, and $I^{\a}_\domain(\g_1, \ldots, \g_{\np})$.
Therefore, by the induction hypothesis, there exists 
a $K$-quasiconformal map $\hat \varphi \colon \domain \to \domain$ 
equaling the identity function on $\partial \domain$ 
and sending $\rfmultichord_j$ to $\g_j$ for all $j \le \np-1$.
Note that $\hat \varphi (\hat{\domain}_{\rfmultichord}) = \hat{\domain}_\np$, where 
$\hat{\domain}_{\rfmultichord}$ and $\hat{\domain}_\np$
are respectively the connected components containing 
$\rfmultichord_\np$ and $\g_\np$.

We will construct the desired map $\varphi$ by modifying 
$\hat \varphi$ in $\hat{\domain}_{\rfmultichord}$ in such a way that the modified map sends $\rfmultichord_\np$ to $\g_\np$ while remaining $K$-quasiconformal.
For this purpose, we let $f_\g \colon \hat{\domain}_\np \to \m H$ 
(resp.~$f_\rfmultichord \colon \hat{\domain}_{\rfmultichord} \to \m H$) 
be a uniformizing conformal map sending $x_{a_\np}$  to $0$ and $x_{b_\np}$ to $\infty$. 
Then, $f_\g \circ \hat \varphi \circ f_\rfmultichord^{-1}$ is a $K$-quasiconformal map from $\m H$ to itself, so its boundary values define 
a $K$-quasisymmetric homeomorphism of $\m R$. 
Let $\psi \colon \m H \to \m H$ 
be its Jerison-Kenig extension: 
a $K$-quasiconformal map that fixes $\ii \m R_+$ 
with boundary values 
\begin{align*}
\psi|_{\m R} = f_\g \circ \hat \varphi \circ f_\rfmultichord^{-1} . 
\end{align*}

Now, item~\ref{it:energy_estimate} of
Lemma~\ref{lem:compare} implies that 
$I_{\m H} (f_\g (\g_\np)) = I_{\hat{\domain}_\np} (\g_\np) \le I^{\a}_\domain (\ad \g)$,
so we may apply 
Lemma~\ref{lem:finite_energy_quasiconformal} to 
find a $K$-quasiconformal map
$\psi_\np \colon \m H \to \m H$ that equals the identity function on 
$\m R$ and maps $\ii \m R_+$ to $f_\g (\g_\np)$. 
In conclusion, the composition 
\begin{align*}
f_\g ^{-1} \circ \psi_\np \circ \psi \circ f_\rfmultichord \colon \hat{\domain}_{\rfmultichord}  \to  \hat{\domain}_\np 
\end{align*}
coincides with $\hat \varphi$ on $\partial \hat{\domain}_{\rfmultichord}$
and sends $\rfmultichord_\np$ to $\g_\np$.
Hence, we can define the sought $K$-quasiconformal map
$\varphi \colon \domain \to \domain$ as
\begin{align*}
\varphi(z) := 
\begin{cases}
\hat \varphi(z) , & z \in  \domain \smallsetminus \hat{\domain}_{\rfmultichord} , \\
(f_\g^{-1} \circ 
\psi_n \circ \psi
\circ f_\rfmultichord ) (z) ,  & z \in \hat{\domain}_{\rfmultichord} ,
\end{cases}
\end{align*}
so $\g_j = \varphi (\rfmultichord_j)$ for all $j$, and $\varphi$ extends to the identity function on $\partial \domain$.
\end{proof}

\begin{prop} \label{prop: H_semicontinuous}
The Loewner potential has the following properties.
\begin{enumerate}
\item $\mc H_\domain$ is lower semicontinuous from 
$\mc X_\a(\domain; x_1, \ldots, x_{2\np})$ to $[\mc M_D^\a(x_1,\ldots, x_{2\np}),\infty]$.

\item The level set 
$\{\ad \g \in \mc X_\a(\domain; x_1, \ldots, x_{2\np}) \;|\; \mc H_D(\ad  \g) \le c\}$ is compact for any $c \in \m R$.

\item There exists a multichord minimizing $\mc H_\domain$ in $\mc X_\a(\domain; x_1, \ldots, x_{2\np})$. 
\end{enumerate}
\end{prop}

\begin{proof}
The lower semicontinuity follows from Definition~\ref{defn:def_H_multi}:
the maps $\g_j \mapsto I_{\domain; x_{a_j}, x_{b_j}}(\g_j)$ are lower semicontinuous by Lemma~\ref{lem: I_semicontinuous},
the quantity $\log P_{\domain; x_{a_j}, x_{b_j}}$ does not depend on the multichords, 
and the map $\ad \g \mapsto m_\domain(\ad \g)$
is continuous by Lemma~\ref{lem: malpha_continuous}. 

The compactness of level sets follows from similar arguments as 
Lemma~\ref{lem: I_semicontinuous}. Indeed, for a sequence
$(\ad \g^k)_{k \in \m N}$  
of multichords in  $\mc X_\a (\domain; x_1, \ldots, x_{2\np})$ with $\mc H_\domain(\ad \g^k) \le c$, Proposition~\ref{prop:qc_multi} 
gives $K(c)$-quasiconformal maps $\varphi^k \colon \domain \to \domain$  sending 
the reference multichord
$\ad \rfmultichord := (\rfmultichord_1, \ldots, \rfmultichord_\np)$
to $\ad \g^k$
and 
such that $\varphi^k |_{\partial \domain}$ is the identity function for each $k$. 
Along a subsequence, $\varphi^k$ converge uniformly on $\ad D$ to a 
quasiconformal map $\varphi$. Hence, along this subsequence, $\ad \g^k$ converge to the multichord $\ad \g := \varphi(\ad \rfmultichord)$ for the Hausdorff metric. 
The lower semicontinuity (property 1) of $\mc H_\domain$ then shows that $\mc H_\domain (\ad \g) \le c$.

The existence of a minimizer 
is an immediate consequence of the compactness of the level sets by considering a minimizing sequence.
This finishes the proof. 
\end{proof}

\section{Minimizers of the potential}
\label{sec:min}

In this section, we investigate minimizers of the 
Loewner potential (and energy). 
Relying on the conformal covariance of $\mc H$ (Lemma~\ref{lem:H_transform_conformal}), 
we shall assume that $\domain = \m H$.
We show that any geodesic multichord
gives rise to a rational function with prescribed set of critical points (Proposition~\ref{prop:gamma_h}). 
This result proves Theorem~\ref{thm:real_rational}. 
Furthermore, by classifying such rational functions, we show that for each 
$\np$-link pattern $\a$ and domain data 
$(\m H; x_1, \ldots, x_{2\np})$, there exists a unique geodesic multichord (Corollary~\ref{cor:geod_unique}). 
Theorem~\ref{thm:main_unique_minimizing} is a consequence of this result and the fact that all 
minimizers of the potential 
are geodesic (Corollary~\ref{cor:min_H_geodesic}).
Finally, in Section~\ref{subsec:null-field_PDE} 
we derive the Loewner flows of the minimizer (Proposition~\ref{prop:main_Loewner})  and the semiclassical null-state PDEs for the minimal potential 
(Proposition~\ref{prop:deterministic_PDEs}).

\subsection{Geodesic multichords and rational functions} \label{subsec:rational}

To begin, we study geodesic multichords
$\ad{\eta} := (\eta_1, \ldots, \eta_\np)$ in $\mc X_{\a}(\m H;  x_1, \ldots,  x_{2\np})$, where $x_1 < \cdots < x_{2\np}$. 
By considering the union of $\ad{\eta}$, its complex conjugate $\ad{\eta}^*$, and the real line $\m R$,
we can regard them as graphs embedded in $\Chat$. 
Precisely, to each geodesic multichord $\ad{\eta}$ we associate a graph $G_\eta$ with vertices $\{x_1,\ldots,x_{2\np}\}$ and edges 
\begin{align*}
E = \big\{\eta_1, \eta_1^*, \ldots, \eta_\np, \eta^*_\np, [x_1,x_2], \ldots, [x_{2\np-1}, x_{2\np}], [x_{2\np}, x_1]\big\},
\end{align*}
where $[x_{2\np}, x_1]$ denotes the segment $[x_{2\np}, +\infty] \cup [-\infty, x_1]$ in $\Chat$.
We call the connected components of $\Chat \smallsetminus \bigcup_{e \in E} \,e $ the \emph{faces} of $G_\eta$.
By symmetry, $G_\eta$ is again a graph with the \emph{geodesic property}, 
that is, each edge $e \in E$ is the hyperbolic geodesic in the domain $\hat \domain_e$ formed by the union of the two faces adjacent to $e$. 
In fact, to each $\ad \eta$ and thus to $G_\eta$, we can associate a unique rational function $h_\eta$ 
whose critical points are given by the endpoints of the chords $\eta_1, \ldots, \eta_\np$ (see Proposition~\ref{prop:gamma_h}).

We first recall some terminology.
A \emph{rational function} is an analytic branched covering $h$ of $\Chat$ over $\Chat$, or equivalently, the ratio $h = P/Q$ of two polynomials $P,Q \in \m C[X]$. 
A point $x_0 \in \Chat$ is a \emph{critical point} (equivalently, a branched point)
with \emph{index} $k \ge 2$ if 
\begin{align*}
h(x) = h(x_0) + C (x - x_0)^k  + O ((x - x_0)^{k+1})
\end{align*}
for some constant $C \neq 0$ in local coordinates of $\Chat$. 
A point $y \in \Chat$ is a \emph{regular value} of $h$ if $y$ is not image by $h$ of any critical point.
The \emph{degree} of $h$ is the number of preimages of any regular value. We call $h^{-1} (\hat {\m R})$ the \emph{real locus} of $h$. Then, $h$ is a  \emph{real} rational function if $P$ and $Q$ can be chosen from $\m R[X]$, or equivalently, $h(\hat {\m R}) \subset \hat {\m R}$.

By the Riemann-Hurwitz formula  on Euler characteristics,
a rational function of degree $\np+1$ has $2\np$ critical points 
if and only if each of them has index two (in this case, we say that the critical points of $h$ are distinct):
\begin{align} \label{eq:Hurwitz}
(\np+1) \chi (\Chat)  - 2\np (2-1) = 2\np + 2 - 2\np = 2 = \chi (\Chat).
\end{align}

\begin{prop} \label{prop:gamma_h}
Fix $x_1 < \cdots < x_{2\np}$. For each geodesic multichord $\ad{\eta}$ in $(\m H;x_1, \ldots, x_{2\np})$, 
there exists a unique rational function $h_\eta \colon \Chat \to \Chat$ of degree $\np+1$, 
up to post-composition by $\PSL(2, \m R)$ and by $ \iota \colon z \mapsto -z$,
such that the set of critical points of $h_\eta$ is exactly $\{x_1, \ldots, x_{2\np}\}$.
Moreover, the graph $G_\eta$ is given by the real locus of $h_\eta$, i.e., $h_\eta^{-1} (\hat {\m R}) = \bigcup_{e\in E} e$. 
\end{prop}

We note that, since $h_\eta$ preserves 
$\hat {\m R}$
it is actually a \emph{real} rational function. 
This result proves Theorem~\ref{thm:real_rational}.

\begin{proof}[Proof of Proposition~\ref{prop:gamma_h}]
The complement $\m H \smallsetminus \ad \eta$ has $\np+1$ faces. 
Pick one face $F$ and consider a uniformizing conformal map $h_\eta$ from $F$ onto $\m H$ (unique up to post-composition with elements of $\PSL(2,\m R)$) or onto $\m H^*$ (by replacing $h_\eta$ with $\iota \circ h_\eta$).
Without loss of generality, we consider the former case. Also, we suppose that $\eta_1$ is adjacent to $F$ and another face $F'$.
Then,  since $\eta_1$ is a hyperbolic geodesic in $\hat \domain_{\eta_1}$, 
the map $h_\eta$ extends by reflection to a conformal map on $\hat \domain_{\eta_1}$.
In particular, this extension of $h_\eta$ maps $F'$ conformally onto $\m H^*$. 
By iterating these analytic continuations across all of the chords $\eta_k$, we obtain a meromorphic function $h_\eta \colon \m H \to \Chat$. Furthermore, 
$h_\eta$ also extends to $\ad{\m H}$,  
and its restriction $h_\eta|_{\hat {\m R}}$ 
takes values in $\hat {\m R}$. Hence, 
via Schwarz reflection, 
we extend $h_\eta$ to $\Chat$ via $h_\eta(z) := h_\eta(z^*)^*$ for all $z \in \m H^*$.

Now, it follows from the construction that $h_\eta^{-1} (\hat {\m R}) = \bigcup_{e\in E} e$. Moreover, $h_\eta$ is a rational function of degree $\np+1$, as exactly $\np+1$ faces are mapped to $\m H$ and $\np+1$ faces to $\m H^*$. 
Finally, another choice of initial face $F$ yields the same function up to post-composition with elements of $\PSL(2,\m R)$, or with $\iota$.
This concludes the proof.
\end{proof}

Now, thanks to Proposition~\ref{prop:gamma_h}, 
in order to classify geodesic multichords it suffices to classify equivalence classes of rational functions 
with prescribed critical points. 
Since post-composition by elements of $\PSL (2, \m C)$ (i.e, M\"obius transformations of $\Chat$) does not change the critical points, 
we may consider the equivalence classes $\Rat_{d}$ of rational functions of degree $d = \np + 1 \geq 2$ modulo post-composition by $\PSL (2, \m C)$. 
We obtain information from them using Goldberg's results from~\cite{Gol91}, which we now briefly recall.

Let $\Poly_d$ denote the space of polynomials of degree at most $d$.
If $h$ is a rational function of degree $d$, then we can write $h = P/Q$ for two linearly independent polynomials $P,Q \in \Poly_d$.
In particular, $V_h : = \operatorname{span}_{\m C} (P, Q)$ lies in the Grassmann manifold $G_2 (\Poly_d)$ of two-dimensional subspaces of $\Poly_d$.
The group $\PSL (2, \m C)$ acts on $h$ by post-composition,
\begin{displaymath}
\begin{pmatrix}
a & b \\ c & d 
\end{pmatrix} \cdot h = \frac{a P + b Q}{ c P + d Q},
\end{displaymath}
and the image generates the same element in  $G_2(\Poly_d)$ as $h$. 
Conversely, any basis $\{P, Q\}$ of $V_h$ gives a rational function $P/Q$ which equals a post-composition of $h$ by an element of $\PSL (2, \m C)$.
Therefore, we identify 
$\Rat_d \simeq G_2(\Poly_d)$ by $[h] \mapsto V_h.$
Furthermore, the \emph{Wronski map} 
\begin{align*}
\Phi_d([h]) := [P'Q - Q'P]/_{\m C^*} , \qquad \text{with} \quad 
[h] \simeq V_h = \text{span}_{\m C}(P,Q) ,
\end{align*}
from $\Rat_d$ to the projective space $\Poly_{2d - 2}^*/\m C^* \simeq  \m P^{2d-2}$ is well-defined, because 
\begin{align*}
\left(\frac{P}{Q}\right)' = & \; \frac{P'Q - Q'P}{Q^2} \neq 0  \\
\text{and } \quad \Phi_d \left[\begin{pmatrix}
a & b \\ c & d 
\end{pmatrix} \cdot h\right] = & \; [(ad - bc) (P'Q - Q'P)]/_{\m C^*} = \Phi_d ([h]).
\end{align*}
Moreover, the zeros of $\Phi_d ([h])$ are exactly the critical points of $h$. We therefore view $\Phi_d$ as a map from a rational function in  $\Rat_d$ to the set of critical points counted with multiplicity.

Goldberg's following result is crucial to our classification of geodesic multichords.
Recall from~\eqref{eq:catalan} that $\Catalan_n$ denotes the $n$:th Catalan number.

\begin{theorem}[{\cite[Prop.\,2.3~and~Thm.\,3.4]{Gol91}}] \label{thm:Goldberg}
The function $\Phi_d$ is a complex analytic map of degree $\Catalan_{d-1}$.
\end{theorem}

In particular, with $d = \np+1$, we have:

\begin{theorem}[{\cite[Thm.\,1.3]{Gol91}}]
A set of $2\np$ distinct points is the set of critical points of at most $\Catalan_\np$ rational functions of degree $\np+1$ that are not $\PSL(2,\m C)$-equivalent.
\end{theorem}

A combination of the above theorems with Proposition~\ref{prop:gamma_h}
and the existence of the geodesic multichord for any given 
boundary data $(x_1, \ldots,x_{2\np};\a)$
from Proposition~\ref{prop: H_semicontinuous} and Corollary~\ref{cor:min_H_geodesic} 
shows that the maximal number is achieved when the critical points are $2 \np$ distinct real numbers.

\begin{cor}\label{cor:geod_unique}
There are exactly $\Catalan_\np$ preimages of $\{x_1 < \cdots <x_{2\np}\}$ by $\Phi_{\np+1}$. 
In particular, for given boundary data $(x_1, \ldots,x_{2\np};\a)$, there exists a unique geodesic multichord
\textnormal{(}which is also the unique potential minimizer in $\mc X_{\a}(\m H;  x_1, \ldots,  x_{2\np})$\textnormal{)}.   
\end{cor}

This result implies Theorem~\ref{thm:main_unique_minimizing}.

\begin{proof}[Proof of Corollary~\ref{cor:geod_unique}]
To each geodesic multichord $\ad \eta$ in $(\m H;x_1, \ldots,x_{2\np})$ 
we associate a rational function $h_\eta$ as in Proposition~\ref{prop:gamma_h}.
Let $\ad{\eta}$ and $\ad \eta'$ be two geodesic multichords. If there exists  $A \in \PSL (2, \m C)$ such that $h_\eta = A \circ h_{\eta'}$, 
then we have $(A \circ h_{\eta'}) ([x_1, x_2]) = h_\eta ([x_1, x_2]) \subset \m R$, which implies that either 
$A$ or $\iota \circ A$ belongs to $\PSL(2,\m R)$, so 
\begin{align*}
h_\eta^{-1} (\hat {\m R}) = h_{\eta'}^{-1} (\hat {\m R}) . 
\end{align*}
Hence, we have $\ad \eta = \ad \eta'$ and the map $\eta \to [h_\eta]$ is injective. 
Theorem~\ref{thm:Goldberg} implies that there exist at most $\Catalan_\np$ geodesic multichords in $(\m H;x_1, \ldots,x_{2\np})$. 
On the other hand, $\Catalan_\np$ also equals the number of $\np$-link patterns $\a$. 
By Proposition~\ref{prop: H_semicontinuous}, for any $\a$ there exists at least one 
minimizer of the potential, 
which is a geodesic multichord.  This proves the corollary. 
\end{proof}

As a by-product, we obtain  
an analytic proof of the following equivalent form of the Shapiro conjecture for the Grassmannian of $2$-planes 
(see~\cite{Sot00} for the general conjecture and~\cite{EG02,EG11} for other proofs).

\corCatalan*

\begin{proof}
If  a rational function of degree $\np + 1$ 
has $2\np$ distinct real critical points, then it is $\PSL(2,\m C)$-equivalent to $h_\eta$ associated to a geodesic multichord $\ad \eta$ via Proposition~\ref{prop:gamma_h}. 
In particular, $h_\eta$ maps the real line to the real line and has real coefficients. 
The general case follows by deforming the polynomial $\Phi_{\np+1} ([h])$ in $\Poly^*_{2\np}/\m C^*$ to those with simple zeros, see~\cite[Sec.\,7]{EG02}.
\end{proof}

\begin{cor}\label{cor:analytic_dependence}
The class $[h_\eta] \in \Rat_{\np+1}$ associated to the unique geodesic multichord $\ad \eta$ 
in $\m H$ with boundary data $(x_1, \ldots, x_{2\np};\a)$ depends analytically on  $x_1 <\cdots <x_{2\np}$.
\end{cor}

\begin{proof}
From Theorem~\ref{thm:Goldberg} and Corollary~\ref{cor:geod_unique}, we see that $\{x_1 <  \cdots <x_{2\np}\}$ is a regular value of $\Phi_{\np+1}$, which is locally an analytic diffeomorphism in a neighborhood of any of its preimages  
(where the link pattern of its real locus is constant).
\end{proof}

\subsection{Characterization by Loewner differential equations}\label{subsec:char_minimizer}

Next, we derive the Loewner flow for the geodesic multichord.
For this purpose, we first need to show that the minimal potential
$\Hmin^{\a}_{\m H}$ is differentiable with respect to variation of the marked points $x_1 < \cdots <x_{2\np}$, for each link pattern $\a$. This is a consequence of the analyticity of the Wronski map (cf.~Corollary~\ref{cor:analytic_dependence}).
 
\begin{prop}\label{prop:M_diff}
Let $\ad \eta$ be the unique geodesic multichord in $\m H$
associated to the boundary data $(x_1, \ldots, x_{2\np};\a)$. 
For each $i \in \{1,\ldots,2\np\}$, the function $\Hmin^{\a}_{\m H} (x_1, \ldots, x_{2\np})$ is differentiable in $x_i$. 
Moreover, for each $j \in \{1,\ldots,2\np\}$, we have
\begin{align*} 
\partial_{a_j} \Hmin^{\a}_{\m H} (x_1, \ldots, x_{2\np}) 
= \partial_1 \Hmin_{\hat{\m H}_j} (x_{a_j}, x_{b_j}) ,
\end{align*}
where $x_{a_j}$ and $x_{b_j}$ are the endpoints of $\eta_j$,  
and $\hat{\m H}_j$ is its component in $\m H \smallsetminus \bigcup_{i \neq j} \eta_i$.
\end{prop}

In the case $\np=1$, the minimal potential has an explicit formula
obtained from~\eqref{eq:initialdef_Hmin} and~\eqref{eq:Poisson_def}, which is obviously differentiable:
\begin{align} \label{eq:Hmin_explicit2}
\Hmin_{\m H}(x_1, x_2) = \frac{1}{2} \log |x_2 - x_1|
\qquad \Longrightarrow \qquad 
\partial_1 \Hmin_{\m H}(x_1, x_2) = \frac{1}{2(x_1 - x_2)} .
\end{align}

\begin{proof}[Proof of Proposition~\ref{prop:M_diff}]
For simplicity of notation and without loss of generality, we assume that $i=1$ and $\link{1}{2} \in \a$, 
and we let $\eta_1$ be the chord connecting $x_1$ and $x_2$. In particular, we keep the variables $x_2, \ldots, x_{2\np}$ fixed.
To vary the position of $x_1$, for all $\ept < x_2$, we let 
$ (\eta_1^{\ept}, \eta_2^{\ept}, \ldots, \eta_\np^{\ept} )$ 
denote the geodesic multichord in $\m H$ associated to $\a$ with boundary points $\{ \ept < x_2 < x_3 < \cdots < x_{2\np} \}$. 
Also, we let $\hat {\m H}_1(\ept)$ 
denote the connected component of $\m H \smallsetminus \{\eta_2^{\ept}, \ldots, \eta_\np^{\ept}\}$
containing $\ept$ and $x_2$ on its boundary.
Note that $\hat {\m H}_1(\ept)$ depends on $\ept$. 
Now, recalling the cascade relation from Lemma~\ref{lem:H_cascade}, we set
\begin{align*}
\psi_{\ept}(\bpt) & 
: =  \mc H_{\m H} (\eta^{\bpt}, \eta_2^{\ept}, \ldots, \eta_\np^{\ept}) 
 = \Hmin_{\hat {\m H}_1(\ept)} (\bpt, x_2) + \mc H_{\m H}(\eta_2^{\ept}, \ldots, \eta_\np^{\ept}) ,
\end{align*}
where $\eta^{\bpt}$ is the hyperbolic geodesic in $(\hat {\m H}_1(\ept);\bpt,x_2)$.
Then, at $\bpt = \ept$, we have $\eta^{\ept} = \eta_1^{\ept}$ and 
$\psi_\ept(\ept) = \Hmin^{\a}_{\m H} (\ept, x_2, \ldots, x_{2\np})$.
We will argue that the function
\begin{align*}
\ept  \; \mapsto \;  \psi_\ept' (\ept) =  \partial_1 \Hmin_{\hat {\m H}_1 (\ept)} (\ept, x_2) 
\end{align*}
is continuous. Indeed, if $\varphi_\ept$ is the conformal map
from $\hat{\m H}_1(\ept)$ onto $\m H$ fixing $x_2$ and such that $\varphi_\ept'(x_2) = 1$ and $\varphi_\ept''(x_2) = 0$, 
then by Lemma~\ref{lem:H_transform_conformal}, we have
\begin{align*}
\Hmin_{\hat{\m H}_1(\ept)} (\cdot, x_2) = \Hmin_{\m H}(\varphi_\ept(\cdot), x_2) -  \frac{\log \varphi_\ept'(\cdot)}{4} .
\end{align*}
Hence, using Corollary~\ref{cor:analytic_dependence} and the explicit formula~\eqref{eq:Hmin_explicit2}, 
we see that the derivative
\begin{align*}
\ept \; \mapsto \; \psi_\ept' (\ept) =
 \frac{\varphi_\ept'(\ept)}{2(\varphi_\ept(\ept) - x_2)} - \frac{\varphi''_\ept(\ept)}{4 \varphi_\ept'(\ept)}
 \end{align*}
is continuous in $\ept$.
With these preparations, we are ready to show
the differentiability of $\Hmin^\a_{\m H}$ in its first variable.
For $\ept$ in a small neighborhood $B_\vare (x_1)$ of $x_1$, we have
\begin{align*}
\Hmin^{\a}_{\m H} (\ept, x_2, \ldots, x_{2\np}) & - \Hmin^{\a}_{\m H} (x_1, x_2, \ldots, x_{2\np}) = \psi_{\ept}(\ept) - \psi_{x_1}(x_1) \\
& =   \psi_{\ept}(\ept) - \psi_{\ept}(x_1) + \psi_{\ept}(x_1)- \psi_{x_1}(x_1) \\
& \ge  \psi_{\ept}(\ept) - \psi_{\ept}(x_1) 
= (\ept - x_1) \psi_\ept'(\ept) - R_\ept (\ept,x_1) ,
\end{align*}
for some remainder $R_\ept (\ept,x_1)$. 
In fact, the remainder is bounded as $|R_\ept (\ept, x_1)| \le  c |\ept - x_1|^2$, 
where $c \in (0,\infty)$ is uniform over all $\ept \in B_\vare (x_1)$,
because $c$ depends on $\psi_\ept''$ in $B_\vare (x_1)$, which depends smoothly on $\ept$ by Corollary~\ref{cor:analytic_dependence}.
Similarly, after inverting the roles of $\ept$ and $x_1$, we have 
\begin{align*}
\Hmin^{\a}_{\m H} (\ept, x_2, \ldots, x_{2\np})  - \Hmin^{\a}_{\m H} (x_1, x_2, \ldots, x_{2\np}) \le  (\ept - x_1) \psi_{x_1}' (x_1)  
+ R_{x_1} (x_1,\ept) .
\end{align*}
Now, $\psi_\ept'(\ept)$ converges to $\psi_{x_1}'(x_1)$ as $\ept$ approaches $x_1$, 
so we obtain the differentiability of $\Hmin^\a_{\m H}$ in the first variable, 
and we also see that
$
\partial_1 \Hmin_{\hat{\m H}_1} (x_1, x_2) = \partial_1 \Hmin^{\a}_{\m H} (x_1, \ldots, x_{2\np}).$
\end{proof}

\begin{remark}
Compared to the differentiability of $\SLE_\k$ partition functions with respect to the boundary points in~\cite{JL_smooth}, 
the above proof is 
different, relying on the analytic dependence of the geodesic multichord on the marked points, which we get from the associated rational function.
\end{remark}

Now, we derive the Loewner flow for the potential-minimizing geodesic multichord.
Recall from the introduction that we wrote $\mc U :=12 \Hmin_{\m H}^\a$ for notational simplicity.

\mainLoewner*

For the proof, we first recall a conformal distortion formula of Loewner driving functions.
Let $\varphi \colon U  \to \tilde U$ be a conformal map between two neighborhoods $U$ and $\tilde U$ in $\m H$ of $\bpt \in \m R$ such that $\varphi(\bpt) = \bpt$. 
Let $\g \in \mc X(\m H; \bpt, \ept)$ and $\tilde \g : = \varphi (\g \cap U)$. 
Let $t\mapsto W_t$ (resp.~$s \mapsto \tilde W_s$) be the driving function of $\g$ (resp.~$\tilde \g$), defined in a neighborhood of $0$. Then, 
by~\cite[Eq.\,(11)]{RW} we know that
$W$ is right-differentiable at $0$ if and only if $\tilde W$ is right-differentiable at $0$. Moreover, we have
\begin{align} \label{eq:conformal_distortion}
 \varphi'(\bpt) \frac{\ud  \tilde W_s} {\ud s} \Big|_{s = 0} = \frac{\ud  W_t} {\ud t} \Big|_{t = 0} - \; 3 \frac{\varphi''(\bpt)}{\varphi'(\bpt)}.
\end{align}

\begin{proof}[Proof of Proposition~\ref{prop:main_Loewner}]
We first consider the single-chord case $\np = 1$. In this case,
the minimizer of $\mc H_{\m H}$  in $\mc X(\m H ; x_1, x_2)$ is given by the hyperbolic geodesic in $\m H$, namely, the semi-circle $\eta$  with endpoints $x_1$ and $x_2$.
Let $\varphi \colon \m H \to \m H$ be a M\"obius
transformation 
with $\varphi(x_1)=x_1$ and $\varphi(x_2)=\infty$.
Then, the driving function of $\varphi(\eta)$ is the constant function $\tilde W_s \equiv x_1$ and we have $\varphi''(x_1)/\varphi'(x_1) = 2/(x_2 - x_1)$. 
Hence, Equation~\eqref{eq:conformal_distortion} gives 
\begin{align*}
\frac{\ud W_t}{\ud t}\Big|_{t = 0}  = \frac{6}{x_2 - x_1} .
\end{align*}
Thus, since hyperbolic geodesics are preserved under their own Loewner flow, we obtain 
\begin{align*} \begin{dcases}
\frac{\ud W_t}{\ud t} = \; & \frac{6}{V_t - W_t}  , 
\qquad W_0 = x_1 , \\
\frac{\ud V_t}{\ud t}  = \; & \frac{2 }{V_t - W_t} , \qquad V_0 = x_2,
\end{dcases}
\end{align*}
where $V_t : = V^2_t$ is the Loewner flow of $x_2$. 
By~\eqref{eq:Hmin_explicit2}, 
this is exactly Equation~\eqref{eqn:DE_for_drivers}.

For the general case, note that under the Loewner flow starting from any point $x_a$, the resulting multichord is still the minimizer for the corresponding boundary point and link pattern. 
Therefore, we only need to prove the asserted equations~\eqref{eqn:DE_for_drivers} at $t = 0$.
For simplicity of notation, we assume that $x_a : = x_{a_1}$ and $x_b := x_{b_1}$ 
are the endpoints of $\eta_1$. 
If $\psi \colon \hat{\m H}_1 \to \m H$ is a conformal map, then since $\psi (\eta_1)$ minimizes $\mc H_{\m H}$ among all chords in $(\m H;\psi(x_a),\psi(x_b))$, 
the driving function $\tilde W$ of $\psi (\eta_1)$ satisfies 
 $(\ud \tilde W_t/\ud t)|_{t = 0} =  -12 \, \partial_1 \Hmin_{\m H} (\psi(x_a), \psi(x_b))$. Therefore,
\begin{align*}
-12 \, \partial_{x_a} \Hmin_{\m H} (\psi(x_a), \psi(x_b)) 
= \; & -12 \, \psi'(x_a) \, \partial_{1} \Hmin_{\m H} (\psi(x_a), \psi(x_b))   \\
=  \; & \psi'(x_a) \frac{\ud \tilde W_t}{\ud t}\Big|_{t = 0} 
=\frac{\ud W_t}{\ud t}\Big|_{t = 0} - \; 3 \frac{\psi''(x_a)}{\psi'(x_a)} ,
\end{align*}
where the last equality follows from Equation~\eqref{eq:conformal_distortion}.  
On the other hand, Lemma~\ref{lem:H_transform_conformal} gives
\begin{align*}
- \Hmin_{\hat{\m H}_1} (x_a, x_b) = \frac{1}{4} \log |\psi'(x_a)\psi'(x_b)| - \Hmin_{\m H} (\psi(x_a), \psi(x_b)) ,
\end{align*}
so we deduce that
\begin{align*}
-12\, \partial_1 \Hmin_{\hat{\m H}_1} (x_a, x_b) = 3 \frac{\psi''(x_a)}{\psi'(x_a)} - 12 \, \partial_{x_a} \Hmin_{\m H} (\psi(x_a), \psi(x_b)) =\frac{\ud W_t}{\ud t}\Big|_{t = 0}.
\end{align*}
From Proposition~\ref{prop:M_diff}, we now obtain the time-evolution of $W$ in~\eqref{eqn:DE_for_drivers} with $\mc U =12 \Hmin_{\m H}^\a$. 
The time-evolutions of $V^i$ for $i \neq a$ 
follow  from the Loewner equation~\eqref{eqn:LE} with $V_t^i= g_t(x_i)$.
\end{proof}

\subsection{Semiclassical null-state PDEs}
\label{subsec:null-field_PDE}

In this section, we derive the 
system of semiclassical null-state PDEs for the minimal potential $\Hmin_{\m H}^\a$. 
Later, in Corollary~\ref{cor:semicl_lim_of_pf} we will show that $\Hmin_{\m H}^\a$ 
describe the semiclassical limit of the $\SLE_\k$ partition functions $\PartF_\a$,
which are solutions to the level two null-state PDEs~\eqref{eq: multichordal SLE PDEs}.

\deterministicPDEs*

\begin{proof}
Let $\ad \eta$ be the unique geodesic multichord with 
boundary data $(x_1, \ldots,x_{2\np};\a)$. 
Then, we have $\mc U :=12 \Hmin_{\m H}^\a = 12\mc H_{\m H} (\ad \eta)$ 
by definition. 
Without loss of generality, we prove~\eqref{eq:deterministic_PDEs}  
for $j=1$, and we assume that $\eta_1$ is the chord connecting $x_1$ to $x_b$, for $\link{1}{b} \in \a$.

We consider the Loewner flow $(g_t)_{t \in [0, T)}$ associated to the driving function $W$ of $\eta_1$, satisfying~\eqref{eqn:DE_for_drivers} with $j=1$.
We write $\ad{\eta}^t = (\eta_1^t, \ldots, \eta_\np^t)$ for the image multichord under the flow $g_t$. 
Since $\ad{\eta}^t$ is still a geodesic multichord, we have
\begin{align} \label{eq:U_eta_t}
12 \, \mc H_{\m H}(\ad \eta^t) 
= \; &  \mc U (W_t, g_t(x_2), g_t(x_3), \ldots, g_t(x_{2\np})) 
\end{align}
for any $t \in [0,T)$.
We will take the time-derivative of~\eqref{eq:U_eta_t} at $t=0$ in two ways, whose equality yields the asserted PDE~\eqref{eq:deterministic_PDEs}.
First, we have
\begin{align} \label{eq:RHS_derivative}
\begin{split}
12 \frac{\ud }{\ud t} \mc H_{\m H} (\ad \eta^t) \Big|_{t = 0}
= \; & \frac{\ud W_t}{\ud t} \Big|_{t = 0}  \; \partial_1 \mc U(x_1, \ldots, x_{2\np})
+ \sum_{i \neq 1} \frac{2 \partial_i \mc U(x_1, \ldots, x_{2\np})}{x_i - x_1} \\
= \; & - (\partial_1 \mc U(x_1, \ldots, x_{2\np}))^2 
+ \sum_{i \neq 1} \frac{2 \partial_i \mc U(x_1, \ldots, x_{2\np})}{x_i - x_1} .
\end{split}
\end{align}

On the other hand, let $\m H^L$ and $\m H^R$ be the two connected components of $\m H \smallsetminus \eta_1$, and $\ad \eta^L$ and $\ad \eta^R$ the sub-multichords of $\ad \eta$ contained in $\m H^L$ and $\m H^R$. Also, for each $t \in [0,T)$, 
denote their images under the conformal map $g_t$ by $\m H^{L,t} := g_t(\m H^L)$, $\m H^{R,t} : = g_t(\m H^R)$, 
$\ad \eta^{L,t} := g_t (\ad \eta^L)$, and $\ad \eta^{R,t} := g_t (\ad \eta^R)$. 
We claim that the following 
decomposition holds:
\begin{align}\label{eq:cascade_2}
12 \, \mc H_{\m H}(\ad \eta^t)  = 12 \, \mc H_{\m H}(\eta_1^t) + 12 \, \mc H_{\m H^{L,t}} (\ad \eta^{L,t}) + 12 \, \mc H_{\m H^{R,t}} (\ad \eta^{R,t}).
\end{align}
Indeed, this follows by applying Lemma~\ref{lem:H_cascade} successively $\np-1$ times to the chords $\eta^t_2, \ldots, \eta^t_\np$
on both sides of~\eqref{eq:cascade_2} (or alternatively, using the conformal covariance of $\mc H$ from Lemma~\ref{lem:H_transform_conformal} and the determinant expression from Theorem~\ref{thm:main_H_det}, whose proof is independent).
Now, from Lemma~\ref{lem:H_transform_conformal} and the Loewner flow~\eqref{eqn:DE_for_drivers}, we obtain
\begin{align*}
& \frac{\ud }{\ud t} \Big[12 \, \mc H_{\m H^{L,t}} (\ad \eta^{L,t}) + 12 \, \mc H_{\m H^{R,t}} (\ad \eta^{R,t})\Big]_{t = 0} \\
& =\frac{\ud }{\ud t}\Big[12 \, \mc H_{\m H^L} (\ad \eta^{L}) + 12 \, \mc H_{\m H^R} (\ad \eta^{R}) + 3 \sum_{i \neq 1,b} \log |g_t'(x_i)|\Big]_{t = 0} 
=  - \sum_{i\neq 1,b} \frac{6}{(x_i - x_1)^2}.
\end{align*}
Also, using the definition~\eqref{eq_DE} of the Loewner energy $I_{\m H;0, \infty}$, we compute
\begin{align*}
12 \,\frac{\ud }{\ud t}\mc H_{\m H} (\eta_1^t) \big|_{t = 0} 
& = \frac{\ud }{\ud t} \Big[ I_{\m H; W_t, g_t(x_b)}(\eta_1^t) 
+ 6 \log |g_t(x_b)- W_t|\Big]_{t= 0}   \\ 
&  = - \frac{ 1}{2} 
\Big( \varphi'(x_1) \; \frac{\ud \tilde W_s}{\ud s}  \Big)^2\Big|_{s = 0}
- \frac{6}{x_b - x_1}  \frac{\ud W_t}{\ud t} \Big|_{t = 0}
+ \frac{12}{(x_b-x_1)^2},
\end{align*}
where $(\varphi'(x_1))^2$ 
is the scaling factor of the capacity parametrization,  
and $\varphi \colon \m H \to \m H$ is a M\"obius 
transformation 
with $\varphi(x_1) = x_1$ and $\varphi(x_b) = \infty$.
Now, Equation~\eqref{eq:conformal_distortion} gives
\begin{align*}
\varphi'(x_1) \, \frac{\ud \tilde W_s}{\ud s} \Big|_{s = 0} 
= \frac{\ud W_t}{\ud t} \Big|_{t = 0} - 3 \; \frac{\varphi''(x_1)}{\varphi'(x_1)} 
= - \partial_1 \mc U (x_1, \ldots, x_{2\np}) - \frac{6}{x_b - x_1}.
\end{align*}
After combining the computations above, we finally obtain
\begin{align*}
12 \, \frac{\ud }{\ud t} \mc H_{\m H} (\ad \eta^t) \big|_{t = 0} 
& = -\frac{1}{2} \Big( \partial_1 \mc U + \frac{6}{x_b - x_1}\Big)^2  
+ \frac{6\,\partial_1 \mc U }{x_b - x_1}
+\frac{12}{(x_b - x_1)^2} - \sum_{i\neq 1,b} \frac{6}{(x_i - x_1)^2} \\
& = -\frac{1}{2}(\partial_1 \mc U)^2 - \sum_{i\neq 1} \frac{6}{(x_i - x_1)^2},
\end{align*}
  and by equating this with the right-hand side of~\eqref{eq:RHS_derivative}, 
  we obtain the asserted PDE~\eqref{eq:deterministic_PDEs}
 with $j=1$. The other PDEs follow by symmetry.
\end{proof}

\begin{remark} 
Like the BPZ PDEs~\eqref{eq: multichordal SLE PDEs},
the semiclassical PDE system~\eqref{eq:deterministic_PDEs} does not depend on the link pattern $\a$.
Therefore, by Proposition~\ref{prop:deterministic_PDEs},
we have already found $\Catalan_\np$ solutions to it. 
From analogy with conformal blocks in boundary CFT,
we believe that the total number of solutions should be given by
counting more general planar link patterns, 
as in~\cite[Sec.\,2.5~and~3.1]{Peltola:BSA}:
\begin{align}\label{eq:Kostka}
\sum_{s \in \{ 0,2,4,...,2\np\} } 
\frac{s + 1}{\np + s/2 + 1} \binom{2\np }{\np + s/2} ,
\end{align}
where each summand is the number of link patterns with $2\np$ indices and 
$s$ ``defects'' (i.e., 
lines going to infinity), using the terminology in \cite{Peltola:BSA}.
Note also that the $s$:th summand in \eqref{eq:Kostka}
is the Kostka number of $(s,1,1,\ldots,1)$,
also used in~\cite{EGSV,EG11} to enumerate nets.
\end{remark}

\section{Large deviations of SLEs}
\label{sec:LDPS}

The purpose of this section is to prove the LDP Theorem~\ref{thm:main_LDP}
for multichordal $\SLE_{0+}$
For this, we consider multichordal $\SLE_\k$ as a probability measure on $\mc X_\a (\domain; x_1, \ldots, x_{2\np})$ endowed with the Hausdorff metric from Definition~\ref{defn: Hausdorff}. 
Since the Loewner energy, the topology, and the SLE
measures are conformally invariant, we use $\domain = \m H$ 
throughout as the reference domain
(so we have a simple description of driving functions). 
We also assume throughout that $\k < 8/3$ for convenience and without loss of generality.

\subsection{Large deviations of single-chord SLE: finite time} \label{subsec:LDP_single_finite_T}

Note that if a Loewner driving function $W \in C^0[0,T]$ 
has finite Dirichlet energy~\eqref{eq_DE}, 
$I_T (W) < \infty$, then the extension of $W$ by a constant function generates a chord (by~Lemma~\ref{lem:finite_energy_quasiconformal}) in $\mc X(\m H; 0,\infty)$ which has finite Loewner energy. 
In particular, $\mc L_T (W)$ is 
necessarily a simple curve 
in $\m H$ starting from $0$
and bounded away from $\infty$. 
Moreover, any capacity-parametrized bounded simple curve 
$\g_{[0,T]}$ in $\m H$ starting from $0$ determines a unique driving function $W \in C^0[0,T]$ such that $\g_{[0,T]} =  \mc L_T (W)$ according to~\eqref{eq:W_from_curve}.
With a slight abuse of notation, we define the finite-time Loewner energy
\begin{align*}
I_T (\g_{[0,T]}) : = I_T (W) \in [0,\infty], \qquad
\text{and} \qquad
I_T (K_T) := \infty =  \inf_{W \in \mc L_T^{-1} (\{K_T\}) }I_T(W)
\end{align*}
for all sets $K_T \in \mc K_T$ that are not bounded simple curves.

Throughout this section, we endow the space $\mc K_T \subset \mc C$  with
the topology induced from the Hausdorff metric (Definition~\ref{defn: Hausdorff})
and we endow the space of driving functions $C^0[0,T]$ with the uniform norm.
Now, recall that the Loewner transform $\mc{L}_T$ is continuous for the Carath\'eodory topology (Definition~\ref{defn:Cara}), but not for the Hausdorff metric. However, 
thanks to Lemma~\ref{lem:closure}, when considering the infimum of the Loewner energy, this is not an issue because driving functions 
at which $\mc{L}_T$ is not continuous do not generate simple curves, 
thus having infinite energy.
More precisely, we have the following lemma.

\begin{lem} \label{lem:infimums}
For any Hausdorff-closed subset $\closed$ and Hausdorff-open subset $\open$ of $\mc K_T$, 
\begin{align*}
\inf_{W \in \ad{\mc{L}_T^{-1} (\closed)}} 
I_T(W) = & \; \inf_{W \in \mc{L}_T^{-1} (\closed)} 
I_T(W) , \\
\inf_{W \in \mc{L}_T^{-1} (\open)^\circ} 
I_T(W) = & \;  \inf_{W \in \mc{L}_T^{-1} (\open)} 
I_T(W),
\end{align*}
where $\ad{\mc{L}_T^{-1} (\closed)}$ 
\textnormal{(}resp.~$\mc{L}_T^{-1} (\open)^\circ$\textnormal{)} 
denotes the closure \textnormal{(}resp.~interior\textnormal{)} 
of $\mc{L}_T^{-1} (\closed)$ 
\textnormal{(}resp.~$\mc{L}_T^{-1} (\open)$\textnormal{)} 
in $C^0 [0,T]$ for the uniform norm.
\end{lem}

\begin{proof}
We first prove the assertion for the closed set $\closed$. 
By Lemma~\ref{lem:closure}, we know that for any driving function 
$W \in \ad{\mc{L}_T^{-1} (\closed)} \smallsetminus \mc{L}_T^{-1} ( \closed)$, 
the corresponding hull $\mc L_T(W)$ 
has a non-empty interior, which 
has infinite Loewner energy. 
This proves the assertion for $\closed$. 
Next,  the complement 
$\tilde{\closed} := \mc K_T \smallsetminus \open$ 
of the open set $\open$ is  closed, and we have
 \begin{align*}
C^0[0,T] \smallsetminus \mc{L}_T^{-1} (\open) =
 \mc{L}_T^{-1} (\tilde{\closed}) \qquad\text{and} \qquad \mc{L}_T^{-1} (\open) \smallsetminus \mc{L}_T^{-1} (\open)^\circ  =  \ad{\mc{L}_T^{-1} (\tilde{\closed})} \smallsetminus \mc{L}_T^{-1} (\tilde{\closed}).
\end{align*}
The previous argument then shows that 
$I_T (W) = \infty$ for $W \in \mc{L}_T^{-1} (\open) \smallsetminus \mc{L}_T^{-1} (\open)^\circ$. 
\end{proof}

The next lemma holds for $\SLE_\k$ with $\k \in (0,4]$, although in the present article we are only concerned with $\k \to 0+$.

\begin{lem} \label{lem:SLE_closure}
For any Hausdorff-closed subset $\closed$ and Hausdorff-open subset $\open$ of $\mc K_T$, 
\begin{align*}
\m P [ \, \sqrt{\k} B_{[0,T]} \in \ad{\mc{L}_T^{-1} (\closed)} \; ] 
= & \; 
\m P [ \sqrt{\k} B_{[0,T]} \in \mc{L}_T^{-1} (\closed) ] 
, \\
\m P [ \sqrt{\k} B_{[0,T]} \in \mc{L}_T^{-1} (\open)^\circ ] 
= & \; 
\m P [ \sqrt{\k} B_{[0,T]} \in \mc{L}_T^{-1} (\open) ] .
\end{align*}
\end{lem}

\begin{proof}
As $\SLE_\k$ is almost surely 
generated by a chord, 
similar arguments as for Lemma~\ref{lem:infimums}
give 
$\m P [ \sqrt{\k} B_{[0,T]} \in \ad{\mc{L}_T^{-1} (\closed)} \smallsetminus \mc{L}_T^{-1} (\closed) ] = \m P [ \sqrt{\k} B_{[0,T]} \in \mc{L}_T^{-1} (\open) \smallsetminus \mc{L}_T^{-1} (\open)^\circ ] = 0$.
\end{proof}

We denote by $\m P^\k$ the $\SLE_\k$  probability measure on  $\mc X(\m H; 0, \infty)$,
and by $\m P$ the standard Wiener measure.
By collecting the results from the previous lemmas, we obtain a LDP for these curves from 
Schilder's theorem on Brownian paths:

\begin{theorem} \label{thm:Schilder}
\textnormal{(Schilder; see, e.g.,~\cite[Ch.\,5.2]{DZ10})} 
Fix $T \in (0,\infty)$.
The process $(\sqrt{\k} B_t)_{t \in [0,T]}$ satisfies the following LDP in $C^0[0,T]$
with good rate function $I_T$\textnormal{:}

For any closed subset $\closed$ and open subset $\open$ of $C^0[0,T]$, we have
\begin{align*}
\limsup_{\k \to 0+} \; \k \log \m P [ \sqrt{\k} B_{[0,T]} \in \closed ] \; \leq \; & - \inf_{W \in \closed} I_T (W) ,
\\
\liminf_{\k \to 0+} \; \k \log \m P [ \sqrt{\k} B_{[0,T]} \in \open ] \; \geq \; & - \inf_{W \in \open} I_T (W). 
\end{align*} 
\end{theorem}

\begin{prop} \label{prop:LDP_finite_time}
Fix $T \in (0,\infty)$. 
The initial segments $\g^\k_{[0,T]} \in \mc K_T$ of chordal $\SLE_\k$ curves 
satisfy the following LDP in $\mc K_T$ with good rate function $I_T$\textnormal{:}

For any Hausdorff-closed subset $\closed$ and Hausdorff-open subset $\open$ of $\mc K_T$, we have
\begin{align}
\begin{split} \label{eq:upper_lower_truncated}
\limsup_{\k \to 0+} \; 
\k \log \m P^{\k} [ \g^\k_{[0,T]} \in \closed ] 
\; & \leq \;  - \inf_{K_T \in  \closed} I_T (K_T) ,
\\
\liminf_{\k \to 0+} \; 
\k \log \m P^{\k} [ \g^\k_{[0,T]} \in \open ] 
\; & \geq \; - \inf_{K_T \in  \open} I_T (K_T) . 
\end{split}
\end{align}
\end{prop}

\begin{proof}
By Schilder's theorem (Theorem~\ref{thm:Schilder}) and Lemma~\ref{lem:infimums}, we know that
\begin{align*}
\limsup_{\k \to 0+} \; 
\k \log \m P [ \, \sqrt{\k} B_{[0,T]} \in \ad{\mc{L}_T^{-1} (\closed)} \; ] 
\; \leq \; &  - \inf_{W \in \ad{\mc{L}_T^{-1} (\closed)}} \; I_T (W) 
\, \; = \;  - \inf_{W \in \mc{L}_T^{-1} (\closed)} I_T (W) ,
\\
\liminf_{\k \to 0+} \; 
\k \log \m P [ \sqrt{\k} B_{[0,T]} \in \mc{L}_T^{-1} (\open)^\circ ] 
\; \geq \; &  - \inf_{W \in \mc{L}_T^{-1} (\open)^\circ} I_T (W) \; = \; - \inf_{W \in \mc{L}_T^{-1} (\open)} I_T (W).
\end{align*}
So~\eqref{eq:upper_lower_truncated} follows from Lemma~\ref{lem:SLE_closure}.
The lower semicontinuity and compactness of the level sets of $I_T$ follow  
as in the proof of Lemma~\ref{lem: I_semicontinuous}, 
using the fact that a finite-energy curve  $\g_{[0,T]}$ is the image of the interval $[0,\ii]$ under a quasiconformal self-map of $\m H$.
\end{proof}

\subsection{Large deviations of single-chord SLE: infinite time}
\label{subsec:LDP_single_whole}

The goal of this section is to establish the LDP for chordal $\SLE_\k$ 
all the way to the target point. 
The idea is to replace the event of the whole $\SLE_\k$ curve being close to a given chord 
with the event that a ``truncated'' $\SLE_\k$ is close to 
the truncated curve, 
and then apply the finite-time LDP from Proposition~\ref{prop:LDP_finite_time}
combined with a suitable estimate for the error made in this truncation.

For brevity, we denote the curve space by $\mc X : = \mc X (\m H; 0, \infty)$, and we denote
$I := I_{\m H; 0, \infty}$.
In fact, we will prove the LDP for the $\SLE_\k$ probability measures
on the compact space $\mc C$ introduced in Section~\ref{subsec:topologies}.
As $\SLE_\k$ curves for small $\k$ belong almost surely to $\mc X$, the following statement is exactly the same as for $\mc X$ endowed with the relative topology induced from $\mc C$. 
The reason for using $\mc C$ instead is purely topological: 
the space $\mc C$ is compact whereas $\mc X$ is not.
Note also that a set is compact in $\mc X$ if and only if it is compact in $\mc C$.

To state our result, we
extend the definition of the Loewner energy to all elements of $\mc C$ by defining $I(K) : = \infty$ if $K \in \mc C \smallsetminus \mc X$. 

\begin{thm} \label{thm: SLE large deviation}
The family $(\m P^\k)_{\k > 0}$ of probability measures of the chordal $\SLE_\k$ curve $\g^\k$
satisfies the following LDP in $\mc C$ with good rate function $I$\textnormal{:}

For any closed subset $\closed$ and open subset $\open$ of $\mc C$, we have
\begin{align}
\limsup_{\k \to 0+} \;  \k \log \m P^\k [ \g^\k \in \closed ] 
\; & \leq \;  - \inf_{K \in \closed} I (K) , \label{eq:single_whole_upper}
\\
\liminf_{\k \to 0+} \; \k \log \m P^\k [ \g^\k \in \open ] 
\; & \geq \; - \inf_{ K\in \open} I (K). \label{eq:single_whole_lower}
\end{align}
\end{thm}

\begin{remark} \label{rem:goodness} 
The compactness of the level sets of $I$ is already proved in Lemma~\ref{lem: I_semicontinuous}.
\end{remark}

In order to apply the finite-time LDP from Proposition~\ref{prop:LDP_finite_time}, 
we introduce a truncated Loewner energy for all sets $K \in \mc C$ (which do not necessarily have driving functions):
we set
\begin{align}\label{eq:def_I_compact_set}
    \tilde I(K) := \inf_{K \subset \g \in \mc X} I(\g) \in [0,\infty],
\end{align}
with the convention that the infimum of an empty set is $\infty$. 
Note that a generic set $K \in \mc C$ can be disconnected and it does not have to contain $0$ nor $\infty$, in which case $I (K)$ is infinite 
while $\tilde I(K)$ can be finite. 
The lower semicontinuity of $I$ implies that $\tilde I$ is lower semicontinuous on $\mc C$. 
Note also that this definition coincides with our original definition of the Loewner energy in the following cases:
\begin{itemize}[itemsep = -2pt]
\item for $K \in \mc C$ containing a simple path $\g$ in $\ad{\m H}$ connecting $0$ to $\infty$ (and which may touch $\m R$), 
we have $\tilde I(K) = I(\g)$ if $K = \g \in \mc X$, and $\tilde I(K) = \infty = I(K)$ otherwise;
       \item for $K_T \in \mc K_T$, we have $I_T(K_T) = \tilde I(K_T)$.
\end{itemize}

To prove Theorem~\ref{thm: SLE large deviation}, we use the following key 
Lemmas~\ref{lem:K_lowerbound_energy}--\ref{lem:open_lower_bound}. 
We write the Euclidean closed half-disc and semi-circle as
\begin{align*}
\bigdisc_R :  = \{z \in \ad{\m H} \; | \; |z| \le R\} 
\qquad \text{and} \qquad 
S_R : = \{z \in \ad{\m H} \; | \; |z| = R\} .
\end{align*}
We also denote by $\mc B^h_{\vare} (K)$ the $\vare$-Hausdorff-neighborhood of $K \in \mc C$,  that is,
\begin{align*}
\mc B^h_{\vare} (K) : = \{\tilde K \in \mc C \; | \; d_h (K, \tilde K) < \vare \} , \qquad
\ad{\mc B}^h_{\vare} (K) : = \{\tilde K \in \mc C \; | \; d_h (K, \tilde K) \le \vare \}.
\end{align*}

\begin{lem}\label{lem:K_lowerbound_energy}
Let $ 0 \le M  <\infty$. If  $K \in \mc C$ satisfies $\tilde I(K)\in [M ,\infty]$, then for all $\d > 0$, there exist $r > 0$ and $\vare > 0$ such that 
   \begin{align} \label{eq:tilde_I_neighborhood}
   \tilde I(\tilde K) \ge M  - \d \qquad \textnormal{for all } \; \tilde K \in \ad{\mc B}^h_{\vare} (K \cap \bigdisc_r).       
   \end{align}
\end{lem}
\begin{proof}
Note that $K \cap \bigdisc_r$ converges to $K$ in $\mc C$
as $r \to \infty$. Hence, by the lower semicontinuity of $\tilde I$ 
and the assumption $\tilde I(K) \ge M$, there exists $r > 0$ such that 
 $\tilde I(K \cap \bigdisc_r) \ge M - \d/2$,
and the lower semicontinuity then again gives $\vare > 0$ such that~\eqref{eq:tilde_I_neighborhood} holds.
\end{proof}

For a curve $\g \in \mc X$ endowed with capacity parametrization
and a radius $R >0$, we define $\tau_R \in (0,\infty)$ to be
the hitting time of $\g$ to $S_R$. 
We also set $T_R : = \mathrm{hcap}(\bigdisc_R) /2$.  
Then, by the monotonicity of the capacity~\eqref{eq:monotonicity_capacity}, we have $\tau_R \leq T_R$.
Note also that when $\g = \g^\k \sim \m P^\k$ is the chordal $\SLE_\k$ curve, 
then $\tau_R$ is a stopping time. 
The following result controls both the decay rate of the probability of the $\SLE_\k$ curve $\g^\k$ and the minimal energy needed for any curve $\g$ to come back to $\bigdisc_r$ after hitting a large radius $R$.

\begin{figure}
 \centering
 \includegraphics[width=0.9\textwidth]{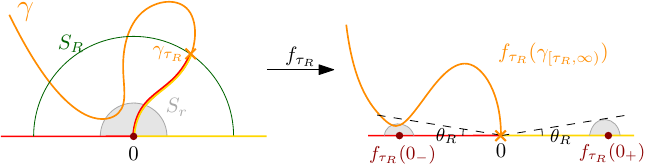}
 \caption{\label{fig:theta_R} 
 Illustration of the conformal map $f_{\tau_R}$, the angle $\t_R$, the left boundary $\partial_-$ (red) and the right boundary $\partial_+$ (gold) used to derive~\eqref{eq:theta_zero} in the proof of Proposition~\ref{prop:coming_back_bound}.}  
 \end{figure}

\begin{prop}\label{prop:coming_back_bound}
 For each $r > 0$ and for any $M \in [0,\infty)$, there exists $R > r$ such that 
 \begin{enumerate}[(i),itemsep=3pt]
 
  \item \label{it:coming_back_bound_deterministic} $\;\, \inf \{I(\g) \; | \; \g \in \mc X, \g_{[\tau_R, \infty)} \cap S_r \neq \emptyset\} \ge M$; and
    
     \item  \label{it:SLE_bound}
 $\underset{\k \to 0+}{\limsup} \k \log \m P^\k [\g^\k_{[\tau_R, \infty)} \cap S_r \neq \emptyset ] \le -M.$
 \end{enumerate}
\end{prop}

\begin{proof}
The claim \ref{it:SLE_bound} is  
treated in Corollary~\ref{cor:return_probability} in Appendix~\ref{app:estimate}. 
To study the claim \ref{it:coming_back_bound_deterministic}, 
as illustrated in Figure~\ref{fig:theta_R}, 
we let  $\t_R (\g)$ be the infimum of $\t \in [0,\pi/2]$ such that 
\begin{align*}
f_{\tau_R} (S_r \smallsetminus \g_{[0, \tau_R]}) \, \subset \, \ad{\m H}  \smallsetminus \ad{\text{Cone} (\t)} \, = \, \{z \in \ad{\m H} \smallsetminus\{0\} \,| \, \arg z < \theta \text{ or } \arg z > \pi - \t \} ,
\end{align*}
where $f_{\tau_R} \colon \m H \smallsetminus \g_{[0, \tau_R]} \to \m H$ is a conformal map 
fixing $\infty$ such that $f_{\tau_R}(\g_{\tau_R}) = 0$. 
Let $\omega_{\g, z}$ be the harmonic measure on the boundary of $\m H \smallsetminus \g_{[0, \tau_R]}$ seen from a point $z \in S_r \smallsetminus \g_{[0, \tau_R]}$. Let $\partial_-$ (resp.~$\partial_+$) be the union of $\m R_-$ (resp.~$\m R_+$) with the left side (resp.~right side) of $\g_{[0, \tau_R]}$.
Let $\m P_z$ be the law of a planar Brownian motion 
$B^{(2)}$
starting from $z$. Then, we have
\begin{align*}
 \max \{\omega_{\g,z} (\partial_-), \omega_{\g,z} (\partial_+)\} \ge \m P_z [B^{(2)} \text{ hits } \m R \text{ before } S_R],
\end{align*}
which implies that 
\begin{align*}
\inf_{\g \in \mc X} \inf_{z \in S_r \smallsetminus \g_{[0, \tau_R]}} \max \{\omega_{\g,z} (\partial_-), \omega_{\g,z} (\partial_+)\} \; \xrightarrow[]{R \to \infty} \; 1.
\end{align*}
Since harmonic measure is preserved under conformal maps, $\arg f_{\tau_R} (z)$ is determined by $\omega_{\g,z} (\partial_-)$,
and we obtain
\begin{align}\label{eq:theta_zero}
 \lim_{R \to \infty} \sup_{\g \in \mc X} \t_R (\g) = 0. 
\end{align}
Now, note that if $\g_{[\tau_R,\infty)} \cap S_r \neq \emptyset$, then $f_{\tau_R}(\g)$ exits Cone$(\t_R(\g))$ 
(see Figure~\ref{fig:theta_R}). 
In particular, Lemma~\ref{lem:perp} 
from Section~\ref{sec:chordal_single_I}
shows that
$I(f_{\tau_R} (\g_{[\tau_R,\infty)})) \ge - 8 \log \sin (\t_R(\g))$. The bound~\eqref{eq:theta_zero} then implies that
\begin{align*}
\lim_{R \to \infty} \inf_\g I (\g) \ge \lim_{R \to \infty}  \inf_\g I(f_{\tau_R} (\g_{[\tau_R, \infty)}))= \infty,
\end{align*}
where the infimums are taken over all $\g \in \mc X$ such that $\g_{[\tau_R,\infty)} \cap S_r \neq \emptyset$.  
Therefore, we can find $R >0$ large enough satisfying both claims~\ref{it:coming_back_bound_deterministic} and~\ref{it:SLE_bound} in the statement.
\end{proof}

Now we state the main lemma that is crucial for the upper bound~\eqref{eq:single_whole_upper}.

\begin{lem}\label{lem:upper_bound}
Let $0 \le M <\infty$ and $\d >0$. 
Let $K \in \mc C$ such that $\tilde I(K) \ge M$.
Let $r$ and $\vare$ be chosen as in Lemma~\ref{lem:K_lowerbound_energy}.
Then, we have 
 \begin{align} \label{eq:single_whole_upper_bound}
\limsup_{\k \to 0+} \k \log \m P^\k \big[ \g^\k \cap \bigdisc_r \in \ad{\mc B}^h_{\vare} (K \cap \bigdisc_r) \big] \le - M + \d.
 \end{align}
\end{lem}

\begin{proof}
First, we choose $R$ as in Proposition~\ref{prop:coming_back_bound}.
For notational simplicity, we denote $T : = T_R$, 
and we recall that $T \geq \tau_R$. We have
\begin{align*}
    \m P^\k \big[\g^\k \cap \bigdisc_r \in \ad{\mc B}^h_{\vare} (K \cap \bigdisc_r) \big]
 \le \; & \m P^\k \big[\g^\k_{[0,T]} \cap \bigdisc_r \in \ad{\mc B}^h_{\vare} (K \cap \bigdisc_r) \big] + \m P^\k \big[\g^\k_{[\tau_R, \infty)} \cap S_r \neq \emptyset \big] .
\end{align*}
Thus, the left-hand side of~\eqref{eq:single_whole_upper_bound} is bounded from above
by the maximum of the two terms
\begin{align*}
& \limsup_{\k \to 0+} \k \log \m P^\k \big[\g^\k_{[0,T]} \cap \bigdisc_r \in \ad{\mc B}^h_{\vare} (K \cap \bigdisc_r)\big] , \\
& \limsup_{\k \to 0+} \k \log \m P^\k [\g^\k_{[\tau_R, \infty)} \cap S_r \neq \emptyset \big] \le -M ,
\end{align*}
where we used~\ref{it:SLE_bound} from Proposition~\ref{prop:coming_back_bound} to immediately bound the second term. 
It remains to bound the first term. 
Since the restriction map $K \mapsto K \cap \bigdisc_r$ is continuous $\mc C \to \mc C$, 
\begin{align*}
F_T : = \{K_T \in \mc K_T \;|\; K_T \cap \bigdisc_r \in \ad{\mc B}^h_{\vare} (K \cap \bigdisc_r)\}
\end{align*}
is Hausdorff-closed in $\mc K_T$.  
Hence, 
Proposition~\ref{prop:LDP_finite_time} shows that 
\begin{align} \label{eq:first_term_bound}
\limsup_{\k \to 0+} \k \log \m P^\k \big[\g^\k_{[0,T]} \cap \bigdisc_r \in \ad{\mc B}^h_{\vare} (K \cap \bigdisc_r)\big] \le - \inf_{K_T \in F_T} I_T (K_T).
\end{align}
Now, we show that $I_T(K_T) \geq M - \d$ for all $K_T \in F_T$. 
In fact,
if there exists $K_T = \g_{[0,T]} \in F_T$ such that $I_T (\g_{[0,T]}) < M$,
then by Proposition~\ref{prop:coming_back_bound}, we know that 
$\g$ does not come back to $S_r$ after hitting $S_R$. 
Let $W$ be the driving function of $\g_{[0,T]}$
and let $\hat \g \in \mc X$ be the chord
driven by  $t \mapsto W_{\min(t,\tau_R)}$, namely, 
the one coinciding with $\g$ up to the hitting time $\tau_R$ and then continued with the hyperbolic geodesic in the slit domain $(\m H \smallsetminus \g_{[0,\tau_R]}; \g_{\tau_R}, \infty)$.
Then, as $\hat \g$ does not come back to $\bigdisc_r$ either, we still have $\hat \g_{[0,T]} \in F_T$, so
Lemma~\ref{lem:K_lowerbound_energy} shows that
\begin{align*}
M - \d \le \tilde I(\hat \g_{[0,T]}) = I (\hat \g) =  I_{\tau_R} (\g) \le I_T (\g) .
\end{align*}
We conclude that~\eqref{eq:first_term_bound} is bounded from above by $- M + \d$. 
This finishes the proof.
\end{proof}

We now state the main lemma for the lower bound~\eqref{eq:single_whole_lower}.
\begin{lem}\label{lem:open_lower_bound}
 Let $\g \in \mc X$ be such that $I(\g) = M < \infty$. Then, for all $\vare > 0$, we have
 \begin{align} \label{eq:single_whole_lower_bound}
 \liminf_{\k \to 0+} \k \log \m P^\k [\g^\k \in \mc B^h_{\vare}(\g)] \ge - M.
 \end{align}
\end{lem}
\begin{proof} 
According to Definition~\ref{defn: Hausdorff}, 
for any $\d > 0$, $N_\d(\g) := \{z \in \ad{\m H}\;|\; \text{dist}_{\ad {\m D}} (\confmap(z), \confmap(\g)) \ge \d\}$
is a bounded set in $\ad{\m H}$, where $\text{dist}_{\ad {\m D}}$ is the Euclidean distance and $\confmap \colon \m H \to \m D$ is the fixed uniformizing conformal map 
from  Section~\ref{subsec:topologies}.
We claim that, for every $\vare > 0$, there exists $\d >0$ 
such that for all $\tilde \g \in \mc X$ with 
$\tilde \g \cap N_\d (\g) = \emptyset$, we have $\tilde \g \in \mc B^h_{\vare} (\g)$. 
Indeed, assuming the opposite, 
let $\vare > 0$ and
$\d_{j} \to 0$ as $j \to \infty$, and let
$\g^j \in \mc X$  such that $\g^j \cap N_{\d_j} (\g) = \emptyset$ while remaining at $\vare$-Hausdorff distance away from $\g$. 
Since $\mc C$ is a compact space, we can extract a subsequence of $\g^j$ that converges to a limit $\tilde K$ in $\mc C$. 
Furthermore, since $\d_{j} \to 0$, we have $\tilde K \subset \g$. 
But $\tilde K$ also connects $0$ to $\infty$ and $\g$ is a chord, so we actually have $\tilde K = \g$, 
contradicting with the assumption that $\tilde K \notin \mc B^h_\vare(\g)$.

Now, fix $\delta = \delta(\vare)$ as above, and 
let $r > 0$ be such that $N_\d(\g) \subset \bigdisc_r$.
Choose $R > r$ as in Proposition~\ref{prop:coming_back_bound} corresponding to $M+1$, 
and let $T > T_R \ge \tau_R$.
Then, we have
\begin{align*}
 \m P^\k \big[ \g^\k \in \mc B^h_{\vare}(\g)\big] 
 & \ge \m P^\k[\g^\k \cap N_\d(\g) = \emptyset] \\
& \ge
\m P^\k \big[ \g_{[0,T]}^\k \subset \ad{\m H} \smallsetminus N_\d(\g) \big] 
- \m P^\k \big[\g_{[\tau_R,\infty)}^\k \cap N_\d(\g) \neq \emptyset
\big]. 
\end{align*}
Now, Proposition~\ref{prop:LDP_finite_time} shows that 
\begin{align} \label{eq: bound1}
\liminf_{\k \to 0+} \k \log \m P^\k \big[ \g_{[0,T]}^\k \subset \ad{\m H} \smallsetminus N_\d(\g) \big]  \ge - \inf_{K_T}
I_T(K_T) \ge - I_T (\g) \ge - M  ,
\end{align}
where the infimum is taken over $\{K_T \in \mc K_T \;|\; K_T \subset \ad {\m H} \smallsetminus N_\d (\g) \}$.
Furthermore, Proposition~\ref{prop:coming_back_bound} shows that
\begin{align} \label{eq: bound2}
\limsup_{\k \to 0+} \k \log \m P^\k \big[\g_{[\tau_R,\infty)}^\k \cap N_\d(\g) \neq \emptyset
\big] \le - M - 1 .
\end{align}
The bounds~\eqref{eq: bound1} and~\eqref{eq: bound2}
together
imply that
the left-hand side of~\eqref{eq:single_whole_lower_bound} is bounded by 
$- \min (I_T (\g), M)$ from below, which tends to $-M$ as $T \to \infty$.
\end{proof}

With these preparations, we are now ready to prove the main result of this section.

\begin{proof}[Proof of Theorem~\ref{thm: SLE large deviation}]
First, we prove the lower bound~\eqref{eq:single_whole_lower}. 
Without loss of generality, we assume that $M_\open : = \inf_{K \in \open} I (K) < \infty$. 
For $\d >0$, take a chord $\g \in \open \cap \mc X$ such that 
$I (\g) \le  M_\open + \d$. 
Since $\open$ is
open,  there exists $\vare > 0$ such that $\mc B^h_{\vare} (\g) \subset \open$. 
Lemma~\ref{lem:open_lower_bound} then shows that
\begin{align*}
\liminf_{\k \to 0+} \k \log \m P^\k \big[ \g^\k \in \open \big] \ge \liminf_{\k \to 0+} \k \log \m P^\k \big[ \g^\k \in \mc B^h_{\vare}(\g)\big] \ge  - M_\open - \d ,
\end{align*}
and we obtain the lower bound~\eqref{eq:single_whole_lower} after letting $\d \searrow 0$.

Second, we prove the upper bound~\eqref{eq:single_whole_upper}. 
Because $\SLE_\k$ curves belong almost surely to 
$\mc X$, we assume  without loss of generality that $\closed$ 
lies inside the closure of $\mc X$ in $\mc C$.
Then, every element of $\closed$ connects $0$ to $\infty$ 
in $\ad{\m H}$, 
so 
\begin{align*}
M_\closed : = \inf_{K \in \closed} I (K) = \inf_{K \in \closed} \tilde I (K) \in [0, \infty] .
\end{align*}
Now, let $M < M_F$ and $\d >0$. 
Since for each $K \in \closed$, we have $\tilde I(K) > M$,  we may choose
a neighborhood $\mc B^h_{\vare}(K)$ for $K$ and $r > 0$ according to Lemma~\ref{lem:upper_bound} for the given values of $M$ and $\d$. In particular, we have
\begin{align*}
\limsup_{\k \to 0+} \k \log \m P^\k [\g^\k \in  \mc B^h_{\vare}(K)] 
\le  \limsup_{\k \to 0+} \k \log \m P^\k [\g^\k \cap \bigdisc_r \in \ad{\mc B}^h_{\vare}(K \cap \bigdisc_r)] \le - M + \d. 
\end{align*}
Since $F$ is compact in $\mc C$ and can be covered by finitely many such balls, we see that
\begin{align*}
\limsup_{\k \to 0+} \k \log \m P^\k [\g^\k \in \closed] \le  - M + \d .
\end{align*} 
Finally, after letting $\d \searrow 0$ and then $M \nearrow M_\closed$, we obtain the upper bound~\eqref{eq:single_whole_upper}.

Remark~\ref{rem:goodness} gives the goodness of the rate function, which concludes the proof.
\end{proof}

\subsection{Large deviations of multichordal SLE}
\label{subsec: LDP multiple}

In this section, we prove the LDP for multiple chordal $\SLE_\k$ curves (Theorem~\ref{thm:main_LDP}) in $\m H$.
We fix an $\np$-link pattern $\a$ and 
boundary points $x_1 < \cdots < x_{2\np}$.
For brevity, we denote the curve space by 
$\mc X_\a := \mc X_\a(\m H;x_1, \ldots, x_{2\np})$.
Recall that this is the space of  
disjoint chords in $\m H$ which connect 
the marked boundary points 
$x_{a_j}$ and $x_{b_j}$ for all $j$ according to the link pattern $\a$.
As before, we use the shorthand notation $\ad \g : = (\g_1, \ldots, \g_{\np})$ for a multichord, and similarly $\ad \eta$ for the unique geodesic multichord in $\mc X_\a$ (cf. Corollary~\ref{cor:geod_unique}).

Let $\mathbb{Q}_\a^\k$ denote 
the product measure of $\np$ independent chordal $\SLE_\k$ curves 
associated to the link pattern $\a$, (i.e., connecting the marked boundary points according to $\a$, but not necessarily disjoint).
On the product space  $\prod_j \mc X (\m H; x_{a_j}, x_{b_j})$
endowed with the product topology induced from the Hausdorff metric on $\mc C$,
the probability measures $(\mathbb{Q}^\k_\a)_{\k > 0}$
satisfy a LDP  with rate function just the sum of the independent rate functions, 
\begin{align*}
I_0^\a (\ad \g) := & \;
\sum_{j=1}^\np I_{\m H; x_{a_j}, x_{b_j} }(\g_j) .
\end{align*}

\begin{lem} \label{lem:indep_SLE_LDP}
The family of laws $(\mathbb{Q}_\a^\k)_{\k > 0}$ of 
$\np$ independent chordal $\SLE_\k$ curves $\ad \g^\k$
satisfies the LDP in $\prod_j \mc X (\m H; x_{a_j}, x_{b_j})$ with good rate function $I_0^\a$.
\end{lem}
\begin{proof}
This follows from Theorem~\ref{thm: SLE large deviation}, as $\g_1^\k , \ldots, \g_\np^\k$ are independent.
\end{proof}

Multichordal $\SLE_\k$
is a family of $\SLE_\k$ curves with interaction. 
In this section, we shall derive a LDP for them with rate function including contribution from the interaction, namely from
the loop measure term introduced in Section~\ref{subsec:a_loop_measure}.
This will be a simple consequence of Varadhan's lemma (stated as Lemma~\ref{lem:Varadhan} below).

We denote by $\mathbb{E}_\a^\k$ the expectation with respect to $\mathbb{Q}_\a^\k$.
Then, as proved, e.g., in~\cite[Thm.\,1.2]{BPW}, 
the multichordal $\SLE_\k$ probability measure $\m P_\a^\k$ on $\mc X_\a$
can be obtained by weighting 
$\mathbb{Q}_\a^\k$ with the Radon-Nikodym 
derivative\footnote{The difference of $1/2$ compared to~\cite{PW19,BPW} is due to normalization conventions for Brownian loop measure.}
\begin{align} \label{eq: multiple_SLE_RN}
R_\a^\k (\ad \g^\k) = \; & 
\frac{\ud \m P_\a^\k}{\ud \mathbb{Q}_\a^\k} (\ad \g^\k) 
:= 
\frac{\exp \Big(\frac{1}{\k} \Phi_\k(\ad \g^\k)\Big)}{\mathbb{E}_\a^\k \Big[ \exp\Big(\frac{1}{\k} \Phi_\k(\ad \g^\k)\Big) \Big]} , \quad \text{where } \quad \Phi_\k(\ad \g) 
:= \frac{\k}{2} c(\k) \, m_{\m H}(\ad \g) 
\end{align}
and $c(\k) = (3\k-8)(6-\k)/2\k.$
Note that when $\k < 8/3$, we have $c(\k) < 0$, so the normalization factor (total mass) is clearly finite.
We also set $\Phi_\k(\ad \g) := -\infty$ if $\ad \g \notin \mc X_\a$.
For all $\ad \g \in \mc X_\a$, as $\k \searrow 0$, we have
\begin{align} \label{eq:Phi_converge}
\Phi_\k(\ad \g) = - \frac{c(\k) \k }{24}  \Phi_0(\ad \g) 
\quad \searrow \quad
\Phi_0(\ad \g) = -12 m_{\m H}(\ad \g) < 0,
\end{align}
since $ - c(\k) \k  \nearrow 24$.
This factor gives the additional contribution to the LDP rate function. 

\begin{lemA} \label{lem:Varadhan}
\textnormal{(Varadhan's lemma~\cite[Lem.\,4.3.4~and~4.3.6]{DZ10})}

Suppose that the probability measures $(\mathbb{Q}^\k)_{\k >0}$ 
satisfy a LDP with good rate function $I_0^\a$.
Let $\Phi \colon \prod_j \mc X (\m H; x_{a_j}, x_{b_j}) \to \m R$ be a function bounded from above.
Then, the following hold.
\begin{enumerate} 
\item \label{item1_Var}
If $\Phi$ is upper semicontinuous, then for any closed subset $\closed$ of $\prod_j \mc X (\m H; x_{a_j}, x_{b_j})$,
\begin{align*}
 \limsup_{\k \to 0} \k \log \mathbb{E}^\k \bigg[ \exp\bigg(\frac{1}{\k} \Phi(\ad \g^\k)\bigg)\1{\{\ad \g^\k \in \closed\}} \bigg]
\leq 
 \; &
- \inf_{\ad \g\in \closed}\,( I_0^\a(\ad \g) - \Phi(\ad \g)).
\end{align*}

\item \label{item2_Var} 
If $\Phi$ is lower semicontinuous, then for any open subset $\open$ of $\prod_j \mc X (\m H; x_{a_j}, x_{b_j})$,
\begin{align*}
\liminf_{\k \to 0} \k \log \mathbb{E}^\k \bigg[ \exp\bigg(\frac{1}{\k} \Phi(\ad \g^\k)\bigg)\1{\{\ad \g^\k \in \open\}} \bigg]
\geq  \; & - \inf_{\ad \g\in \open}\,( I_0^\a(\ad \g) - \Phi(\ad \g)).
\end{align*}
\end{enumerate}
\end{lemA}

Using these results, it is straightforward to derive the LDP for multichordal $\SLE_{0+}$. 
Recalling Definition~\ref{df:multichord_energy}, 
we denote  the multichordal Loewner energy in $(\m H; x_1, \ldots, x_{2\np})$ by 
\begin{align*}
I^\a(\ad \g)  :=  I_{\m H}^\a(\ad \g) 
& = 12 \, \big( \mc H_{\m H}(\ad \g) - \Hmin_{\m H}^\a (x_1, \ldots,x_{2\np}) \big)  = I_0^\a(\ad \g) - \Phi_0 (\ad \g) - L,
\end{align*}
where
\begin{align*}
L := I_0^\a(\ad \eta) - \Phi_0 (\ad \eta) =  \inf_{\ad \g \in \mc X_\a} (I_0^\a(\ad \g) - \Phi_0(\ad \g)) .
\end{align*}
Theorem~\ref{thm:main_LDP} follows from the next result.

\begin{thm}\label{thm:LDP_sec}
The family of laws $(\m P_\a^\k)_{\k > 0}$ of 
multichordal $\SLE_\k$ curves $\ad \g^\k$
satisfies the following LDP in $\mc X_\a$ with good rate function $I^\a$\textnormal{:}

For any closed subset $\closed$ and open subset $\open$ of $\mc X_\a$, we have
\begin{align}
\limsup_{\k \to 0+} \k \log  \m P_\a^\k [ \ad \g^\k \in \closed ] 
\leq \; & - \inf_{\ad \g \in \closed}  I^\a(\ad \g), \label{eqn:closed_bound} \\
\liminf_{\k \to 0+} \k \log  \m P_\a^\k [ \ad \g^\k \in \open ] 
\geq \; & - \inf_{\ad \g \in \open}  I^\a(\ad \g) \label{eqn:open_bound}.
\end{align}
Furthermore, we have
\begin{align} \label{eq: total_mass_limit}
\lim_{\k \to 0+} \k \log \mathbb{E}_\a^\k \bigg[ \exp\bigg(\frac{1}{2}c(\k) \, m_{\m H}(\ad \g^\k)\bigg) \bigg]
= - L.
\end{align}
\end{thm}

\begin{proof}
The lower semicontinuity and compactness of the level sets of $I^\a$ follow from Proposition~\ref{prop: H_semicontinuous}.
Recall that the multichordal SLE$_\k$ takes values in $\mc X_\a$ almost surely.
Using the Radon-Nikodym derivative~\eqref{eq: multiple_SLE_RN},
for any Borel set $B \subset \mc X_\a$, we have
\begin{align*}
\k \log \m P_\a^\k[\ad \g^\k \in B]  
= \; & \k \log \mathbb{E}_\a^\k \Big[ \exp\Big(\frac{1}{\k} \Phi_\k(\ad \g^\k)\Big)\1{\{\ad \g^\k \in B\}} \Big]
- \k \log \mathbb{E}_\a^\k \Big[ \exp\Big(\frac{1}{\k} \Phi_\k(\ad \g^\k)\Big) \Big] .
\end{align*}
Equation~\eqref{eq: total_mass_limit} asserts that 
the second term  on the right-hand side has the limit $L$ as $\k \to 0+$.
We will establish this together with~(\ref{eqn:closed_bound},~\ref{eqn:open_bound}).

\noindent {\bf Upper bound.} 
Let $\tilde \closed$ be the closure of $\closed$ in $\prod_j \mc X (\m H; x_{a_j}, x_{b_j})$.
Let $\vare, M > 0$. By~\eqref{eq:Phi_converge},
there exists $\k_0 \in (0,8/3)$ such that we have
$\Phi_\k(\ad \g) < (1 - \vare) \Phi_0 (\ad \g)$ for all $\ad \g$ and for all $\k \in [0,\k_0)$. 
Applying Item~\ref{item1_Var} of Varadhan's Lemma~\ref{lem:Varadhan} to $\tilde \closed$  
and the continuous function 
$\Phi^{M,\vare} (\ad \g) := \max \{ (1-\vare)\Phi_0 (\ad \g), -M \}$ 
(which is continuous by Lemma~\ref{lem: malpha_continuous}),
we obtain
\begin{align*}
& \limsup_{\k \to 0} \k \log \mathbb{E}^\k \bigg[ \exp\bigg(\frac{1}{\k} \Phi_\k(\ad \g^\k)\bigg) \1{\{\ad \g^\k \in \closed\}} \bigg]\\
 \leq \; & \limsup_{\k \to 0} \k \log \mathbb{E}^\k \bigg[ \exp\bigg(\frac{1}{\k} \Phi_\k(\ad \g^\k)\bigg) \1{\{\ad \g^\k \in \tilde \closed\}} \bigg] \\
\leq  \; & \limsup_{\k \to 0} \k \log \mathbb{E}^\k \bigg[ \exp\bigg(\frac{1 }{\k} \Phi^{M,\vare}(\ad \g^\k)  \bigg) \1{\{\ad \g^\k \in \tilde \closed\}} \bigg] \\
\leq \; & -  \inf_{\ad \g \in \tilde \closed} (I_0^\a(\ad \g) - \Phi^{M,\vare}(\ad \g)) 
\qquad \xrightarrow[\vare \searrow 0]{M \nearrow \infty} \qquad 
-  \inf_{\ad \g \in F} (I_0^\a(\ad \g) - \Phi_0(\ad \g)),
\end{align*}
since we have  $- \Phi_0(\ad \g) = \infty$ for $\ad \g \in \tilde \closed \smallsetminus \closed$.
Also, taking $\closed = \mc X_\a$
shows that 
\begin{align} \label{eq:limsup}
\limsup_{\k \to 0} \k \log \mathbb{E}_\a^\k \bigg[ \exp\bigg(\frac{1}{\k} \Phi_\k(\ad \g)\bigg) \bigg]
\leq \; & -L .
\end{align}

\noindent {\bf Lower bound.} 
Without loss of generality, we assume $M_\open := \inf_{\ad \g \in \open} (I_0^\a(\ad \g) - \Phi_0(\ad \g)) <\infty$. Let $\vare>0$ and $\ad \g_\vare \in \open$ such that $I_0^\a(\ad \g_\vare) - \Phi_0(\ad \g_\vare) \le M_\open + \vare$, in particular, $\Phi_0(\ad \g_\vare) > -\infty$. By continuity of $\Phi_0$, let $\mc B \subset \open$ be an open neighborhood of $\ad \g_\vare$ in $\prod_j \mc X(\m H; x_{a_j}, x_{b_j})$  such that $\Phi_0 \ge \Phi_0(\ad \g_\vare) - 1$ on $\mc B$.
Let $\Phi$ be the lower semicontinuous function that equals $\Phi_0$ on $\mc B$ and $\Phi_0(\ad \g_\vare) - 2$ otherwise.
Since $\Phi_0 \le \Phi_\k \le 0$, 
 applying Item~\ref{item2_Var} of Varadhan's Lemma~\ref{lem:Varadhan} to $\mc B$ and the lower semicontinuous function $\Phi$,
we obtain
\begin{align*}
& \liminf_{\k \to 0} \k \log \mathbb{E}^\k \bigg[ \exp\bigg(\frac{1}{\k} \Phi_\k(\ad \g^\k)\bigg) \1{\{\ad \g^\k \in \open\}} \bigg] \\
\ge \: & \liminf_{\k \to 0} \k \log \mathbb{E}^\k \bigg[ \exp\bigg(\frac{1}{\k} \Phi_0(\ad \g^\k)\bigg) \1{\{\ad \g^\k \in \mc B\}} \bigg]\\
\geq \; & - \inf_{\ad \g \in \mc B} (I_0^\a(\ad \g) - \Phi (\ad \g)) 
\geq  - (I_0^\a(\ad \g_\vare) - \Phi_0 (\ad \g_\vare)) \\
\ge \: &  - M_\open - \vare \qquad \xrightarrow[]{\vare \searrow 0} 
\qquad - M_\open.
\end{align*}
Taking $\open = \mc X_\a$ shows that 
\begin{align} \label{eq:liminf}
\liminf_{\k \to 0} \k \log \mathbb{E}_\a^\k \bigg[ \exp\bigg(\frac{1}{\k} \Phi_\k(\ad \g^\k)\bigg) \bigg]
\geq \; & - L .
\end{align}
Combining~(\ref{eq:limsup},~\ref{eq:liminf}) yields~\eqref{eq: total_mass_limit}.
The above upper bound and lower bound then give~\eqref{eqn:closed_bound} and~\eqref{eqn:open_bound}, since 
$I_0^\a(\ad \g) - \Phi_0 (\ad \g) - L = I^\a (\ad \g)$. 
This concludes the proof.
\end{proof}

The LDP immediately implies the following corollary (that we state in the general domain $\domain$ as in the introduction).

\mainlimitcor*

\begin{proof}
Recall from Corollary~\ref{cor:geod_unique}
that the geodesic multichord  $\bar{\eta}$ is the unique minimizer of $I_\domain^\a$.
Let $\mc B^h_\vare(\bar{\eta}) \subset \mc X_\a(\domain) : = \mc X_\a(\domain;x_1,\ldots,x_{2\np})$ 
be a Hausdorff-open ball of radius $\vare$ around $\bar{\eta}$.
Then, we have 
\begin{align*}
\limsup_{\k \to 0+} \k \log \m P^\k [\ad \g^\k \in \mc X_\a(\domain) \smallsetminus \mc B^h_\vare(\bar{\eta})] 
\leq - \inf_{\ad \g \in \mc X_\a(\domain) \smallsetminus \mc B^h_\vare(\bar{\eta})} I_\domain^\a(\ad \g) < 0,
\end{align*}
by Theorem~\ref{thm:LDP_sec}. This proves the corollary. 
\end{proof}

Lastly, we state another corollary. The normalization factor in~\eqref{eq: multiple_SLE_RN} 
determines the multichordal $\SLE_\k$ \emph{partition function} $\PartF_\a$ associated to $\a$,
\begin{align} \label{eqn:pf_def}
\PartF_\a(\domain;x_1, \ldots,x_{2\np}) := 
\bigg( \prod_{j=1}^\np P_{\domain;x_{a_j},x_{b_j}} \bigg)^{(6-\k)/2\k} \; \mathbb{E}_\a^\k \bigg[
\exp\bigg(\frac{1}{2} c(\k) \, m_\domain(\ad \g^\k)\bigg) \bigg] .
\end{align}
Note that $\PartF_\a$ is a function of the boundary points $x_1, \ldots,x_{2\np}$ and the domain $\domain$. 
Equation~\eqref{eq: total_mass_limit}
implies that the minimal potential $\Hmin^\a_\domain$ can be regarded as a semiclassical limit of the multichordal $\SLE_\k$ partition function $\PartF_\a$ in the following sense:

\begin{cor} \label{cor:semicl_lim_of_pf}
We have
\begin{align*}
- \lim_{\k \to 0+} \k \log \PartF_\a(\domain;x_1, \ldots,x_{2\np})
= \; & 12 \, \Hmin^\a_\domain (x_1, \ldots,x_{2\np}) .
\end{align*}
\end{cor}
\begin{proof}
This follows immediately from~\eqref{eq: total_mass_limit}
and definitions~\eqref{eq:initialdef_Hmin} and~\eqref{eqn:pf_def}.
\end{proof}

\section{Determinants of Laplacians and the Loewner potential}
\label{sec:det}

In this section, we show Theorem~\ref{thm:main_H_det}: 
the identity between the Loewner potential $\mc H$ 
and zeta-regularized determinants of Laplacians.  
This relies on the Polyakov-Alvarez conformal anomaly formula for domains with corners (see Theorem~\ref{thm:PA} and Appendix~\ref{app:PA} for a more detailed discussion),
and is applicable only to smooth multichords. 
Throughout this section, we consider bounded curvilinear polygonal domains with piecewise smooth boundary:

\begin{df} \label{def:curvilinear}
We say that $\domain \subsetneq \m C$ is a \emph{curvilinear polygonal domain} if its boundary  $\partial \domain$
is a piecewise smooth Jordan curve with finitely many 
corners $z_1, z_2, \ldots, z_m \in \partial \domain$ with opening interior angles $\pi \b_1, \pi \b_2, \ldots, \pi \b_m$
such that
\begin{enumerate}
\item the boundary $\partial \domain$ in a neighborhood of each corner $z_j$ is defined by a continuous curve 
$\g \colon (-\vare, \vare) \to \m C$ for some $\vare > 0$ such that $\g(0) = z_j$, 
the curve $\g$ is smooth on $(-\vare,0]$ and $[0, \vare)$, with
$|\g'(t)| = 1$ for all $t \in (-\vare, \vare)$, and 
\begin{align*}
\lim_{t \to 0-} \g'(t) \qquad \text{and} \qquad \lim_{t \to 0+} \g'(t)
\end{align*}
are tangent vectors at $\g(0) = z_j$; and

\item for each $j\in\{1,\ldots,\np\}$, 
the opening angle $\pi \b_j$ at the corner at $z_j$ is the interior angle between the tangent vectors
$\lim_{t \to 0-} \g'(t)$ and $\lim_{t \to 0+}\g'(t)$; so $\b_j \in (0,1) \cup (1,2)$.
\end{enumerate}
In other words, $\partial \domain$ is asymptotically straight on each side of the corner point. 
 \end{df}

We call two metrics $g$ and $g_0$ on $\domain$
\emph{Weyl-equivalent} if $g$ is a \emph{Weyl-scaling} of $g_0$,
i.e., we have $g = e^{2\s} g_0$ for some $\s \in C^{\infty} (\ad{\domain})$. Beware that our notion of Weyl-equivalence is not the same as conformal invariance: 
especially, we do not allow any $\log$-singularities of $\s$, which would change the opening angles of the boundary.
Throughout this section, we only consider metrics that  
are Weyl-equivalent to the Euclidean metric $\ud z^2$, so that the opening angle on the boundary is the same as for the Euclidean metric.
We use the following notation with respect to the metric $g$\textnormal{:}
\begin{itemize} \setlength\itemsep{-0.3em}

\item $\D_g : = \D_{\domain;g}$ is the (positive) Laplacian on $\domain$, see~\eqref{eq:Laplacian}, with Dirichlet boundary conditions; 

\item $\ud \vol_g$ is the area measure;

\item $\ud \mathrm l_g$ is the arc-length measure on the boundary;

\item $K_g$ is the Gauss curvature in the interior of $\domain$; and

\item $k_g$ is the geodesic curvature on the boundary $\partial \domain$.
\end{itemize}

\subsection{Determinants of Laplacians}
\label{subsec:Lap_det}

The purpose of this section is to give a brief summary of  zeta-regularized determinants of Laplacians and to state the Polyakov-Alvarez conformal anomaly formula.
Let $\domain \subsetneq \m C$ be a curvilinear polygonal domain having corners at $z_1, \ldots, z_m \in \partial \domain$ with opening interior angles $\pi \b_1, \ldots, \pi \b_m$.

The Dirichlet Laplacian $\D_g$ on $\domain$ has a discrete spectrum. 
We order its eigenvalues as $0 < \l_1 \leq \l_2 \leq \cdots$ and recall that the heat kernel in $\domain$ 
is represented as the series
\begin{align} \label{eq: heat kernel series}
p_t(z,w) := \sum_{ j = 1}^{\infty} e^{-\l_j t} u_j(z) u_j(w) ,
\end{align}
where $u_j$ are the eigenfunctions of $\D_g$ corresponding to $\l_j$,
forming an orthonormal basis for $L^2(\domain)$.
Following Ray \&~Singer \cite{RS71},
a notion of determinant for the Laplace operator $\D_g$ makes sense using its \emph{spectral zeta function}, 
defined in terms of the heat trace 
\begin{align*}
\sum_{j \ge 1} e^{-\l_j t}
= \tr (e^{-t \D_g}) = \tr (p_t)
=  \int_\domain p_t(z,z) \ud \vol_g(z).
\end{align*}
For $\Re(s) > 1$, the spectral zeta function is defined as
\begin{align} \label{eq: spectral zeta function}
\zeta_{\D_g}(s) 
:= \sum_{\l_j > 0} \l_j^{-s} =  \frac{1}{\G(s)} \int_0^{\infty} t^{s-1} \tr (p_t)\, \ud t ,
\end{align}
where $\G(\cdot)$ is the Gamma function.  The spectral zeta function $\zeta_{\D_g}(s)$ 
is a holomorphic function on $\{s \in \m C \; | \; \Re(s) > 1\}$. 
 Using fine estimates of the short-time expansion of the heat kernel in curvilinear domains, only established very recently in~\cite[Thm.\,1.4]{NRS}, one shows that the analytic continuation of $\zeta_{\D_g}$ is holomorphic in a neighborhood of $0$. 
The \emph{zeta-regularized determinant of $\D_g$} is then defined as 
\begin{align*}
\detz \D_g: = \exp \big( - \zeta_{\D_g}'(0) \big).
\end{align*}
In fact, the determinant of an operator is formally the product of its eigenvalues: 
\begin{align*}
\zeta_{\D_g}'(s) = \sum_{\l_j > 0} \log (\l_j) \l_j^{-s}
\qquad \Longrightarrow \qquad
``\zeta'_{\D_g}(0) = - \log \prod_{\l_j > 0} \l_j  = - \log \det \D_g" .
\end{align*}
Importantly, the zeta-regularized determinant of the Laplacian 
depends on the metric.  
The change of
the determinant under a Weyl-scaling is given by the Polyakov-Alvarez conformal anomaly formula. For curvilinear polygonal domains, this formula was proved recently in~\cite{R_inprep}. 

\begin{restatable}{theorem}{thmPA}\textnormal{[Generalized Polyakov-Alvarez conformal anomaly formula \cite[Thm.\,3]{R_inprep}]}
\label{thm:PA}
Consider a metric $g = e^{2\s} g_0$ on a curvilinear polygonal domain $\domain$ which is Weyl-equivalent to a reference metric $g_0$.
Then, we have
\begin{align}
\begin{split} \label{eq: Generalized Polyakov-Alvarez}
 & \log \detz (\D_0)-\log \detz (\D_{g})  \\
= & \; \frac{1}{6\pi} \left[ \frac{1}{2} \int_{\domain} \abs{\nabla_0 \s}^2 \ud \vol_0  + \int_\domain K_0 \s \ud \vol_0 + \int_{\partial \domain \smallsetminus \{z_1, \ldots, z_m\}} k_0 \s \ud \mathrm{l}_0  \right] \\
 \;& + \frac{1}{4\pi} \int_{\partial \domain \smallsetminus \{z_1, \ldots, z_m\}} \partial_{\nu_0} \s \ud \mathrm{l}_0 
 + \frac{1}{12} \sum_{j=1}^m  \left( \frac{1}{\b_j} - \b_j \right) \s(z_j) ,
\end{split}
\end{align}
where $\partial_{\nu_0}$ is the outward normal derivative with respect to the metric $g_0$,
and for notational simplicity, we replace the subscripts ``${g_0}$'' by ``$0$''.
\end{restatable}

For the readers' convenience, we outline the key steps in the proof and discuss some heuristics about this formula in Appendix~\ref{app:PA}.

\subsection{Identity with the Loewner potential}
\label{subsec:Det_identity_proof}

In this section, we consider a 
curvilinear polygonal
domain
$\domain$ with marked boundary points $x_1, \ldots, x_{2\np} \in \partial \domain$ on smooth boundary segments. 
We also fix an $\np$-link pattern $\a$ throughout. 
We say that a multichord $\ad \g \in \mc X_{\a}(\domain;x_1, \ldots, x_{2\np})$ is \emph{smooth}, if each $\g_j$ 
is the image of an injective $C^{\infty}$-function of $[-1,1]$, with $j=1,\ldots,\np$. 
We let $\mc X_{\a}^\infty(\domain) := \mc X_{\a}^\infty(\domain;x_1, \ldots, x_{2\np})$ be 
the space of smooth 
finite-energy multichords (dropping the notation $\a$ when $\np = 1$). 
According to Lemmas~\ref{lem:perp} and~\ref{lem:H_finite}, each chord of
$\ad \g \in \mc X_{\a}^\infty(\domain)$ meets $\partial \domain$ perpendicularly, and the connected components of $\domain \smallsetminus \ad \g$ are curvilinear polygonal domains.
We define
\begin{align} \label{eq:def_tilde_H}
\tilde{\mc{H}}_\domain (\ad \g ;g) := \log\detz \D_{\domain;g} - 
\sum_C \log\detz \D_{C;g}    
\end{align}
for each $\ad \g  \in \mc X_{\a}^\infty(\domain)$ and for each metric $g$ on $\domain$, 
where the sum is taken over all connected components $C$  of $\domain \smallsetminus \bigcup_{i} \g_i$.
We also define 
\begin{align*}
\tilde{\Hmin}_\domain^\a (x_1, \ldots,x_{2\np};g)
:= \inf_{\ad{\g}} \tilde{\mc{H}}_\domain (\ad \g ;g) ,
\end{align*}
where the infimum is taken over all $\ad \g  \in \mc X_\a^\infty(\domain; x_1, \ldots, x_{2\np})$.

\bigskip

The goal of this section is to prove the following result, equivalent to Theorem~\ref{thm:main_H_det}.

\begin{thm}\label{thm:H_H}
There exists a universal constant $\cst \in \m R$ such that for all $n \ge 1$ and for all 
$\ad \g  \in \mc X_{\a}^\infty(\domain)$, we have
\begin{align*}
\tilde {\mc H}_\domain (\ad \g  ; \ud z^2) = \mc H_\domain (\ad \g ) + \np \cst.
\end{align*}
\end{thm}

\begin{remark}
To determine the numerical value of  $\cst$, we apply Theorem~\ref{thm:H_H} to the chord $[-1,1]$ in $\m D$.
The determinant of the Laplacian on the flat unit disc with Dirichlet boundary conditions
was evaluated in~\cite[Eq.\,(28)]{Weisberger}: 
\begin{align*}
\log\detz \D_{\m D; \ud z^2} = - \frac{1}{6} \log 2 - \frac{1}{2} \log \pi - 2 \zeta_R'(-1) - \frac{5}{12} 
,
\end{align*}
where $\zeta_R$ is the Riemann zeta function.
From~\cite[Sect.\,3]{R_inprep}, we obtain the determinant of the Laplacian on the flat half-unit disc $\m D_+ : = \{z \in \m D \;|\; \Im z >0\}$ with Dirichlet boundary conditions\footnote{This formula appears to disagree with the corresponding formula in~\cite[Eq.\,(28)]{Dowker},  claiming that \\ $\log\detz \D_{\m D_+;\,\ud z^2} = \frac{1}{12} - \zeta_R'(-1) - \frac{2}{3} \log 2 - \frac{1}{2} \log \pi$.}:
\begin{align*}
\log\detz \D_{\m D_+;\,\ud z^2} = - \frac{5}{24} - \zeta_R'(-1) - \frac{1}{3} \log 2 - \frac{1}{2} \log \pi .
\end{align*}
Using these formulas, we obtain
\begin{align*}
\tilde {\mc H}_{\m D} ([-1,1], \ud z^2) = \log\detz \D_{\m D; \ud z^2} - 2 \log\detz \D_{\m D_+; \, \ud z^2}  =  \frac{1}{2} \log 2 + \frac{1}{2} \log \pi .
\end{align*}
Also, using the conformal map $\varphi (z) = \ii (\ii - z)/(\ii +z)$ from $\m D$ to $\m H$, with $|\varphi'(-1)| = |\varphi'(1)| = 1$, we see that
\begin{align*}
\mc H_{\m D} ([-1,1]) = - \frac{1}{4}\log P_{\m D; -1,1} = - \frac{1}{4} \log \frac{|\varphi'(1) \varphi'(-1)|}{(\varphi(1) - \varphi(-1))^2} = \frac{1}{2} \log 2 ,
\end{align*}
and combining everything, we evaluate the constant $\cst$ as
\begin{align*}
\cst = \tilde {\mc H}_{\m D} ([-1,1], \ud z^2) - \mc H_{\m D} ([-1,1])  = \frac{1}{2} \log \pi \approx 0.5724.
\end{align*}
 The constant $\cst$ is also  given by Corollary~\ref{cor:M_poisson_kernel}.
\end{remark}

To prove Theorem~\ref{thm:H_H}, we first consider the case of a single chord, and 
define, for $\g \in \mc X^{\infty} (\domain; \bpt, \ept)$,
\begin{align} \label{eq:def_J}
J_{\domain}(\g; g) 
:= \; & 12 \big( \tilde{\mc{H}}_\domain (\g ;g) - \tilde{\Hmin}^\a_\domain (\bpt,\ept ;g) \big).
\end{align}
Recall that here $(\domain; \bpt, \ept)$ is a curvilinear domain with boundary points $\bpt, \ept \in \partial \domain$ on smooth boundary segments, that is, there are smooth neighborhoods $U_\bpt \ni \bpt$ and $U_\ept \ni \ept$ on $\partial \domain$. 
If $(\domain'; \bpt', \ept')$ is another such domain, 
then a conformal map $\varphi$ from $(\domain; \bpt, \ept)$ to $(\domain'; \bpt', \ept')$ may not map corners of $\domain$ to corners of $\domain'$. Therefore, $\sigma = \log \abs{\varphi'}$ is in general only smooth on $\domain \cup U_\bpt \cup U_\ept$ with possibly  logarithmic singularities elsewhere on $\partial \domain$. 

\begin{lem} \label{lem:J_conformal_invariance}
Let $\s \colon \overline{\domain} \to \m R$ be a smooth function 
and $g = e^{2\s} \ud z^2$ a metric  Weyl-equivalent to $\ud z^2$.
Then, for all $\g\in \mc X^{\infty} (\domain;\bpt,\ept)$, we have
\begin{align}\label{eq:var_H}
\tilde {\mc H}_{\domain} (\g; g) - \tilde {\mc H}_{\domain} (\g; \ud z^2)
= \frac{1}{4} (\s (\bpt) + \s(\ept))  .
\end{align} 
In particular, $J$ is Weyl-invariant, i.e., $J_{\domain}(\g;g) = J_{\domain}(\g;\ud z^2)$.

Moreover, if $\varphi$ is a conformal map from $(\domain; \bpt, \ept)$ to $(\domain'; \bpt', \ept')$ and $\sigma = \log \abs{\varphi'}$, then~\eqref{eq:var_H} still holds. 
In particular,  
we have 
\begin{align}\label{eq:J_conformal_inv}
    J_{\domain'}(\varphi(\g);\ud z^2)=   J_{\domain}(\g;\ud z^2) ,
\end{align}
i.e., $J$ is invariant under $\varphi$. 
\end{lem}

By the Weyl-invariance, we will write  $J_\domain(\g)$, omitting the notation for the metric $g$
when it is Weyl-equivalent to the Euclidean metric $\ud z^2$. 

\begin{proof}[Proof of Lemma~\ref{lem:J_conformal_invariance}]
Note that~\eqref{eq:var_H} shows that the variation of the functional $\tilde {\mc H}_{\domain} (\g; g)$ 
under a Weyl-scaling of the metric is independent of the curve $\g$, so $J$ is Weyl-invariant.
We first prove the asserted  identity~\eqref{eq:var_H} when $\s \in C^\infty(\ad \domain)$ and $\partial \domain$ is smoooth.  
Let $\domain^L$ and $\domain^R$ 
be the two connected components of the complement $\domain \smallsetminus \g$.
First, applying Theorem~\ref{thm:PA} with $g_0 = \ud z^2$, we have 
\begin{align*}
& \log \detz \D_{\domain;0}  - \log \detz \D_{\domain;g} \\
= & \frac{1}{6\pi} \left[ \frac{1}{2} \int_\domain \abs{\nabla_0 \s}^2 \ud \vol_0 
+ \int_{\partial \domain} k_0 \s \ud \mathrm{l}_0 \right] + \frac{1}{4\pi} \int_{\partial \domain} \partial_{\nu_0} \s \ud \mathrm{l}_0, 
\end{align*}
since $K_0 \equiv 0$.  Second, applying Theorem~\ref{thm:PA} to 
$\domain^L$, 
which has two corners at $\bpt$ and $\ept$ both with an opening angle of $\pi/2$, we obtain 
\begin{align*}
&  \log \detz \D_{\domain^L;0}  - \log \detz \D_{\domain^L;g}\\
= & \frac{1}{6\pi} \left[ \frac{1}{2} \int_{\domain^L} \abs{\nabla_0 \s}^2 \ud \vol_0 
+ \int_{\partial \domain^L \smallsetminus \{\bpt, \ept\}} k_0 \s \ud \mathrm{l}_0 \right] 
+ \frac{1}{4\pi} \int_{\partial \domain^L \smallsetminus \{\bpt, \ept\}} \partial_{\nu_0} \s \ud \mathrm{l}_0
+ \frac{1}{8} (\s (\bpt) + \s(\ept)) , 
\end{align*}
and similarly for $\domain^R$. Hence, we see that in the difference
$\tilde {\mc H}_{\domain} (\g; g) - \tilde {\mc H}_{\domain} (\g; \ud z^2)$
all terms except the corner contributions~\eqref{eq:var_H} 
cancel, since $\s$ is continuous across $\g$.

Next, we consider the case where  $\sigma = \log |\varphi'|$, which is smooth on  $\domain \cup U_\bpt \cup U_\ept$. Let $\tilde \domain \subset \domain$ be a smooth domain such that $U_\bpt \cup U_\ept \subset \partial \tilde \domain$ and $\g \subset \tilde \domain$. Using the Brownian loop measure interpretation of $\log \detz \D$ from Proposition~\ref{prop:UV_cutoff}  and letting $\d \to 0$, we see that
\begin{align*}
\tilde {\mc H}_{\tilde \domain} (\g; g) - \tilde {\mc H}_{\domain} (\g; g) 
= \tilde {\mc H}_{\tilde \domain} (\g; \ud z^2) - \tilde {\mc H}_{\domain} (\g; \ud z^2) ,
\end{align*} 
where both sides also equal the total mass of Brownian loops in $\domain$ intersecting both $\domain \smallsetminus \tilde \domain$ and $\g$ (which is finite and conformally invariant).
Now,~\eqref{eq:var_H} with $\sigma = \log |\varphi'|$ follows from the identity~\eqref{eq:var_H} applied to $\tilde \domain$. Furthermore, the $\varphi$-invariance then follows:
\begin{align*}
J_{\domain'} (\varphi(\g); \ud z^2) = J_\domain (\g; g) = J_\domain (\g; \ud z^2) 
\end{align*}
shows the asserted identity~\eqref{eq:J_conformal_inv}.
\end{proof}

\begin{cor} \label{cor:M_poisson_kernel}
$\tilde \Hmin_\domain(\bpt, \ept ; \ud z^2) + \frac{1}{4} \log P_{\domain;\bpt,\ept} =: \cst \in \m R$ is a universal constant. 
\end{cor}
\begin{proof}
Let $\varphi$ be a conformal map between  
$(\domain; \bpt, \ept)$ and $(\domain'; \varphi(\bpt), \varphi(\ept))$. 
Then, for any $\g \in \mc X^{\infty} (\domain;\bpt,\ept)$, we have
\begin{align*}
\tilde {\mc H}_{\domain'} (\varphi(\g); \ud z^2) = \tilde{\mc H}_{\domain} (\g; e^{2\s (z)} \ud z^2) ,
\qquad \text{where} \quad \s (z) = \log \abs{\varphi'(z)} .
\end{align*} 
Therefore, we obtain from Lemma~\ref{lem:J_conformal_invariance} and~\eqref{eq:Poisson_def} that
\begin{align} \label{eq:conformal_restriction_M}
\begin{split}
 \tilde \Hmin_{\domain'}(\varphi(\bpt), \varphi(\ept) ; \ud z^2) -  \tilde\Hmin_\domain (\bpt, \ept ; \ud z^2) 
= \; & \tilde\Hmin_\domain (\bpt, \ept ; e^{2\s(z)} \ud z^2) - \tilde\Hmin_\domain (\bpt, \ept ; \ud z^2) \\
= \; &  \frac{1}{4} (\s(\bpt) + \s(\ept))
= \frac{1}{4} \log \abs{\varphi'(\bpt) \varphi'(\ept)} \\
= \; &  \frac{1}{4} \left( \log P_{\domain;\bpt,\ept}
- \log P_{\domain';\varphi(\bpt), \varphi(\ept)} \right) ,
\end{split}
\end{align}
which implies the claim.
\end{proof}

In Proposition~\ref{prop: IJmin}, we will show that 
$J_\domain$ coincides with the single-chord Loewner energy~\eqref{eq_LE}
on smooth chords, i.e., $J_\domain(\g) = I_\domain(\g)$ for all $\g \in \mc X^\infty(\domain; \bpt, \ept)$.
Assuming this fact, we now prove the main result of this section.

\begin{proof}[Proof of Theorem~\ref{thm:H_H}]
For $\np = 1$, the assertion follows immediately from 
the definition~\eqref{eq:def_J} of $J$,
Proposition~\ref{prop: IJmin}, Corollary~\ref{cor:M_poisson_kernel}, 
and the definition~\eqref{eq:def_H_single} of $\mc H_\domain$\textnormal{:} we have
\begin{align*}
\tilde {\mc H}_\domain (\g ; \ud z^2) 
= \frac{1}{12} J_{\domain}(\g) + \tilde {\mc M}_\domain (\bpt,\ept ; \ud z^2) 
= \frac{1}{12} I_\domain(\g) - \frac{1}{4} \log P_{\domain;\bpt, \ept} + \cst = \mc H_\domain(\g) + \cst. 
\end{align*}
The general case follows by induction on $\np \ge 2$.
We use the Euclidean metric below and omit it from the notation.
Let $\domain^L$ and $\domain^R$ be the two connected components of $\domain \smallsetminus \bigcup_i \g_i$ adjacent to the chord $\g_1$,
so that $\hat \domain_1 = \domain^L \cup \domain^R \cup \g_1$ (recall Figure~\ref{fig:NSLE}). 
Then, we obtain
\begin{align*}
\tilde {\mc H}_{\domain}(\ad \g ) & = \log\detz \D_{\domain} - \sum_C \log\detz \D_{C} \\
& = \log\detz \D_{\domain} - \bigg(\sum_{C \neq \domain^L, \domain^R} \log\detz \D_{C} + \log \detz \D_{\hat \domain_1} \bigg) \\
& \quad + \log \detz \D_{\hat \domain_1} - \log\detz \D_{\domain^L} - \log\detz \D_{\domain^R}  \\
& = \tilde {\mc H}_\domain(\g_2, \ldots, \g_\np) + \tilde {\mc H}_{\hat \domain_1}(\g_1) .
\end{align*}
Hence, $\tilde {\mc H}$ and $\mc H$ satisfy the same cascade relation  (cf. Lemma~\ref{lem:H_cascade}),
so it follows by induction that they are equal up to adding the constant $\np \cst$.
This proves Theorem~\ref{thm:H_H}. 
\end{proof}

The rest of this section is devoted to showing that $I_\domain = J_\domain$ (Proposition~\ref{prop: IJmin}).

\begin{remark}\label{rem:slit_sphere_identity}
When $(\domain;\bpt,\ept) = (\m C \smallsetminus \m R_+;0,\infty)$, we obtain from~\cite[Thm.\,7.3]{W2} that
\begin{align}\label{eq:I_J_slit_sphere}
I_{\m C \smallsetminus \m R_+} (\g) = J_{\m C \smallsetminus \m R_+; g} (\g) ,
\end{align}
where $g$ is Weyl-equivalent to the spherical metric (here, $\m C \smallsetminus \m R_+$ is unbounded).
Indeed, because for a chord $\g$ in $(\m C \smallsetminus \m R^+; 0 ,\infty)$, the chordal energy is the same as the loop energy $I^L$ of $\g \cup \m R_+$,~\cite[Thm.\,7.3]{W2} shows that 
\begin{align*}
I_{\m C \smallsetminus \m R_+} (\g) &= I^L (\g \cup \m R_+) =12 \tilde {\mc H}_{\Chat; g}(\g \cup \m R_+) -  12 \tilde {\mc H}_{\Chat; g} (\m R)  \\
& = 12 \tilde {\mc H}_{\m C \smallsetminus \m R_+; g} (\g) - 12 \tilde{\mc H}_{\m C \smallsetminus \m R_+; g} (\m R_-) ,
\end{align*}
which yields~\eqref{eq:I_J_slit_sphere} since
$I_{\m C \smallsetminus \m R_+}(\cdot) \geq 0$, so 
$\m R_-$ minimizes $\tilde{\mc H}_{\m C \smallsetminus \m R^+} (\cdot, g)$ in $\mc X^{\infty} (\m C \smallsetminus \m R^+; 0 ,\infty)$. 
\end{remark}

Since the Loewner energy $I$ is conformally invariant, 
if we could show the full conformal invariance of $J$, 
then from~\eqref{eq:I_J_slit_sphere} we would deduce that $I = J$ holds for all 
chords in any domain $\domain$. 
However, in order to establish this, we should understand in general 
how $\detz \D_\domain$ changes when varying opening angles of corners.
(For instance, when $\partial \domain$ is smooth near the marked points, 
$\domain^L$ makes right angles at those points, 
whereas when $\domain = \m C \smallsetminus \m R_+$, $\domain^L$ has no corners.)
Although the variation of $\detz \D_\domain$ while changing the angles is 
computable~\cite{R_inprep}, this is rather technical, so 
we restrict ourselves to Weyl-scalings only.

Now, recall that the chordal Loewner energy $I$ satisfies a conformal restriction formula, Lemma~\ref{lem: I_conformal_restriction}.
We next check that the same formula holds for $J$. 
Let $U \subset D$ be a curvilinear subdomain which 
agrees with $\domain$ in neighborhoods of $\bpt$ and $\ept$. 
Suppose that $\g \in \mc X^\infty(U; \bpt, \ept)$
and 
let $\varphi \colon U \to \domain$ be a conformal map 
fixing $\bpt$ and $\ept$. 

\begin{prop} \label{prop: J_conformal_restriction}
The functional $J$ satisfies the same conformal restriction formula as $I$\textnormal{:} 
\begin{align*}
J_{U}(\g) - J_{\domain}(\g)
= 3\log \abs{\varphi'(\bpt) \varphi'(\ept)} + 12 \, \mc B_{\domain}(\g, \domain \smallsetminus U).
\end{align*}
\end{prop}

\begin{proof}
According to~\cite[Prop.\,2.1]{Dub_couplings}, or Proposition~\ref{prop:UV_cutoff}, 
the Brownian loop measure $\mc B_{\domain} ( \g , \domain \smallsetminus U)$ can be written in terms of
$\tilde{\mc H}$ as
\begin{align} \label{eq:loop_measure_as_H}
\mc B_{\domain} ( \g, \domain \smallsetminus U) =  \tilde{\mc H}_{U} (\g; dz^2)  - \tilde{\mc H}_{\domain} (\g; dz^2).
\end{align}
This formula gives 
\begin{align*}
 J_{U} (\g) -  J_{D}(\g) = &\; 12 ( \tilde{\mc H}_{U} (\g; dz^2)  - \tilde{\mc H}_{\domain} (\g; dz^2)) + 12 (\tilde \Hmin_{\domain} (\bpt, \ept; dz^2) - \tilde \Hmin_{U} (\bpt, \ept; dz^2)) \\
= &\; 12 \, \mc B_{\domain} (\g , \domain \smallsetminus U) + 3 \log \abs{\varphi'(\bpt) \varphi'(\ept)} ,
\end{align*}
where the second equality follows from~\eqref{eq:loop_measure_as_H} and~\eqref{eq:conformal_restriction_M}. 
\end{proof}

\begin{cor} \label{cor:difference_of_IJ}
Let $\domain'$ be a curvilinear polygonal 
domain such that the boundaries
$\partial \domain$ and $\partial \domain'$ agree in neighborhoods of $\bpt$ and $\ept$. 
Suppose $\g \subset \domain \cap \domain'$. Then, we have
\begin{align}\label{eq:difference_I=J}
I_{\domain'}(\g) - I_{\domain}(\g) = J_{\domain'}(\g) - J_{\domain}(\g).
\end{align}
\end{cor}

\begin{proof}
The conformal restriction formula (Lemma~\ref{lem: I_conformal_restriction}) applied to a domain $U \subset D \cap D'$ containing $\g$ gives
\begin{align} \label{eq:general_I_restriction}
I_{\domain'}(\g) - I_{\domain}(\g) 
= &\; 3\log \abs{\varphi'(\bpt) \varphi'(\ept)} + 12 \, \mc B_{\domain}(\g, \domain \smallsetminus \domain') - 12 \, \mc B_{\domain}(\g, \domain' \smallsetminus \domain) ,
\end{align}
where $\varphi$ is a conformal map from $\domain'$ to $\domain$ fixing $\bpt$ and $\ept$.
The same argument 
(with Proposition~\ref{prop: J_conformal_restriction})
also shows that the identity~\eqref{eq:general_I_restriction} holds for $J$ when $\g \in \mc X^{\infty} (\domain;\bpt,\ept)$.
\end{proof}

Now we are ready to show that $I_{\domain} = J_{\domain}$,
using the conformal restriction formula and approximation of smooth chords by analytic chords.

\begin{prop} \label{prop: IJmin} 
The functional $J_{\domain}$ 
attains its infimum in $\mc X^{\infty}(\domain;\bpt,\ept)$ 
and its unique minimizer is the hyperbolic geodesic $\eta$.
In particular, we have $J_{\domain}(\g) = I_{\domain} (\g)$ for all smooth chords $\g \in \mc X^{\infty}(\domain;\bpt, \ept)$. 
\end{prop}

\begin{proof}
We assume without loss of generality (by Lemma~\ref{lem:J_conformal_invariance}) that $(\domain; \bpt, \ept) = (\m D; -1, 1)$.
We call $\g$ an \emph{analytic} curve in $\mc X^\infty : = \mc X^{\infty}(\m D; -1, 1)$ 
if there exists a neighborhood $U$  of $\g$ with smooth boundary which coincides with $\m D$ 
near $-1$ and $1$, and a conformal map $\varphi \colon \m D \to U$ such that $\varphi([-1,1])=\g$.
We denote the family of such  analytic curves by $\mc X^{\o}$. 
We first show that the segment $[-1,1]$ minimizes $J_\domain$ in $\mc X^{\o}$.
Corollary~\ref{cor:difference_of_IJ} gives
\begin{align} \label{eq:I_J_difference}
I_{\m D}(\g) = I_{\m D} (\g) - I_U (\g)= J_\domain (\g) - J_U (\g) = J_{\m D}(\g) - J_{\m D}([-1,1]),
\end{align}
using also the facts that 
$I_U(\g) = I_{\m D}([-1,1]) = 0$ and $J_{\m D}([-1,1]) = J_U (\g)$, 
thanks to the Weyl-invariance of $J$ (Lemma~\ref{lem:J_conformal_invariance}), 
equivalent to the conformal invariance here, for both $\domain$ and $U$ have smooth boundary.
Since $I_{\m D}(\g) \ge 0$, we have indeed 
\begin{align} \label{eq:J_compare_analytic}
J_{\m D} (\g) \ge J_{\m D}([-1,1]) \qquad \text{for all } \g \in \mc X^{\o}.
\end{align}

Now we claim that, 
for every $\g \in \mc X^\infty$ and an approximating sequence  $\g^k \in \mc X^\o$ for the  $C^3$-norm, 
as $k \to \infty$, the functionals 
$J_{\m D}(\g^k)$ and $I_{\m D}(\g^k)$ converge respectively to $J_{\m D}(\g)$ and $I_{\m D}(\g)$.  
(Since $\g$ is the image of $[-1,1]$ by an $C^{\infty}$-function, such an approximation always exists by Taylor expansion.)
To prove this claim, let $\domain^L_k$ and $\domain^R_k$ (resp.~$\domain^L$ and $\domain^R$) 
be the two connected components of $\m D \smallsetminus \g^k$ (resp.~$\m D \smallsetminus \g$) containing respectively $\ii$ and $-\ii$ on their boundary. 
Let $\m D_+ : = \{z \in \m D \;|\; \Im z >0\}$. Note that 
\begin{align*}
\log \detz \D_{\domain^L_k, \, \ud z^2} 
= \log \detz \D_{\m D_+, \, e^{2\s_k} \ud z^2} ,
\qquad \text{where} \quad \s_k = \log \big|(\varphi^k)'\big| ,
\end{align*}
and 
where $\varphi^k \colon \m D_+ \to D^{L}_k$ is a conformal map fixing the corners $-1$ and $1$. 
From Theorem~\ref{thm:PA}, 
it is not hard to check that 
$\log \detz \D_{\domain^L_k}$ converges to $\log \detz \D_{\domain^L}$, 
since all terms in the conformal anomaly formula~\eqref{eq: Generalized Polyakov-Alvarez} converge. 
We cautiously note that to properly deal with the corners, one extends $\varphi^k$ by Schwarz reflection across $S_+ := \{z \in S^1\;|\; \Im z \ge 0 \} \subset \partial \domain^L_k$,
and then, the  $C^{2,1-\vare}$-regularity of $\varphi^k$ up to the boundary follows 
from Kellogg's theorem \cite[Thm.\,II.4.3]{GM}.
Combining this with a similar analysis for $\log \detz \D_{\domain^R_k}$, we establish 
that $J_{\m D} (\g^k) \to J_{\m D} (\g)$ as $k \to \infty$.
Similarly, using~\eqref{eq:I_J_slit_sphere} and the conformal invariance of $I_{\m D}$, the above argument also shows the convergence of $I_{\m D} (\g^k)$ to $I_{\m D}(\g)$.

Now, Equation~\eqref{eq:J_compare_analytic} combined with
approximation of smooth chords by analytic chords  show that $[-1,1]$ also minimizes $J_{\m D}$ in $\mc X^{\infty}$.
Therefore, we conclude that $J_{\m D}([-1,1]) = 0$ and 
$J_{\m D}(\g) = I_{\m D}(\g)$ for all $\g \in \mc X^{\o}$ by~\eqref{eq:I_J_difference}.
The same approximation argument again allows us to finally conclude the equality $I_{\m D} = J_{\m D}$ for all $\g \in \mc X^{\infty}$.
\end{proof}

\subsection{UV-cutoff for Brownian loop measure}\label{subsec:UV}

In this section, we explain how Theorem~\ref{thm:main_H_det} can be directly related to Brownian loop measure.
This follows rather directly from the short-time expansion of the heat trace (Theorem~\ref{thm: heat trace multpilied}).
Brownian loop measure $\mu_{\domain;g}^{\mathrm{loop}}$ on $\domain$ with respect to a Riemannian metric $g$ is defined similarly as in Section~\ref{subsec:BLM}, by taking the diffusion generator to be $-\D_g$.

\begin{prop}[UV-cutoff of Brownian loop measure]\label{prop:UV_cutoff}
For a curvilinear polygonal domain $(\domain; g)$, the total mass of loops in $\domain$ with quadratic variation greater than $4 \d$ under Brownian loop measure has the expansion 
\begin{align*}
\frac{\vol_g (\domain)}{4\pi \d} - \frac{\mathrm l_g(\partial \domain)}{4\sqrt{\pi \d}} -  \log \detz \D_g - \bigg(\frac{1}{6} + \frac{1}{24}\sum_{j = 1}^m \Big(\b_j - 2 + \frac{1}{\b_j}\Big) \bigg) (\log \d + \boldsymbol{\upgamma}) + O(\d^{1/2} \log \d), 
\end{align*}
as $\d \to 0+$, where $\boldsymbol{\upgamma} \approx 0.5772$ is the Euler-Mascheroni constant.
\end{prop}

Our proof follows the same lines as in~\cite{ang2020brownian}.  
First, let us make a few remarks.
\begin{itemize}[itemsep=-2pt]
\item $4\d$ is the quadratic variation of a two-dimensional Brownian path run at speed $2$  until time $\d$ (as its generator is $-\D_g$ instead of $- \D_g/2$). 
We will write $[\ell]$ for the quadratic variation of $\ell$.
    
\item By the same proof, after replacing the term $1/6$ by $\chi(\domain)/6$, 
Proposition~\ref{prop:UV_cutoff} holds in a more general setup of (not necessarily planar nor simply connected) curvilinear domains on Riemannian surfaces.  
\end{itemize}

\begin{proof} [Proof of Proposition~\ref{prop:UV_cutoff}]  
Recall that Brownian loop measure on $\domain$ is defined as
\begin{align*}
\mu_{\domain;g}^{\mathrm{loop}}: = \int_0^{\infty} \frac{\ud t}{t} \int_\domain \m W_{z \to z}^t \dd \vol_g(z) ,
\end{align*}
where, in particular, $\m W_{z \to z}^t$ has total mass $p_t (z,z)$. 
Therefore, we have  
\begin{align*}
\mu^{\mathrm{loop}}_{\domain; g} \, \big(\{\ell\;|\; [\ell] \ge 4\d \}\big) = \int_\d^\infty t^{-1} \tr (e^{-t \D_g}) \dd t.
\end{align*}
As $t \to 0+$, $\tr (e^{-t \D_g})$ has the following expansion 
(see Theorem~\ref{thm: heat trace multpilied} with $\s \equiv 1$):
\begin{align} \label{eq:heat_trace}
\tr (e^{-t \D_{g}}) = \; \frac{\vol_g (\domain)}{4\pi t} - \frac{\mathrm l_g(\partial \domain)}{8\sqrt{\pi t}} + a_2 + O( t^{1/2} \log t )  
\end{align} 
where  
\begin{align*}
a_2 := a_2(g,1) 
=  \; & \frac{1}{12\pi} \int_\domain K_{g} \ud \vol_{g} + \frac{1}{12\pi} \int_{\partial \domain\smallsetminus \{z_1, \ldots, z_m\}}  k_{g} \ud \mathrm{l}_{g}  +  \frac{1}{24} \sum_{j=1}^m \left( \frac{1}{\b_j} - \b_j \right) \\
 = \; & \frac{1}{6} + \frac{1}{24}\sum_{j=1}^m \left( \frac{1}{\b_j} -2 +\b_j \right) ,
\end{align*}
and where
the last equality follows from the Gauss-Bonnet theorem:
\begin{align*}
\int_\domain K_{g} \ud \vol_{g} + \int_{\partial \domain\smallsetminus \{z_1, \ldots, z_m\}}  k_{g} \ud \mathrm{l}_{g}  + \sum_{j = 1}^m \pi (1 - \b_j) = 2\pi \chi (\domain) = 2\pi. 
\end{align*}
The spectral zeta function can be computed for $\Re(s) > 1$ as
\begin{align*}
\zeta_{\D_g} (s)  = \; & \frac{1}{\G(s)} \int_\d^{\infty} t^{s-1} \tr (e^{-t\D_g})\, \ud t + \frac{1}{\G(s)} \int_0^\d O( t^{s-1/2}\log t)\,\ud t \\
&
+ \frac{1}{\G(s)} \int_0^\d t^{s-2} \frac{\vol_g (\domain)}{4\pi } \ud t - \frac{1}{\G(s)} \int_0^\d t^{s-3/2}  \frac{\mathrm l_g(\partial \domain)}{8\sqrt{\pi }} \ud t + \frac{1}{\G(s)} \int_0^\d t^{s-1} a_2 \,\ud t \\
= \; & \frac{1}{\G(s)} \int_\d^{\infty} t^{s-1} \tr (e^{-t\D_g})\, \ud t + \frac{1}{\G(s)} \int_0^\d O( t^{s-1/2}\log t)\,\ud t \\
&
+ \frac{\vol_g (\domain)}{4\pi  } \frac{\d^{s-1}}{(s-1)\G(s)} -  \frac{\mathrm l_g(\partial \domain)}{8\sqrt{\pi } } \frac{\d^{s-1/2}}{(s-1/2)\G(s)} + a_2  \frac{\d^s}{\G(s+1)} \\
= : \; &   I_1 (s)+ I_2 (s)+ I_3(s) + I_4(s) + I_5(s) ,
\end{align*}
with obvious notation for the five terms, respectively. 
Taking the derivative of their analytic continuation (powers of $s$) at $s=0$ and using the formulas
\begin{align*}
\lim_{s \to 0} s \G (s) = 1 \qquad \text{and} \qquad 
\frac{\ud }{\ud s} \frac{1}{\G(s+1)} \bigg|_{s= 0} = \boldsymbol{\upgamma} ,
\end{align*}
we obtain
\begin{align*}
I_1'(0) & = \int_\d^{\infty} t^{-1} \tr (e^{-t\D_g})\, \ud t = \mu^{\mathrm{loop}}_{\domain; g} \,\big(\{\ell\;|\; [\ell] \ge 4\d \}\big); \\
I_2'(0) & = O(\sqrt \d \log \d); \qquad
I_3'(0)  =  - \frac{\vol_g (\domain)}{4\pi \d}; \qquad
I_4'(0)  =   \frac{\mathrm l_g(\partial \domain)}{4\sqrt{\pi \d}}; \\
I_5'(0) & =  a_2 (\log \d + \boldsymbol{\upgamma}).
\end{align*}
We now obtain the claimed expansion from the definition $\zeta_{\D_g}'(0) = - \log \detz \D_g$.
\end{proof}

Next, we consider a smooth Jordan domain $\domain$ with metric $g = \ud z^2$ (omitting the metric from the notation).
It follows from Theorem~\ref{thm:main_H_det} that the Loewner potential of $\ad \g \in \mc X^{\infty}(\domain; x_1, \ldots, x_{2\np})$ can be interpreted 
as the constant term in the expansion of the mass of Brownian loops touching $\ad \g$. 

\mainUV*

\begin{proof}
Applying Proposition~\ref{prop:UV_cutoff} to $\domain$ and to all of the $\np +1$ connected components $C$ of $\domain \smallsetminus \ad \g$, we obtain 
\begin{align*}
&\mu^{\mathrm{loop}}_{\domain} \,\big(\{\ell \;|\;[\ell] \ge 4\d \text{ and } \ell \cap \ad \g 
\neq \emptyset \}\big)  \\ & = \mu^{\mathrm{loop}}_{\domain} \,\big(\{\ell \;|\;[\ell] \ge 4\d \text{ and } \ell \cap \ad \g 
\neq \emptyset\} \big)  - \sum_C \mu^{\mathrm{loop}}_{C} \,\big(\{\ell \;|\;[\ell] \ge 4\d \text{ and } \ell \cap \ad \g 
\neq \emptyset\} \big)   \\
& = \frac{\mathrm l(\ad \g)}{2\sqrt{\pi \d}} -  (\mc H_\domain (\ad \g) + \np \cst)  - \Big(-\frac{\np}{6} - \frac{1}{24}\cdot 4 \np \cdot \frac{1}{2}\Big) (\log \d + \boldsymbol{\upgamma}) + O(\d^{1/2} \log \d) \\
& = \frac{\mathrm l(\ad \g)}{2\sqrt{\pi \d}} -  \mc H_\domain (\ad \g) - \np \cst  + \frac{\np}{4} (\log \d + \boldsymbol{\upgamma}) + O(\d^{1/2} \log \d) , \qquad \text{as } \; \d \to 0+ ,
\end{align*}
where the volume terms cancel out, 
the boundary length terms count each chord in $\ad \g$ twice, 
and there are $4 \np$ corners of opening angle $\pi/2$ contributing to the constant term.
\end{proof}

We can interpret the above result heuristically as
\begin{align*}
\mc H_\domain (\ad \g) \,``=" - \mu^{\mathrm{loop}}_{\domain} \,\big(\{\ell \;|\; \ell \cap \ad \g  \neq \emptyset \}\big) ,
\end{align*}
after renormalizing by taking out small loops and the divergent term 
proportional to the total length of $\ad \g$. 
The rest of the mass then only depends  on the number of chords. 
In fact, this interpretation is consistent with Definition~\ref{defn:def_H_multi} thanks to~\eqref{eq: m_alpha}:
\begin{align*}
\mc H_\domain (\ad \g) & = \sum_{j = 1}^\np \mc H_\domain (\g_j) + m_\domain (\ad \g) \\
& ``="  - \sum_{j = 1}^\np \mu^{\mathrm{loop}}_{\domain} \,\big(\{\ell \;|\; \ell \cap \g_j 
\neq \emptyset \}\big) + \mu_\domain^{\mathrm{loop}} \big((N(\ell) -1)\1\{N(\ell) \neq 0\}\big) 
\\
&  ``=" - \mu^{\mathrm{loop}}_{\domain} \,\big(\{\ell \;|\; \ell \cap \ad \g 
\neq \emptyset \}\big) , 
\end{align*}
where $N(\ell)$ is the number of chords in $\{\g_1, \ldots, \g_\np\}$ that $\ell$ intersects.

\appendix

\section{Refined estimate of SLE return probability}
\label{app:estimate}

The goal of this appendix is to prove the claim \ref{it:SLE_bound} of Proposition~\ref{prop:coming_back_bound}, 
given in Corollary~\ref{cor:return_probability}.

\begin{thm} \label{thm:constant_control}
Let $\k \in (0,4]$. 
Let $\tilde{S}_r = S_r + 1$ be a semicircle of radius $r \in (0,1)$ centered at $1 \in \m R$,
and let $\g^\k$ be $\SLE_\k$ in $(\m H;0,\infty)$. Then, we have
\begin{align*}
\m P^\k [\g^\k \cap \tilde{S}_r \neq \emptyset ] \leq c_\k'' 
r^{8/\k-1} ,
\qquad
\text{where}
\qquad
c_\k'' = \frac{\Gamma(12/\k)}{\Gamma(8/\k) \Gamma(4/\k+1)},
\end{align*}
and where $\Gamma (z) : = \int_0^\infty x^{z-1} e^{-x} \,\ud x$ is the Gamma function.
\end{thm}

Versions of this result have appeared, e.g., in~\cite[Thm.\,3.2]{AlbertsKozdron} and in~\cite[Thm.\,1.1]{FieldLawler}, but these theorems did not control the multiplicative $\k$-dependent constant in the estimate.
Because we must keep track of the behavior of the constant $c_\k''$ as $\k \to 0$, we present a sharper argument, using ideas of Greg Lawler~\cite{Greg_Minkowski}.

From Stirling's formula, we find that 
\begin{align} \label{eq: constant kappa primeprime}
\lim_{\k \to 0+} \k \log c_\k'' = C'' 
: = 12 \log 3- 8 \log 2
\approx 7,63817 . 
\end{align}

\begin{proof}[Proof of Theorem~\ref{thm:constant_control}]
The proof follows very closely~\cite[Sec.\,2]{Greg_Minkowski}. 
It is convenient to measure the distance of the $\SLE_\k$ curve to
the point $1 \in \m R$ in terms of the conformal radius. 
Because the point of interest is on the boundary of the domain $\m H$, 
we first define the reflected domain including the positive real line, 
\begin{align*}
\domain_t :=  H_t \cup \{ z^* \; | \; z \in H_t \} \cup \m R_+ ,
\end{align*}
where $H_t = \m H \smallsetminus \g^\k_{[0,t]}$. Then, for each $z \in \domain_t$, we set 
\begin{align*}
\Upsilon_t(z) := \frac{1}{4} {\rm crad}_{\domain_t}(z) ,
\end{align*}
where ${\rm crad}_{\domain_t}(z)$ denotes the conformal radius of $\domain_t$ seen from $z$.
We consider the process $\Upsilon = (\Upsilon_t)_{t \geq 0} := (\Upsilon_t(1))_{t \geq 0}$.
By Koebe $1/4$ theorem, we have $\Upsilon_t \leq {\rm dist} (1, \partial \domain_t)$, 
so upon taking $t \to \infty$, we see that $\Upsilon_\infty \leq {\rm dist} (1, \g^\k)$. 
Hence, we have 
\begin{align*}
\m P^\k [\g^\k \cap \tilde{S}_r \neq \emptyset ] = 
\m P^\k [{\rm dist} (1, \g^\k) \leq r ] \leq 
\m P^\k [\Upsilon_\infty \leq r ] ,
\end{align*}
and it suffices to show that
\begin{align} \label{eq:upsilon_estimate}
\m P^\k [\Upsilon_\infty \leq r ] \leq c_\k'' 
r^{8/\k-1} .
\end{align}

To estimate this, we tilt the $\SLE_\k$ measure by a suitable (local) martingale.
We consider the Loewner flow of $\g^\k$ parametrized as follows: $g_t \colon H_t \to \m H$ is the uniformizing map normalized at $\infty$ and satisfying
\begin{align*}
\frac{\ud g_t(z)}{\ud t}  = \frac{2/\kappa}{g_t(z) + B_t} , \qquad g_0(z) = z ,
\end{align*}
where $B$ is the standard Brownian motion\footnote{The sign of the Brownian motion here is chosen for convenience.  Note that the Loewner flow here is \emph{not} normalized as in Section~\ref{sec:prelim} but differs only by a time-reparametrization. The choice is made to be consistent with \cite{Greg_Minkowski} and to simplify the computations.}.

We set $X_t := g_t(1) + B_t$
and $O_t := g_t(0_+) + B_t$, where $g_t(0_+)$ is the image of the prime end at $0$ approached from the right. Then, we have $X_0 = 1$ and $O_0 = 0$, and
\begin{align*}
\ud X_t = \frac{2/\k}{X_t} \ud t + \ud B_t , \qquad
\ud O_t = \frac{2/\k}{O_t} \ud  t + \ud  B_t ,
\end{align*}
that is, both are Bessel processes of effective dimension $4/\k+1 \geq 2$
(so in particular, $X_t > 0$ and $O_t > 0$ for all $t > 0$ almost surely).
We also set 
\begin{align*}
Y_t := X_t - O_t \geq 0 \qquad \text{and} \qquad J_t := \frac{Y_t}{X_t} \in [0,1] .
\end{align*}
Then, it can be verified by a calculation that the following equations hold~\cite[Sec.\,2]{Greg_Minkowski}: 
\begin{align*}
\frac{\ud}{\ud t} Y_t =  - \frac{2}{\k} \frac{Y_t}{X_t^2 \, (1 - J_t)} , \qquad 
\frac{\ud}{\ud t} \Upsilon_t =  - \frac{2}{\k} \frac{\Upsilon_t J_t}{X_t^2 \, (1 - J_t)} ,
\end{align*}
and in particular,
\begin{align*}
M_t = \bigg( \frac{g_t'(1)}{X_t} \bigg)^{8/\k-1} , \qquad
\ud M_t = \left( 1 - \frac{8}{\k} \right) \frac{M_t}{X_t} \ud B_t ,
\end{align*}
is a (local) martingale with initial value $M_0 = 1$.
We consider the measure $\tilde{\m P}^\k$ obtained from $\m P^\k$ via tilting by $M$. Then, Girsanov's theorem gives a standard Brownian motion $\tilde{B}$ with respect to $\tilde{\m P}^\k$ so that the change of measure between $B$ and $\tilde{B}$ is given by
\begin{align*}
\ud B_t = \left( 1 - \frac{8}{\k} \right) \frac{\ud t}{X_t} +  \ud \tilde{B}_t .
\end{align*}
With this change of measure and the following time-change, we can analyze $\Upsilon_\infty$: we set
\begin{align*}
\sigma(t) := \inf\{ s \geq 0 \; | \; \Upsilon_s = e^{-2t/\k} \}
\end{align*}
and on the event $\{\sigma(t) < \infty\}$, we define the time-changed processes
\begin{align*}
\hat{\Upsilon}_t := \Upsilon_{\sigma(t)} = e^{-2t/\k} , \qquad
\hat{X}_t := X_{\sigma(t)} , \qquad 
\hat{J}_t := J_{\sigma(t)} , \qquad 
\hat{M}_t := M_{\sigma(t)} , 
\end{align*}
so that after a calculation~\cite[Sec.\,2]{Greg_Minkowski}, we have
\begin{align*}
\hat{M}_t = \big( e^{2t/\k} \hat{J}_t \big)^{8/\k-1} , \qquad
\ud \hat{J}_t = \left( \frac{4}{\k} - \frac{6}{\k} \hat{J}_t \right) \ud t + \sqrt{\hat{J}_t (1-\hat{J}_t)} \, \ud \hat{B}_t ,
\end{align*}
where $\hat{B}$ is another standard Brownian motion. 
Now, we have
\begin{align*}
\m P^\k [\Upsilon_\infty \leq e^{-2t/\k} ] 
= \; & \big( e^{-2t/\k} \big)^{8/\k-1} \, \m E^\k [ \hat{M}_t \hat{J}_t^{1-8/\k} \; \1{\{\Upsilon_\infty \leq e^{-2t/\k}\}} ] \\
= \; & \big( e^{-2t/\k} \big)^{8/\k-1} \, \tilde{\m E}^\k [ \hat{J}_t^{1-8/\k} ] ,
\end{align*}
so writing $r = e^{-2t/\k} \in (0,1)$, we see that the probability of interest is 
\begin{align*}
\m P^\k [\Upsilon_\infty \leq r] = r^{8/\k-1} \, \tilde{\m E}^\k [ \hat{J}_t^{1-8/\k} ] .
\end{align*}
Thus, to establish~\eqref{eq:upsilon_estimate} it remains to prove the estimate  
\begin{align} \label{eq:upsilon_estimate2}
\tilde{\m E}^\k [ \hat{J}_t^{1-8/\k} ]  \leq c_\k'' .
\end{align}
It is shown in~\cite{Greg_Minkowski} 
that the process $\hat{J}$ has the explicit invariant distribution
\begin{align*}
u \; \mapsto \; \frac{\Gamma(12/\k)}{\Gamma(8/\k) \Gamma(4/\k)} u^{8/\k-1} (1-u)^{4/\k-1} .
\end{align*}
Now, note that the process $\hat{J}$ starts at $\hat{J}_0 = 1$, and 
we have $\hat{J}_t \in [0,1]$ for all $t$ almost surely.
Let $\hat{J}_t^i$ be an independent process starting at the invariant distribution.
Using Doeblin's coalescing coupling, we couple $\hat{J}^i$ and $\hat{J}$ 
so that they evolve independently until they collide, and after that they evolve together.
Then, we have 
\begin{align*}
\hat{J}_t^i \leq \hat{J}_t
\quad \text{for all } t \geq 0 .
\end{align*}
Thus, using the fact that $1-8/\k < 0$, we see that
\begin{align*}
\tilde{\m E}^\k [ \hat{J}_t^{1-8/\k} ] \leq \tilde{\m E}^\k [ \big( \hat{J}_t^i \big)^{1-8/\k} ] = \; &
\frac{\Gamma(12/\k)}{\Gamma(8/\k) \Gamma(4/\k)} \int_0^1 u^{1-8/\k} u^{8/\k-1} (1-u)^{4/\k-1} \ud u \\ 
= \; &
\frac{\Gamma(12/\k)}{\Gamma(8/\k) \Gamma(4/\k+1)} = c_\k'' .
\end{align*}
This shows~\eqref{eq:upsilon_estimate2} and concludes the proof.
\end{proof}

Suppose $\domain$ has a smooth boundary.
If $A_1, A_2 \subset \partial \domain$ are two disjoint 
subsets of the boundary, 
the \emph{Brownian excursion measure} between $A_1, A_2$ in $\domain$ is
\begin{align*}
\mathcal{E}_\domain(A_1,A_2) := \int_{A_1} \int_{A_2} P_{\domain;\bpt,\ept} |\ud \bpt| |\ud \ept| .
\end{align*}
If $A_1, A_2$ are smooth disjoint chords in $D$, we also write $\mc E_D(A_1, A_2)$ for $\mc E_C(A_1,A_2)$, where $C$ is the connected component of $D \smallsetminus (A_1 \cup A_2)$ such that  $A_1, A_2 \subset \partial C$.

Because Brownian excursion measure is conformally invariant
(this can be checked using the conformal covariance of the Poisson excursion kernel, see, e.g.,~\cite[Prop.\,5.8]{Law05}), 
it is also well-defined in the case when $A_1, A_2 \subset \partial \domain$ are not smooth. 
The domain monotonicity of the Poisson kernel
(Corollary~\ref{cor:compare_Poisson}) implies that $\mathcal{E}_\domain$ has the same property:
if $U \subset \domain$ is a simply connected subdomain 
which agrees with $\domain$ in neighborhoods of $A_1$ and $A_2$, then $\mathcal{E}_U(A_1,A_2) \leq \mathcal{E}_\domain(A_1,A_2)$.

\begin{figure}
 \centering
 \includegraphics[width=0.9\textwidth]{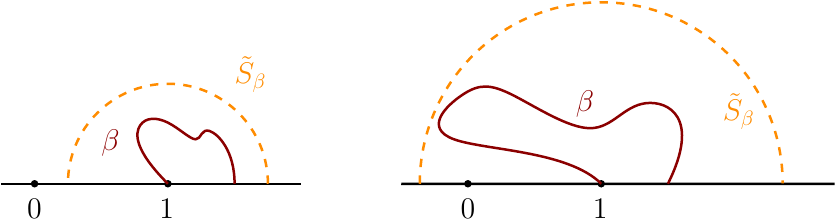}
 \caption{\label{fig:semicircle} 
 Illustration of the semicircle $\tilde{S}_\ARBC$ and the chord 
 $\ARBC$ in the proof of Lemma~\ref{lem:general_crosscut}. The left endpoint of the chord $\ARBC$
 is at $1 \in \m R$. 
 The left figure depicts the case when ${\rm diam}(\ARBC) < {\rm dist}(\ARBC,0) \leq 1$,
 and the right figure depicts the case when ${\rm diam}(\ARBC) \geq {\rm dist}(\ARBC,0)$.}
 \end{figure}

\begin{lem} [Variant of {\cite[Prop.\,3.1]{FieldLawler}}]
\label{lem:general_crosscut}
Let $\k \in (0,4]$. 
There exist constants 
\begin{align} \label{eq: constant kappa prime}
c_\k' \in (0,\infty) \qquad  \text{such that}
\qquad 
\lim_{\k \to 0+} \k \log c_\k' = C' \in (-\infty, \infty) ,
\end{align}
and the following holds. 
Let $\domain$ be a simply connected domain and $\bpt,\ept \in \partial \domain$ two distinct boundary points.
Let $\g$ be a chord from $\bpt$ to $\ept$ in $\domain$, and let $\ARBC$ be a chord
\textnormal(with arbitrary endpoints\textnormal)
in $\domain$ disjoint from $\g$.
Finally, let $\g^\k$ be $\SLE_\k$ in $(\domain;\bpt,\ept)$. Then, we have
\begin{align} \label{eq:SLE_crosscut_intersect}
\m P^\k [\g^\k \cap \ARBC \neq \emptyset ] \leq c_\k' 
\mathcal{E}_\domain(\ARBC,\g)^{8/\k-1} .
\end{align}
\end{lem}

Note that chords are also called crosscuts in the literature.

\begin{proof}
Without loss of generality (by conformal invariance and symmetry), 
we assume that $(\domain;\bpt,\ept) = (\m H;0,\infty)$ and 
$\ARBC$ has its endpoints on the positive real axis, with its left endpoint at $1 \in \m R$,
see Figure~\ref{fig:semicircle}.
Let $\tilde{S}_\ARBC = S_{R_\ARBC} + 1$ be a semicircle of radius ${\rm diam} (\ARBC)$ centered at $1$, thus encircling $\ARBC$. Then, Theorem~\ref{thm:constant_control} implies that
\begin{align} \label{eq:SLE_crosscut_intersect2}
\m P^\k [\g^\k \cap \ARBC \neq \emptyset ] \leq 
\m P^\k [\g^\k \cap \tilde{S}_\ARBC \neq \emptyset ] \leq 
c_\k'' \,
\bigg( \min \bigg\{ \frac{{\rm diam}(\ARBC)}{{\rm dist}(\ARBC,0)}, 1 \bigg\} \bigg)^{8/\k-1} ,
\end{align}
where the constant satisfies
(because $c_\k''$ is decreasing on $\k \in (0,4]$ and $c_4'' = 2$)
\begin{align*}
c_\k'' = \frac{\Gamma(12/\k)}{\Gamma(8/\k) \Gamma(4/\k+1)} \geq 2 ,
\end{align*}
and \eqref{eq:SLE_crosscut_intersect2} holds trivially when diam$(\ARBC) \ge {\rm dist}(\ARBC, 0)$.
Also, by Stirling's formula, we have
\begin{align*}
\lim_{\k \to 0+} \k \log c_\k'' = C'' := 2 \log \left(3^6/2^4  \right) \approx 7,63817.
\end{align*}

On the other hand, by~\cite[Cor.\,5.2]{FieldLawler} there exists a constant $c \in (0,\infty)$ such that
\begin{align*}
\mathcal{E}_{\m H}(\ARBC,\g) 
\geq \mathcal{E}_{\m H}(\ARBC,(-\infty,0)) \geq \frac{1}{c} \,
\min \bigg\{ \frac{{\rm diam}(\ARBC)}{{\rm dist}(\ARBC,0)}, 1 \bigg\} .
\end{align*}
Combining this with~\eqref{eq:SLE_crosscut_intersect2}, we conclude that 
\begin{align*}
\m P^\k [\g^\k \cap \ARBC \neq \emptyset ] \leq 
c_\k''  \,
\bigg( \min \bigg\{  \frac{{\rm diam}(\ARBC)}{{\rm dist}(\ARBC,0)}  , 1 \bigg\} \bigg)^{8/\k-1} 
\leq 
c_\k'\, \mathcal{E}_{\m H}(\ARBC,\g)^{8/\k-1} , 
\end{align*}
where $c_\k' := c_\k'' \, c^{8/\k-1}$, and in particular, we have 
\begin{align*} 
\lim_{\k \to 0+} \k \log c_\k' = C'' + 8 \log c =: C' \in (-\infty, \infty) ,
\end{align*}
so the constant satisfies the asserted property~\eqref{eq: constant kappa prime}.
\end{proof}

\begin{prop}[Variant of {\cite[Prop.\,3.4]{FieldLawler}}]
\label{prop:SLE_estimate}
Let $\k \in (0,4]$. 
There exist constants 
\begin{align} \label{eq: constant kappa}
c_\k \in (0,\infty) \qquad  \text{such that}
\qquad 
\lim_{\k \to 0+} \k \log c_\k = C \in (-\infty, \infty) ,
\end{align}
and the following holds. 
Let $K \subset \ad{\m H}$ be a hull such that 
$K \cap (\m H \smallsetminus {\m D}) = \{z\}$, 
and let $\g^\k$ be $\SLE_\k$ in $(\m H \smallsetminus K; z,\infty)$.
Then, for any $r \in (0,1/3)$, we have
\begin{align} \label{eq: SLE_intersect_semicircle}
\m P^\k [\g^\k \cap S_r \neq \emptyset ] \leq c_\k r^{8/\k-1} .
\end{align}
\end{prop}

The setup is illustrated in Figure~\ref{fig:SLE_estimate}.

\begin{figure}
 \centering
 \includegraphics[width=0.5\textwidth]{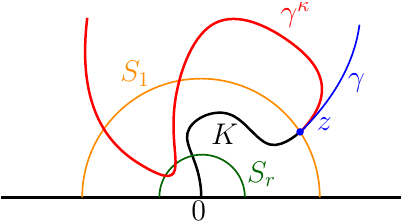}
 \caption{\label{fig:SLE_estimate} 
 Illustration of setup in the proof of Proposition~\ref{prop:SLE_estimate}. The curve $\g^\k$ is SLE$_\k$ in $(\m H \smallsetminus K; z, \infty)$, for which we estimate the probability to hit $S_r$, and $\g$ is an arbitrary simple curve from $z$ to $\infty$ in $\m H$ such that $\g$ intersects $S_1$ only at $z$.}
 \end{figure}

\begin{proof}
Write $\domain = {\m H} \smallsetminus K$ and 
$\domain \cap S_r = \cup_{j=1}^\infty \ARBC_j$ as a union of chords in $\domain$. 
Then, we have
\begin{align*}
\m P^\k [\g^\k \cap S_r \neq \emptyset ]  =  \; &
\m P^\k \bigg[\g^\k \cap \bigcup_{j=1}^\infty \ARBC_j \neq \emptyset \bigg] \\
\leq \; & c_\k' \,
\sum_{j=1}^\infty \mathcal{E}_\domain(\ARBC_j,\g)^{8/\k-1} 
\qquad  && \text{[by Lemma~\ref{lem:general_crosscut}]}
\\
\leq \; &  c_\k' \,
\Big( \sum_{j=1}^\infty \mathcal{E}_\domain(\ARBC_j,\g) \Big)^{8/\k-1} ,
\qquad  && \text{[by $L^p$-norm monotonicity, since $\k \leq 4$]} 
\end{align*}
for any simple curve $\g$ from $z$ to $\infty$ in $\m H$ such that
$\g \cap S_1 = \{z\}$, 
where the constant $c_\k' \in (0,\infty)$ satisfies~\eqref{eq: constant kappa prime}. 
It remains to estimate the sum of the Brownian excursion measures $\mathcal{E}_\domain(\ARBC_j,\g)$.
To this end, we note that 
\begin{align*}
\sum_{j=1}^\infty \mathcal{E}_\domain(\ARBC_j,\g) 
\leq \; & \sum_{j=1}^\infty \mathcal{E}_\domain(\ARBC_j,S_1) 
\qquad \qquad && \text{[by monotonicity of $\mathcal{E}$]} 
\\
\leq \; &
2 \mathcal{E}_\domain(S_r,S_1)
\qquad \qquad && \text{[by \cite[Lem.\,3.3]{FieldLawler}]} \\
\leq \; &
2 \mathcal{E}_{\m H}(S_r,S_1) .
\qquad \qquad && \text{[by domain monotonicity of $\mathcal{E}$]} 
\end{align*}
By~\cite[Ex.\,5.10,~Eq.\,(5.11)]{Law05}
there exists a constant $c \in (0,\infty)$ such that $
\mathcal{E}_{\m H} (S_r, S_1) \leq c r$, since the semi-annulus  bounded by $S_r$, $S_1$, $[r,1]$ and $[-1, -r]$ is conformally equivalent to the rectangle $[0, - \log r] \times [0,\pi]$, where $S_r$ and $S_1$ are both mapped to a side of length $\pi$.
Thus, we obtain
\begin{align*} 
\m P^\k [\g^\k \cap S_r \neq \emptyset ] \leq  c_\k' \,
\Big( \sum_{j=1}^\infty \mathcal{E}_\domain(\ARBC_j,\g) \Big)^{8/\k-1}
\leq c_\k \, r^{8/\k-1},
\end{align*}
where $c_\k := c_\k' (2c)^{8/\k-1}$.
The asserted property~\eqref{eq: constant kappa} follows from property~\eqref{eq: constant kappa prime}. 
\end{proof}

\begin{cor} \label{cor:return_probability}
 For each $r > 0$ and for any $M \in [0,\infty)$, there exists $R > r$ such that 
 \begin{align} \label{eq:SLE_bound}
 \limsup_{\k \to 0+} \k \log \m P^\k [\g^\k_{[\tau_R, \infty)} \cap S_r \neq \emptyset ] \le -M .
  \end{align}
\end{cor}

\begin{proof}
Proposition~\ref{prop:SLE_estimate} with 
$K = \g^\k_{[0,\tau_R]}$, 
combined with scale-invariance and the strong Markov property of $\SLE_\k$, gives the estimate
\begin{align} \label{eq: GregLaurieEstimate}
\m P^\k [\g^\k_{[\tau_R, \infty)}\cap S_r \neq \emptyset ] \leq c_\k 
\left( r/R \right)^{8/\k-1} ,
\end{align}
for any  $R > 3 r$, 
where the constant $c_\k$ satisfies~\eqref{eq: constant kappa}. 
We furthermore choose $R$ large enough such that 
$$\limsup_{\k \to 0+} \k \log c_\k + 8 \log (r/R) = C + 8 \log (r/R) \le - M$$
from \eqref{eq: constant kappa}. 
This proves~\eqref{eq:SLE_bound}.  
\end{proof}

\section{Polyakov-Alvarez conformal anomaly formula}
\label{app:PA}

The \emph{Polyakov-Alvarez conformal anomaly formula} gives the change of the determinant of 
the Laplacian under a Weyl-scaling.
The case of closed surfaces goes back to Polyakov~\cite{Pol81} 
and the case of compact surfaces with smooth boundary to Alvarez~\cite{Alv83}. 
Osgood,  Phillips \&~Sarnak give a straightforward derivation of 
this formula for compact surfaces without boundary in~\cite[Sec.\,1]{OPS}.
(We have not found an explicit derivation for the case with boundary, although it is considered well-known.)
When the surface has conical singularities, the classical derivation~\cite{Pol81, Alv83, OPS}
is not directly applicable. 
One of the difficulties is to derive rigorously the trace of the heat kernel 
multiplied by the variation of the $\log$-conformal factor (Theorem~\ref{thm: heat trace multpilied}).

For curvilinear polygonal domains, it is relatively straightforward to find 
the appropriate contribution to the heat trace 
from the corners, as has been observed in many works~\cite{Fed64, Kac66, BS88, AS94, LR16,AR18}.
However, the rigorous proof of the short-time expansion of the heat kernel  including the constant term was only
established very recently in~\cite{NRS},
which is then used to prove the trace expansion 
Theorem~\ref{thm: heat trace multpilied}. 
Finally, the Polyakov-Alvarez formula 
proved in~\cite{R_inprep} 
follows as we explain briefly below.
We use the notations from Section~\ref{sec:det}.

\thmPA*

Recall that $\log \detz (\D_{g}) = -\zeta_{\D_g}'(0)$. 
We compare $\zeta_{\D_0}'(0)$ and $\zeta_{\D_{g}}'(0)$ by a variational computation.
For this purpose, we define $g_u := e^{2u \s} g_0$ for $u \in [0,1]$, so $g = g_1$. 
Theorem~\ref{thm: heat trace multpilied} and Proposition~\ref{prop:variation_of_zeta} together give the variation $\partial_u \zeta_{\D_{g_u}}'(0)$,
which is integrated over $u \in [0,1]$ to prove Theorem~\ref{thm:PA}.

The first step of the proof of Theorem~\ref{thm:PA} is the following short-time expansion of the trace $\tr (\s e^{-t\D_{g_u}})$
of the heat kernel multiplied by the $\log$-conformal factor $\s$.

\begin{theorem} \label{thm: heat trace multpilied}
The operator $\s e^{-t \D_{g_u}}$ is trace class and we have
\begin{align} \label{eq:s_trace}
\tr (\s e^{-t \D_{g_u}}) = \; & \frac{a_0(g_u,\s)}{t} + \frac{a_1(g_u,\s)}{\sqrt{t}} + a_2(g_u,\s)
+ O( t^{1/2} \log t )  , \qquad t \to 0+ ,
\end{align} 
where
\begin{align*}
a_0(g_u,\s) = \; & \frac{1}{4\pi} \int_\domain \s \ud \vol_{g_u} , \\
a_1(g_u,\s) = \; & - \frac{1}{8\sqrt{\pi t}} \int_{\partial \domain\smallsetminus \{z_1, \ldots, z_m\}} \s \ud \mathrm{l}_{g_u}  , \\
a_2(g_u,\s) = \; & \frac{1}{12\pi} \int_\domain \s K_{g_u} \ud \vol_{g_u} + \frac{1}{12\pi} \int_{\partial \domain\smallsetminus \{z_1, \ldots, z_m\}}  \s k_{g_u} \ud \mathrm{l}_{g_u}\\
 \; &
+ \frac{1}{8\pi } \int_{\partial \domain\smallsetminus \{z_1, \ldots, z_m\}} \partial_{\nu_u} \s  \ud \mathrm{l}_{g_u} +  \frac{1}{24} \sum_{j=1}^m \left( \frac{1}{\b_j} - \b_j \right) \s (z_j) ,
\end{align*}
and $\partial_{\nu_u} \s$ is the outward normal derivative with respect to the metric $g_u$.
\end{theorem}

\begin{proof}[Proof idea]
Because $\s$ is bounded, $\s e^{-t \D_{g_u}}$ is trace class. To obtain the asserted expansion heuristically, 
one can approximate the heat kernel in the bulk by the explicit heat kernel of the whole space; 
near smooth boundary by the explicit heat kernel of the half-space; and near a corner by 
the Sommerfeld-Carslaw heat kernel formula as in~\cite{Kac66, AS94}, 
or alternatively, using the Kontorovich-Lebedev transform of the Green's function as in~\cite{Fed64, BS88}.
On a heuristic level,  the so-called \emph{Kac's locality principle}~\cite{Kac66}
(i.e., asymptotic locality of the heat kernel at short times) suggests that 
the error made in this approximation is negligible. 
(Intuitively, Brownian motion does not reach far from its starting point in a short time, thus only feels the local geometry of the domain.)
After calculating all of these different local contributions, 
one patches them together to obtain~\eqref{eq:s_trace}.
However, such a patchwork a priori works rigorously only when the local geometries coincide \emph{exactly}, for instance, for polygonal domains (when 
the edges are not curved). 
The rigorous treatment for 
curvilinear domains is rather subtle and technical, 
and we refer to the recent literature~\cite{AR18,NRS,R_inprep} for this.
\end{proof}

\begin{prop} \label{prop:variation_of_zeta}
We have
\begin{align} \label{eqn:variation_of_zeta}
\partial_u \zeta_{\D_{g_u}}'(0) :=
\partial_\vare \zeta_{\D_{g_{u + \varepsilon}}}'(0) \big|_{\varepsilon = 0} = \; & 2 a_2(g_u,\s) .
\end{align}
\end{prop}

\begin{proof}
Using $\D_{g_{u + \d u}} = e^{-2 (\d u) \s } \D_{g_u}$, we have
\begin{align*}
\d \tr(e^{-t \D_{g_u}}) & = \tr (\d e^{-t \D_{g_u}}) = \tr \left(-t (\d \D_{g_u}) e^{-t \D_{g_u}}\right) = 2 t \tr \left(\s \D_{g_u} e^{-t \D_{g_u}}\right) \d u \\
& =- 2  t  \frac{\ud}{\ud t} \tr \left(\s e^{-t \D_{g_u}}\right) \d u.
\end{align*}
For $\Re(s) > 1$, the variation of the zeta function is  
\begin{align*} 
\d \zeta_{\D_{g_u}}(s) 
= \; & \frac{1}{\G(s)} \int_0^{\infty} t^{s-1} \delta  \tr (e^{-t \D_{g_u}}) \ud t
= - \frac{2 \,\d u}{\G(s)} \int_0^{\infty} t^{s} \frac{\ud}{\ud t}\tr \Big( \s e^{-t \D_{g_u}}  \Big) \ud t \\
= \; &  \d u \,\frac{2s}{\G(s)} \int_0^{\infty}  t^{s-1}  \tr ( \s e^{-t \D_{g_u}}) \ud t  ,
\end{align*}
where we used integration by parts.
Splitting the integral into $\int_0^1$ and $\int_1^\infty$, we have
\begin{align} \label{eq: zeta variation two terms}
\frac{\d \zeta_{\D_{g_u}}(s) }{\d u}
= \; & \frac{2s}{\G(s)} \int_0^{1}  t^{s-1}  \tr ( \s e^{-t \D_{g_u}}) \ud t 
+ \frac{2s}{\G(s)} \int_1^{\infty}  t^{s-1}  \tr ( \s e^{-t \D_{g_u}}) \ud t ,
\end{align}
where the second term is holomorphic on $\{s \in \m C \; | \; \Re(s) > - 1\}$. 
Now, using Theorem~\ref{thm: heat trace multpilied}, 
the integral in the first term in~\eqref{eq: zeta variation two terms} can be written 
in the form
\begin{align*}
\int_0^{1}  t^{s-1}  \tr ( \s e^{-t \D_{g_u}}) \ud t 
= \; & \frac{a_0(g_u, \s)}{s-1} + \frac{a_1(g_u, \s) }{s-1/2} + \frac{a_2(g_u, \s)}{s} 
+\int_0^{1} O(t^{s-1/2}) \ud t   ,
\end{align*}
which shows that
\begin{align*}
\frac{\d \zeta_{\D_{g_u}}(0) }{\d u}
= \; & \frac{2s}{\G(s)} \bigg(
\int_0^{1}  t^{s-1}  \tr ( \s e^{-t \D_{g_u}}) \ud t +
\int_1^{\infty}  t^{s-1}  \tr ( \s e^{-t \D_{g_u}}) \ud t \bigg) \bigg|_{s = 0} = 0 ,
\end{align*}
because $2s/\G(s) = 2s^2 + O(s^3)$.
Therefore, as $|s| \to 0$, we can expand
\begin{align*}
\delta \big( \G(s) \zeta_{\D_{g_u}}(s) \big)
&= \G(s) \big( \delta \zeta_{\D_{g_u}}(0) + s \delta \zeta_{\D_{g_u}}'(0) + O(s^2) \big)
\\
&= \G(s) \big( s \delta \zeta_{\D_{g_u}}'(0) + O(s^2) \big) .
\end{align*}
Dividing by $s\G(s)$, we obtain
\begin{align*}
\d \zeta_{\D_{g_u}}'(0)
= \; & \frac{1}{s \G(s)} \delta \big( \G(s) \zeta_{\D_{g_u}}(s) \big) \Big|_{s = 0}\\
= \; & 
\d u\,\big( 2s + O(s^2) \big) 
\left( \int_0^{\infty}  t^{s-1}  \tr ( \s e^{-t \D_{g_u}}) \ud t \right) \bigg|_{s = 0} 
= 2a_2(g_u,\s) \,\d u,
\end{align*}
which is the sought variation of the derivative of the spectral zeta function.
\end{proof}

\begin{proof}[Proof of Theorem~\ref{thm:PA}]
By Proposition~\ref{prop:variation_of_zeta}, we have
\begin{align*}
\log \detz (\D_0) - \log \detz (\D_g) 
=  \int_0^1 \partial_u \zeta_{\D_{g_u}}'(0) \ud u 
= 2 \int_0^1 a_2(g_u,\s) \ud u ,
\end{align*} 
where the coefficient $a_2(g_u,\s)$ of interest is given in
Theorem~\ref{thm: heat trace multpilied}.  
To write this in different form, we use the following simple transformation rules with $g_u = e^{2u\s} g_0$,
\begin{align*}
\D_{g_u} & = e^{-2u\s} \D_0,  
\qquad \qquad \dd \vol_{g_u} = e^{2u\s} \dd \vol_0, 
\qquad \quad \;
K_{g_u}  = e^{-2u\s} (K_0 + u \D_0 \s),  \\
 \partial_{\nu_u} \s & = e^{-u\s} \partial_{\nu_0} \s 
\qquad \qquad \quad \, \dd \mathrm l_{g_u} = e^{u\s} \dd \mathrm l_0,
\qquad \qquad \quad 
k_{g_u} = e^{-u\s} (k_0 + u \partial_{\nu_0} \s),
\end{align*}
to obtain
\begin{align*}
a_2(g_u,\s) 
= & \; \frac{1}{12\pi} \bigg(\int_\domain  \s (K_0 + u\D_0 \s) \dd \vol_0
+ \int_{\partial \domain\smallsetminus \{z_1, \ldots, z_m\}}  \s (k_0  + u\partial_{\nu_0} \s) \dd \mathrm{l}_0 \bigg) \\
\; & + \frac{1}{8\pi}  \int_{\partial \domain\smallsetminus \{z_1, \ldots, z_m\}} \partial_{\nu_0} \s  \ud \mathrm{l}_0 
 + \frac{1}{24} \sum_{j=1}^m \left( \frac{1}{\b_j} - \b_j \right) \s (z_j) .
\end{align*}
We finally obtain the asserted formula after using Stokes' formula
\begin{align*}
\int_\domain  \s \D_0 \s \dd\vol_0  + \int_{\partial \domain\smallsetminus \{z_1, \ldots, z_m\}} \s \partial_{\nu_0} \s \dd \mathrm l_0 =  \int_\domain |\nabla_0 \s|^2 \dd \vol_0
\end{align*}
and integrating over $u \in[0,1]$.
\end{proof}

\newpage

\bibliographystyle{alpha}

\end{document}